\documentclass[12pt]{article}
\usepackage{graphicx}
\usepackage{amssymb}
\usepackage{amsmath}
\usepackage{setspace}
\usepackage{natbib}
\newcommand{\hyp}{\ensuremath{\operatorname{hyp}}}
\newcommand{\epi}{\ensuremath{\operatorname{epi}}}
\newcommand{\ind}{\ensuremath{\operatorname{\mathbf{1}}}}
\newcommand{\E}{\ensuremath{\operatorname{\mathbb{E}}}}
\newcommand{\EL}{\ensuremath{\operatorname{EL}}}
\newcommand{\HL}{\ensuremath{\operatorname{HL}}}
\newcommand{\SL}{\ensuremath{\operatorname{SL}}}
\newcommand{\HS}{\ensuremath{\operatorname{HS}}}
\newcommand{\HP}{\ensuremath{\operatorname{HP}}}

\newtheorem{theorem}{Theorem}
\newtheorem{definition}[theorem]{Definition}
\newtheorem{proposition}[theorem]{Proposition}
\newenvironment{proof}[1][Proof]{\noindent\textbf{#1.} }{\ \rule{0.5em}{0.5em}}

\usepackage{Sweave}
\begin{document}

\title{Local Half-Region Depth for Functional Data}

\author{
Claudio Agostinelli \\
Department of Mathematics \\
University of Trento, Italy \\
claudio.agostinelli@unitn.it \\
}

\maketitle

\begin{abstract}
Data depth proves successful in the analysis of multivariate data sets, in particular deriving an overall center and assigning ranks to the observed units. Two key features are: the directions of the ordering, from the center towards the outside, and the recognition of a unique center irrespective of the distribution being unimodal or multimodal. This behaviour is a consequence of the monotonicity of the ranks that decrease along any ray from the deepest point. Recently, a wider framework allowing identification of partial centers was suggested in \citet{agostinelli_local_2011}. The corresponding generalized depth functions, called \textit{local depth functions} are able to record local fluctuations and can be used in mode detection, identification of components in mixture models and in cluster analysis. 

Functional data \citep{ramsay_functional_2006} are become common nowadays. Recently, \citet{lopez-pintado_half-region_2011} has proposed the half-region depth suited for functional data and for high dimensional data. Here, following their work we propose a local version of this data depth, we study its theoretical properties and illustrate its behaviour with examples based on real data sets.

\end{abstract}

\section{Introduction}
\label{SecIntroduction}
Functional data, such as continuous trajectories of a process, high frequency time series and irregularly spaced time series, are  encountered in many fields due to increased sensitivity and power of recording of measuring devices and computing tools. Statistical analysis of functional data sets needs specific methods and some recent overviews are \citet{ramsay_functional_2006}, \citet{ferraty_nonparametric_2006}, \citet{gonzalez-manteiga_statistics_2007}, \citet{HorKok2012} and \citet{cuevas2014}. A general task in functional data analysis is to devise an ordering within a sample of curves, with properties similar to the classical (univariate) order statistic. Here, the tools provided by data depth can prove useful because they can be applied to general types of data and they are by no means confined to the standard multivariate data. The basic tool in data depth is the depth function which describes the degree of centralness of the points of the reference space according to the underlying probability distribution. Statistical depth functions are invariant to all (non singular) affine transformations, for symmetric distributions they reach the maximum value at the center of symmetry and they become negligible when the norm of the point tends to infinity. Another critical property is ray-monotonicity, that is, the statistical depth does not increase along any ray from the center. We refer to \citet{liu_notion_1990}, \citet{zuo_general_2000} for motivation and discussion of the previous properties. The classical notion of data depth has been extended to functional data and the different available implementations aim to describe the degree of centrality of curves with respect to an underlying probability distribution or a sample. Some desirable properties of functional depths are for example invariance under some class of transformations, vanishing when the norm of a curve tends to infinity, (semi-)continuity and consistency, see \citet{MosPol2012} for more details). Once a depth function has been defined, the main tools for data analysis are the depth values and the depth regions. The depth values measure the degree of centrality of the points of the space thus providing a ranking of multivariate observations. This ranking is very different from the usual ranking of the points on the real line because it puts the highest ranks on the points near the center and the ranks monotonically decrease moving from the center towards the outside. The depth regions enclose the points with centrality not less than a fixed value and their location, size and shape convey a lot of information about the underlying distribution. Note that a depth region can be reparameterized so as to depend on the probability content instead of the depth value. It then follows that the boundary of the depth regions can be interpreted as a quantile surface of the appropriate order \citep{serfling_generalized_2002}.

Statistical depth was firstly introduced for multivariate observations, and many different definitions of depth are currently used, see \citet{liu_data_2006} and the reference therein for a review. 

One well-known example is Tukey halfspace depth \citep{tukey_mathematics_1975}. We write $\HS_{\mathbf{u}}(a)$ for the closed halfspace $\left\{\mathbf{z} \in \mathbb{R}^p: \mathbf{u}^t \mathbf{z} \geq a \right\}$. Here $\mathbf{u}$ is a $p$-vector satisfying $\mathbf{u}^t \mathbf{u} = 1$. Note that $\HS_{\mathbf{u}}(a)$ is the positive side of the hyperplane $\HP_{\mathbf{u}}(a) = \left\{ \mathbf{z} \in \mathbb{R}^p : \mathbf{u}^t \mathbf{z} = a \right\}$. The negative side $\HS_{-\mathbf{u}}(a)$ of $\HP_{\mathbf{u}}(a)$ is similarly defined. Halfspace depth is the minimum probability of all halfspaces including $\mathbf{x}$.
\begin{definition}[Halfspace Depth, \citealt{tukey_mathematics_1975}]
The halfspace depth $d_{HS}(\mathbf{x};\mathbf{X})$ maps $\mathbf{x} \in \mathbb{R}^p$ to the minimum probability, according to the random vector $\mathbf{X}$, of all closed halfspaces including $\mathbf{x}$, that is
\begin{align*}
d_{HS}(\mathbf{x}; \mathbf{X}) & = \inf_{\mathbf{u}: \mathbf{u}^t \mathbf{u}=1} \Pr_{\mathbf{X}}(\HS_{\mathbf{u}}(\mathbf{u}^t \mathbf{x})) \\
& = \inf_{\mathbf{u}: \mathbf{u}^t \mathbf{u}=1} \Pr_{\mathbf{X}}( \mathbf{z} \in \mathbb{R}^p : \mathbf{u}^t \mathbf{z} \geq \mathbf{u}^t \mathbf{x}) \ .
\end{align*}
\end{definition}

Practical implementations of data depth paradigm is often computational demanding, this is the case, for example, of the above depth function \citep{agostinelli_local_2008}. This prevents a straightforward extension of depth functions invented for multivariate data to the context of high-dimensional and/or functional data. Recently, specific notions of depth for functional data have been introduced which can be adapted to high-dimensional data without a large computational burden. Two notable definitions are the band and the half-region depths \citep{lopez-pintado_concept_2009, lopez-pintado_half-region_2011}. Here, we concentrate on the half-region depth and we propose the local (modified) half-region depth and we study its properties in the infinite and finite dimensional cases. We provide an application in the field of environmental studies. In particular, we consider a wind speed data set obtained from the meteorological station situated in Col De la Roa, San Vito di Cadore, Belluno, Italy. The station records wind speeds regularly every 15 minutes, and we consider the period 2001-2004. After removing 42 incomplete days, we obtain the final wind speed data set composed of 1420 curves which can be interpreted as daily wind speed profiles measured at $4 \times 24 = 96$ time points. The analysis is reported in Section \ref{SecExamples}. Another application of the local modified half-region depth is available in \citet{agostinelli2015} where the analysis of the shape of macroseismic fields is performed by the local modified half-region depth. \citet{sguera2014} proposes another functional depth which is local-oriented and kernel-based version of the functional spatial depth \citet{ChakrabortyChaudhuri2014}.

The paper is organized as follows. Section \ref{SecFunctionalDataDepth} reports the definitions and main properties of the half-region depth and Section \ref{SecLocalFunctionalDepth} reviews the notion of local depth and introduces the local half-region depth. An illustration, with real data set, is presented in Section \ref{SecExamples}. Concluding remarks are reported in Section \ref{SecConclusions}. An Appendix includes further examples and an illustration of the \texttt{R} package \texttt{ldfun} which implements all the methods presented in this work. The package is available upon request to the author.

\section{Functional Data Depth}
\label{SecFunctionalDataDepth}
In this section we introduce the notation for functional data sets and we discuss the half-region depth. A typical model for functional data is a function $y=y(t)$, $t \in T$, belonging to the space $C(T)$ of real continuous functions on some compact interval $T \subset \mathbb{R}$. A functional data set is a collection of functions $\mathbf{y}_n = \{ y_i \in C(T): i=1, \cdots, n\}$. The graph of a function $y$ is the subset of the space $G(y) = \{(t,y(t)): t \in T \}$.

The hypograph ($\hyp$) and the epigraph ($\epi$) of a function $y \in C(T)$ are defined as follows
\begin{align*}
\hyp(y) & = \{ (t,z): t \in T, z \le y(t) \} \ , \\
\epi(y) & = \{ (t,z): t \in T, z \ge y(t) \} \ . 
\end{align*}
The proportions of graphs that belong to the hypograph (epigraph) of a function $y$ are
\begin{align*}
R_{hyp}(y; \mathbf{y}_n) & = \frac{1}{n} \sum_{i=1}^n \ind(G(y_i) \subset \hyp(y)) = \frac{1}{n} \sum_{i=1}^n \ind(y_i(t) \le y(t), \forall t \in T) \ , \\
R_{epi}(y; \mathbf{y}_n) & = \frac{1}{n} \sum_{i=1}^n \ind(G(y_i) \subset \epi(y))  = \frac{1}{n} \sum_{i=1}^n \ind(y_i(t) \ge y(t), \forall t \in T) \ .
\end{align*}
These definitions can be extended in a straightforward way to a stochastic process $Y=Y(\omega,t)$ with trajectories in $C(T)$. In this case, we define $R_{hyp}(y; Y) = P(G(Y) \subset \hyp(y))$ and $R_{epi}(y; Y) = P(G(Y) \subset \epi(y))$. A simple notion of functional depth is the minimum probability of a trajectory of the random process $Y$ to belong in the hypograph or in the epigraph of a given curve $y$.
\begin{definition}[Half-Region Depth, \citealt{lopez-pintado_half-region_2011}]
 For a functional data set $\mathbf{y}_n$, the half-region depth of a curve $y$ is
\begin{equation*}
d_{HR}(y; \mathbf{y}_n) = \min(R_{hyp}(y; \mathbf{y}_n), R_{epi}(y; \mathbf{y}_n)) \ .
\end{equation*}
The population version is $d_{HR}(y; Y) = \min(R_{hyp}(y; Y), R_{epi}(y; Y))$.
\end{definition}
The notions of hypograph and epigraph can be adapted to finite-dimensional data. Parallel coordinates \citep[see][]{inselberg_plane_1985, unwin_multivariate_2006} are a convenient tool to visualize a set of points in $\mathbb{R}^p$. The Cartesian orthogonal axes become parallel and equally spaced in parallel coordinates; thus, points with dimension larger than three can be easily represented. Observations in $\mathbb{R}^p$ can be seen as real functions defined on the set of indexes $(1, 2, \cdots, p)$ and expressed as $\mathbf{x}=(x(1), x(2), \cdots, x(p))$. The hypograph and epigraph of a point $\mathbf{x} \in \mathbb{R}^p$ can be expressed respectively as $\hyp(\mathbf{x}) = \{(k, z) \in (1, 2, \cdots, p) \times \mathbb{R} : z \le x(k)\}$ and $\epi(\mathbf{x}) = \{(k, z) \in (1, 2, \cdots, p) \times \mathbb{R} : z \ge x(k)\}$. Let $\mathbf{e}_j$ be a $p$-vector with $1$ at the $j$-th coordinate and $0$ otherwise, $j=1, \cdots, p$. The half-region depth in Cartesian coordinates on $\mathbb{R}^p$ for the random vector $\mathbf{X}$ at the point $\mathbf{x}$ is given by
\begin{equation*}
d_{HR}(\mathbf{x}; \mathbf{X}) = \min \left( \Pr_{\mathbf{X}} \left( \bigcap_{j=1}^p \HS_{e_j}(x(j)) \right), \Pr_{\mathbf{X}} \left( \bigcap_{j=1}^p \HS_{-e_j}(x(j)) \right) \right) \ .
\end{equation*}

The properties of half-region depth are studied in \citet{lopez-pintado_half-region_2011}. In the finite dimensional case, half-region depth is not affine invariant, however it is invariant with respect to translations and some types of dilation, that is, linear transformations operated by positive (or negative) definite diagonal matrices. For $p=1$ half-region depth is equivalent to the halfspace depth. It vanishes as the norm of $\mathbf{x}$ increases. The empirical version uniformly almost surely converges to the corresponding population version, the empirical maximizer is a consistent estimator of the unique maximizer of the population version. In the general case half-region depth is invariant under the transformations $aY+b$ where $a(t)$ is either positive or negative in $T$. Similar results, as for the finite dimensional case, hold for the asymptotic behaviour.

A modified version of half-region depth which is less restrictive than the definition above, is proposed in \citet{lopez-pintado_half-region_2011}. This modified version can be used for the analysis of irregular (non-smooth) curves with many crossing points. We denote the superior ($\EL$) and the inferior ($\HL$) lengths as
\begin{align*}
\EL(y; Y) & = \frac{1}{\lambda(T)} \E \left( \lambda(t \in T: y(t) \le Y(t)) \right) \\
\HL(y; Y) & = \frac{1}{\lambda(T)} \E \left( \lambda(t \in T: y(t) \ge Y(t)) \right)
\end{align*}
where $\lambda$ stands for the Lebesgue measure on $\mathbb{R}$. $\EL(y; Y)$ can be interpreted as the ``proportion of time'' the stochastic process $Y$ is greater than the curve $y$, similarly for $\HL(y; Y)$. The modified half-region depth at $y$ is:
\begin{equation*}
d_{MHR}(y; Y) = \min(\EL(y;Y), \HL(y;Y)) \ .
\end{equation*}
Let $\mathbf{y}_n=(y_1, \cdots, y_n)$ be a set of curves from the stochastic process $Y$. The sample version of this depth is
\begin{equation*}
d_{MHR}(y; \mathbf{y}_n) = \min(\EL(y;\mathbf{y}_n), \HL_n(y;\mathbf{y}_n)),
\end{equation*}
where
\begin{align*}
\EL(y;\mathbf{y}_n) & = \frac{1}{n \lambda(T)} \sum_{i=1}^n \lambda\left(t \in T: y(t) \le y_i(t) \right) \\
\HL(y;\mathbf{y}_n) & = \frac{1}{n \lambda(T)} \sum_{i=1}^n \lambda\left(t \in T: y(t) \ge y_i(t) \right) \ .
\end{align*}
are the sample means of the proportions of the lengths respectively in the epigraph and hypograph.

\section{Local Functional Depth}
\label{SecLocalFunctionalDepth}
At the beginning of data depth it was almost a postulate that depth ranks could single out just one center of a distribution, corresponding to the maximizer of the ranks, whatever the shape of the distribution, \textit{unimodal or multimodal}. Current developments are showing that certain modified depth functions can indeed account for multimodal data, having multiple centers. These generalized definitions, called \textit{local} depth, measure centrality conditional on a neighbourhood of each point of the space and provide a tool that is sensitive to local features of the data, while retaining most features of regular depth functions. In most cases, when the neighbourhood radius tends to infinity, the usual (global) depth is recovered. More precisely, the monotonicity along rays from the center of the distribution prevents the depth function to describe local features of a distribution, like the modes (e.g. \citet[p. 461]{zuo_general_2000}, see also \citet{rafalin_proximity_2005}). However, when the distribution is unimodal, local depth behaves very similarly to a regular depth and provides a similar ranking. For halfspace depth, halfspaces are replaced by infinite slabs with finite width. When the threshold (slab width) tends to infinity, the ordinary definition is recovered. 

In more details, for $\tau>0$, the closed slab $\SL_{\mathbf{u}}(a,a+\tau)$ is the intersection of the positive side of $\HP_{\mathbf{u}}(a)$ and the negative side of $\HP_{\mathbf{u}}(a+\tau)$
\begin{align*}
\SL_{\mathbf{u}}(a,a+\tau) & = \HS_{\mathbf{u}}(a) \cap \HS_{-\mathbf{u}}(a+\tau) \\
& = \left\{ \mathbf{z} \in \mathbb{R}^{p}: a \leq \mathbf{u}^t \mathbf{z} \leq a+\tau \right\} \ .
\end{align*}

\begin{definition}[Local Halfspace Depth, \citealt{agostinelli_local_2011}]
For any $\tau>0$, local halfspace depth is the minimum probability, according to the random vector $\mathbf{X}$, of the closed slab $\SL_{\mathbf{u}}(\mathbf{u}^{t} \mathbf{x}, \mathbf{u}^{t} \mathbf{x} + \tau)$, that is
\begin{align*}
ld_{HS}(\mathbf{x}; \mathbf{X}, \tau) & = \inf_{\mathbf{u}: \mathbf{u}^t \mathbf{u}=1} \Pr_{\mathbf{X}}(\SL_{\mathbf{u}}(\mathbf{u}^t \mathbf{x}, \mathbf{u}^t \mathbf{x} + \tau)) \\
& = \inf_{\mathbf{u}: \mathbf{u}^t \mathbf{u}=1} \Pr_{\mathbf{X}}( \mathbf{z} \in \mathbb{R}^p : \mathbf{u}^t \mathbf{x} \leq \mathbf{u}^t \mathbf{z} \leq \mathbf{u}^t \mathbf{x} + \tau) \ .
\end{align*}
\end{definition}

\subsection{Local half-region depth}

A local version of half-region depth is obtained by imposing a restriction on the size of the hypograph and epigraph under consideration. 
%% local half-region depth
For a given non negative function $\tau$, the closed negative slab $\hyp(y; \tau)$ is the intersection between $\hyp(y)$ and $\epi(y-\tau)$, that is
\begin{equation*}
\hyp(y; \tau) = \hyp(y) \cap \epi(y-\tau) = \{ (t,z): t \in T, y(t) - \tau(t) \le z \le y(t) \}
\end{equation*}
and, in the similar way for the closed positive slab $\epi(y; \tau)$. The local half-region depth is defined as the minimum probability of these two sets.
\begin{definition}[Local Half-Region Depth]
For a stochastic process $Y$, the local half-region depth for the curve $y$ is defined as
\begin{equation*}
ld_{HR}(y; Y, \tau) = \min(P(G(Y) \subset \hyp(y;\tau)), P(G(Y) \subset \epi(y;\tau))) \ .
\end{equation*}
For a functional data set $\mathbf{y}_n$, the empirical version is
\begin{equation*}
ld_{HR}(y; \mathbf{y}_n, \tau) = \min\left( \frac{1}{n} \sum_{i=1}^n \ind(G(y_i) \subset \hyp(y;\tau)), \frac{1}{n} \sum_{i=1}^n \ind(G(y_i) \subset \epi(y;\tau))\right) \ .
\end{equation*}
\end{definition}
In the finite-dimensional case and Cartesian coordinates, this leads to consider slabs instead of halfspaces, each slabs has its own dimension given by $\tau(j)$ ($j=1, \cdots, p$)
\begin{equation*}
ld_{HR}(\mathbf{x}; \mathbf{X}, \boldsymbol{\tau}) = \min \left( \Pr_{\mathbf{X}} \left( \bigcap_{j=1}^p \SL_{-e_j}(x(j)-\tau(j),x(j)) \right), \Pr_{\mathbf{X}} \left( \bigcap_{j=1}^p \SL_{e_j}(x(j),x(j)+\tau(j)) \right) \right) \ .
\end{equation*}
A different formulation is possible using an Inclusion-Exclusion formula, in fact, the probability of the hyper cube can be rewritten in terms of the distribution function as follows
\begin{align} \label{EquLHRDistributionFunction}
\Pr_{\mathbf{X}} \left( \bigcap_{j=1}^p \SL_{-e_j}(x(j)-\tau(j),x(j)) \right) & = F_{\mathbf{X}}(\mathbf{x}) + \sum_{k=1}^p (-1)^k \sum_{I \subseteq \{1, \cdots, p\}, |I|=k} F_{\mathbf{X}}((\mathbf{x} - \boldsymbol{\tau}_I)^-) \\
\Pr_{\mathbf{X}} \left( \bigcap_{j=1}^p \SL_{e_j}(x(j),x(j)+\tau(j)) \right) & = F_{\mathbf{X}}(\mathbf{x}+\tau(j)) + \sum_{k=1}^p (-1)^k \sum_{I \subseteq \{1, \cdots, p\}, |I|=k} F_{\mathbf{X}}((\mathbf{x} + \boldsymbol{\tau}_I)^-) \nonumber
\end{align}
where $\boldsymbol{\tau}_I$ is a $p$-vector with the elements indexed by $I$ equal to the corresponding element in $\boldsymbol{\tau}$ and $0$ for all the others, and $|\cdot|$ is the cardinality of the set.

%%%%%\subsubsection{Properties of the local half-region depth in the infinite dimensional case}
Properties of the local half-region depth applied to functional data are now studied while Section \ref{sec:plhrdfdc}, available in the Appendix, provides results for the finite dimensional case. Let $\mathbf{y}_n = (y_1, \cdots, y_n)$ a functional data set from a stochastic process $Y$ in $C(I)$ with distribution function $P$. Assume that the stochastic process $Y$ is tight, i.e.,   
\begin{equation*}
\Pr (\| Y \|_\infty \ge M) \rightarrow 0 \ , \qquad M \rightarrow \infty \ .
\end{equation*}

The local half-region depth satisfies a linear invariance property provided that the threshold function $\tau$ is transformed appropriately, i.e., consider $a$ and $b$ functions in $C(I)$, where $a(t) > 0$ or $a(t) < 0$ for every $t \in I$, then
\begin{equation*}
ld_{HR}(a y + b; a Y + b, |a| \tau) = ld_{HR}(y; Y, \tau) \ .
\end{equation*}
In the next proposition we study the behaviour of the local half-region depth according to the behaviour of the $\tau$ function.
\begin{proposition}
\label{PropLocalHalfRegionTauInfiniteDimension}
Let $ld_{HR}(\cdot, Y, \tau)$ be the local half-region depth. Then, for any given function $y$
\begin{enumerate}
\item[i.] if $\tau_1 < \tau_2$, $ld_{HR}(y; Y, \tau_1) \leq ld_{HR}(y; Y, \tau_2)$;

\item[ii.] for any function $\tau$ such that $\tau \ge 0$, $P(G(Y)=g(y)) \leq ld_{HR}(y; Y, \tau) \leq d_{HR}(y; Y)$;

\item[iii.] $\lim_{\min_{t \in I} \tau(t) \rightarrow \infty} ld_{HR}(y; Y, \tau) = d_{HR}(y; Y)$;

\item[iv.] $\lim_{\| \tau \|_\infty \rightarrow 0^{+}} ld_{HR}(y; Y, \tau) = P(G(Y)=g(y))$.
\end{enumerate}
\end{proposition}
Proofs are reported in Section \ref{sec:proofs} of the Appendix. In the analysis of a functional data set, the function $\tau$ should be determined. In most cases $\tau$ can be the constant function which takes value equal to the quantile of the empirical distribution of the distance between all pairs of curves. In practice, the use of sup norm to measure the distance and a quantile order in the interval $5\%$-$30\%$ proved to be effective in most situations.

The local half-region depth of a function converges to zero when its norm tends to infinity.
\begin{proposition} \label{PropLocalHalfRegionZero}
The local half-region depths  $ld_{HR}(y; Y, \tau)$ and $ld_{HR}(y; \mathbf{y}_n, \tau)$, for all function $\tau$, always greater than zero, satisfy
\begin{align*}
\sup_{\| y \|_\infty \ge M} ld_{HR}(y; Y, \tau) \rightarrow 0 \ , \qquad \text{when} \ M \rightarrow \infty \ , \\
\sup_{\| y \|_\infty \ge M} ld_{HR}(y; \mathbf{y}_n, \tau) \rightarrow 0 \ , \qquad \text{when} \ M \rightarrow \infty \ .
\end{align*}    
\end{proposition}
It is a continuous function when the marginal distributions of the random process $Y$ are absolutely continuous.  
\begin{proposition}
\label{PropSemiContinuousFunctional}
Let $\tau$ be a predefined function then $ld_{HR}(\cdot; Y, \tau)$ is an upper semicontinuous functional. Moreover, if $P$ has absolutely continuous marginals, then $ld_{HR}(\cdot; Y, \tau)$ is continuous.
\end{proposition}  
In the next theorem we establish the strong consistency of the sample local half-region depth.
\begin{theorem}
\label{TeoConsistency}
$ld_{HR}(\cdot; \mathbf{Y}_n, \tau)$ is strongly consistent, 
\begin{equation*}
ld_{HR}(y; \mathbf{Y}_n, \tau) \stackrel{a.s.}{\longrightarrow} ld_{HR}(y; Y, \tau)
\end{equation*}  
\end{theorem}  
Hereafter, we establish the uniform consistency of the sample half-region depth and the strong consistency of its global maximizer using the approach based on empirical processes \citep[see for instance][]{kosorok_2008}. We recall that a subset $E$ of $C(I)$ is equicontinuous if for each $\nu > 0$ there exists $\delta(\nu) > 0$ such that for every $y \in E$ and for every $t, s \in I$, if $|t-s| < \delta(\nu)$ then $|y(t) - y(s)| < \nu$. 
\begin{theorem} \label{TeoEquiContinuousFunctional}
Let $E$ be an equicontinuous subset of $C(I)$ and assume that the stochastic process $Y$ satisfies the following condition: (i) for a given $\varepsilon > 0$, there exists $\nu(\varepsilon) > 0$ such that for every pair of functions $z_i, z_j \in C(I)$, if $\| z_i - z_j \|_\infty \leq \nu(\varepsilon)$ then $\Pr( z_j \leq Y \leq z_i) \leq \varepsilon$ and (ii) given a non negative function $\tau \in E$ such that $\inf_{t \in I} \tau(t) \geq \underline{\tau} > 2 \nu(\varepsilon) > 0$, then $ld_{HR}(\cdot; \mathbf{Y}_n, \tau)$ is strongly uniform consistent, i.e.,
\begin{equation*}
\sup_{y \in E} |ld_{HR}(y; \mathbf{Y}_n, \tau) - ld_{HR}(y; Y, \tau)| \stackrel{a.s.}{\longrightarrow} 0 \ .
\end{equation*}  
\end{theorem}
The following theorem proves the uniform convergence of the global maximizer of the sample local half-region depth.
\begin{theorem} \label{TeoMaximizerFunctional}
Under the assumptions stated in the previous theorem, if $ld_{HR}(y; Y, \tau)$ is uniquely maximized at $\hat{y} \in E$ and $\hat{y}_n$ is a sequence of functions in $E$ with $ld_{HR}(\hat{y}_n; \mathbf{y}_n, \tau) = \sup_{y \in E} ld_{HR}(y; \mathbf{y}_n, \tau)$ then
\begin{equation*}
\hat{y}_n \stackrel{a.s.}{\longrightarrow} \hat{y} \ ,    \qquad \text{when} \ n \rightarrow \infty \ .
\end{equation*}
\end{theorem}

\subsection{Local modified half-region depth}

To introduce a modified version of local half-region depth, the expected proportions of time a process (or a function) stays in the upper and lower slabs need be defined: 
\begin{align*}
\EL(y; Y, \tau) & = \frac{1}{\lambda(T)} \E \left( \lambda(t \in T: y(t) \le Y(t)) \le y(t)+\tau(t) \right) \\
& \times \ind \left( y(t) - \tau \le Y(t) \le y(t) + \tau(t); \forall t \in T \right) \\
\HL(y; Y) & = \frac{1}{\lambda(T)} \E \left( \lambda(t \in T: y(t) - \tau(t) \le Y(t) \le y(t)) \right) \\
& \times \ind \left( y(t) - \tau(t) \le Y(t) \le y(t) + \tau; \forall t \in T \right)
\end{align*}
Then, the local modified half-region depth at $y$ is:
\begin{equation*}
ld_{MHR}(y; Y, \tau) = \min(\EL(y;Y,\tau), \HL(y;Y,\tau)) \ .
\end{equation*}
The sample version of this local depth is
\begin{equation*}
ld_{MHR}(y; \mathbf{y}_n, \tau) = \min(\EL(y;\mathbf{y}_n,\tau), \HL_n(y;\mathbf{y}_n,\tau)) \ ,
\end{equation*}
where
\begin{align*}
\EL(y;\mathbf{y}_n,\tau) & = \frac{1}{n \lambda(T)} \sum_{i=1}^n \lambda \left(t \in T: y(t) \le y_i(t) \le y(t) + \tau(t) \right) \\
& \times \ind \left( y(t) - \tau(t) \le y_i(t) \le y(t) + \tau; \forall t \in T \right) \\
\HL(y;\mathbf{y}_n,\tau) & = \frac{1}{n \lambda(T)} \sum_{i=1}^n \lambda\left(t \in T: y(t) - \tau(t) \le y_i(t) \le y(t) \right) \\
& \times \ind \left( y(t) - \tau(t) \le y_i(t) \le y(t) + \tau(t); \forall t \in T \right) \ .
\end{align*}

\subsection{Similarity based on half-region depth}

The half-region depth similarity $s_{HR}(x, y; \mathbf{y}_n)$ of the trajectories pair $(x, y)$ given a functional data set $\mathbf{y}_n$ is
\begin{equation*}
s_{HR}(x, y; \mathbf{y}_n) = \min(R_{hyp}(x \wedge y; \mathbf{y}_n), R_{epi}(x \vee y; \mathbf{y}_n)) \ .
\end{equation*}
where 
\begin{align*}
x \wedge y & = \{ (t,z(t)) : t \in T, z(t) = \min[x(t), y(t)] \} \ , \\
x \vee y & = \{ (t,z(t)) : t \in T, z(t) = \max[x(t), y(t)] \} \ .
\end{align*}

The local half-region depth similarity $ls_{HR}(x, y; \mathbf{y}_n)$ of the trajectories pair $(x, y)$ given a functional data set $\mathbf{y}_n$ is given by
\begin{equation*}
ls_{HR}(x, y; \mathbf{y}_n, \tau) = \min(R_{hyp}(x \wedge y; \mathbf{y}_n, \tau), R_{epi}(x \vee y; \mathbf{y}_n, \tau)) \ .
\end{equation*}
Notice that, $ls_{HR}(y, y; \mathbf{y}_n, \tau) = ld_{HR}(y; \mathbf{y}_n, \tau)$ and
\begin{equation*}
ls_{HR}(x, y; \mathbf{y}_n, \tau) \leq \min(ld_{HR}(x; \mathbf{y}_n, \tau), ld_{HR}(y; \mathbf{y}_n, \tau))
\end{equation*}  
for all $\tau > 0$. Let $z = x \vee y$ and $w = x \wedge y$. The sample version of the local modified half-region depth similarity for the couple $(x,y)$ of trajectories, given the functional data set $\mathbf{y}_n$ is
\begin{equation*}
ls_{MHR}(x, y; \mathbf{y}_n, \tau) = \min(\EL(z, w; \mathbf{y}_n, \tau), \HL(z, w;\mathbf{y}_n,\tau)) \ ,
\end{equation*}
where
\begin{align*}
\EL(z, w;\mathbf{y}_n,\tau) & = \frac{1}{n \lambda(T)} \sum_{i=1}^n \lambda \left(t \in T: z(t) \le y_i(t) \le z(t) + \tau \right) \\
& \times \ind \left( z(t) - \tau \le y_i(t) \le w(t) + \tau; \forall t \in T \right) \\
\HL(z, w;\mathbf{y}_n,\tau) & = \frac{1}{n \lambda(T)} \sum_{i=1}^n \lambda\left(t \in T: w(t) - \tau \le y_i(t) \le w(t) \right) \\
& \times \ind \left( z(t) - \tau \le y_i(t) \le w(t) + \tau; \forall t \in T \right) \ .
\end{align*}
Again, we have $ls_{MHR}(y, y; \mathbf{y}_n, \tau) = ld_{MHR}(y; \mathbf{y}_n, \tau)$ and
\begin{equation*}
ls_{MHR}(x, y; \mathbf{y}_n, \tau) \leq \min(ld_{MHR}(x; \mathbf{y}_n, \tau), ld_{MHR}(y; \mathbf{y}_n, \tau))
\end{equation*}  
for all $\tau > 0$.

There are several ways to construct Dissimilarity/Distance matrix $D=[d_{ij}]$ from Similarity matrix $S=[s_{ij}]$. In our analysis we use the following transformation proposed by \citet{gower_1966}
\begin{equation*}
d_{ij} = \left( s_{ii} + s_{jj} -2 s_{ij} \right)^{1/2} \ .
\end{equation*}
Once a dissimilarity matrix is available hierarchical cluster analysis can be easily performed on the set of observed curves.
\section{Example}
\label{SecExamples}
In this section we illustrate the behaviour of the modified half-region depth and its local version using a data set on wind speed. Further examples are available in the Appendix, where two other data sets are analyzed: individual household electric power consumption and data from remote sensing devices.

\subsection{Wind speed}

%%%%% TRE GRUPPI
For this example we use a data set obtained from the Col De la Roa (Italy) meteorological station available at \texttt{intra.tesaf.unipd.it/Sanvito}. We concentrate on the wind speed recorded regularly every $15$ minutes in the years $2001$-$2004$. After removing $42$ incomplete days we obtain $1420$ time series of length $4 \times 24 = 96$. We apply the modified half-region depth and the local modified half-region depth with $\tau=2.619$, which corresponds to the $20\%$ quantile order of the empirical distribution of the sup norm between two time series.

The two depths provide very different ranking for most series, while some ranks are preserved, see the Figure \ref{fig_vel_dd} in the Appendix where we report the DD plot. This is confirmed by the plots in Figure \ref{fig_vel_ranks} where each series is plotted with a different color and width according to its rank, darker and wider means higher rank. We also investigate the possible presence of subgroups or structures. At this aim, we evaluate the dissimilarity measures based on the local modified half-region depth and modified half-region depth and run a hierarchical cluster analysis based on the Ward distance. In Figure \ref{fig_vel_den_sil} the dendrograms and the silhouette graphs are reported. Inspections of the dendrograms and several silhouette plots (not reported) shows that two groups are suggested by the modified half-region depth while three groups are supported by the local version. The classification is very different. Groups based on depth dissimilarity are very similar and they do not identify different pattern (see Figure \ref{fig_vel_clusters2} in the Appendix), while, the first group provided by local depth dissimilarity is characterized by days of wind calm with modest changes during the day; the second group shows an interesting pattern during the day: wind calm during the night and in the early morning with almost absence of wind between 8am-10am, moderate wind speed for the rest of the day until 9pm. The third group is formed by days with higher wind speed, with variability during the day and picks. Further comments and analysis are performed on this data set and reported in the Appendix.

\section{Concluding remarks}
\label{SecConclusions} 
Statistical data depth provides appropriate tools for the analysis of complex datasets such as functional data. However, since the nature of depth functions the analysis is not very suited for non homogeneous data, where subgroups/substructures are present. Instead, local depth is able to capture local features of the data and to provide a better analysis of the observations. Local modified half-region depth proved to be a suitable tool for the analysis of functional data as well as high dimensional data. Hierarchical clustering can be easily performed using similarity measures based on local modified half-region depth, classification is also straightforward but not reported in this work.

\section*{Acknowledgements}
All statistical analysis were performed on SCSCF (www.dais.unive.it/scscf), a multiprocessor cluster system owned by Ca' Foscari University of Venice running under GNU/Linux.

\clearpage

\begin{figure}
\begin{center}
  \includegraphics[width=0.45\textwidth]{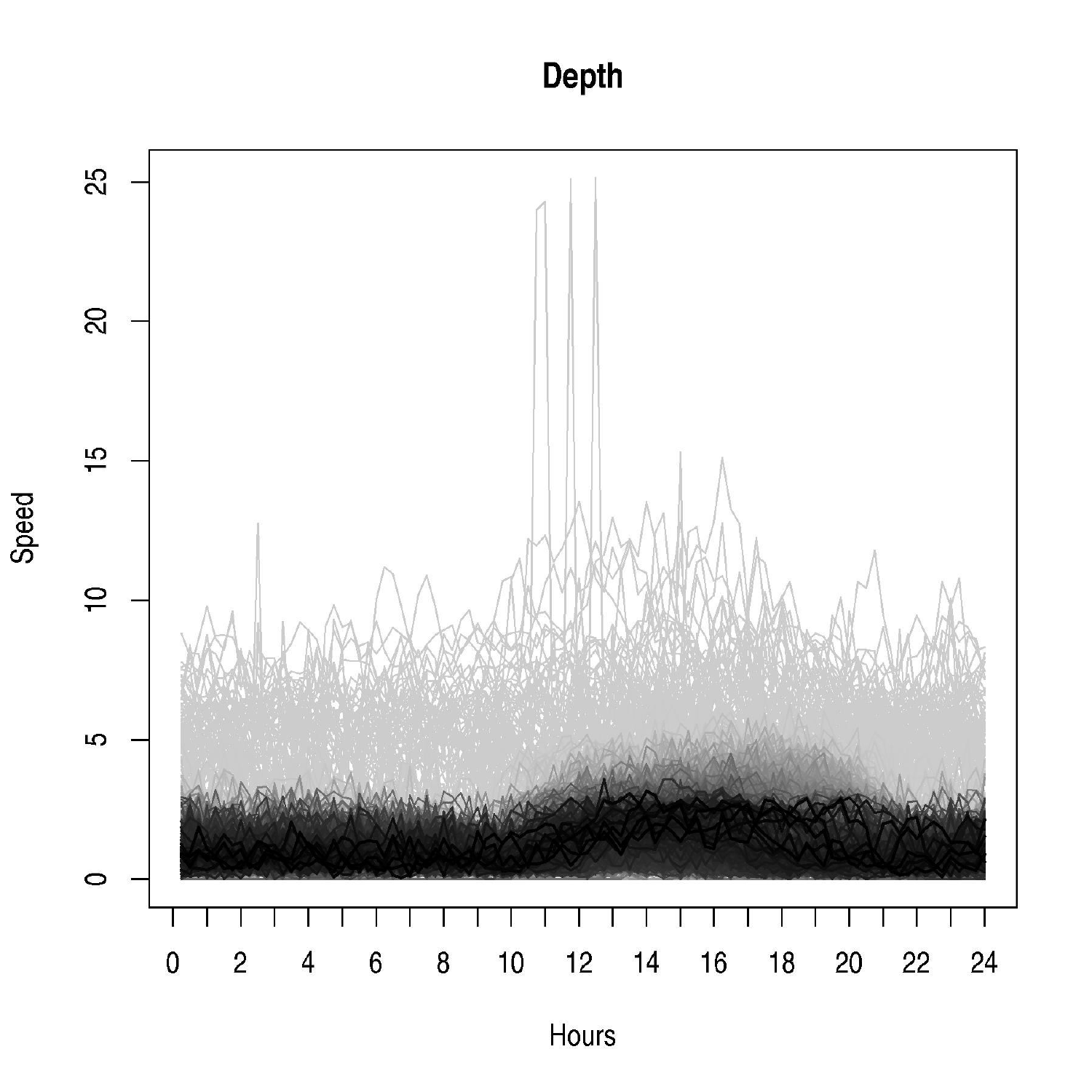}
  \includegraphics[width=0.45\textwidth]{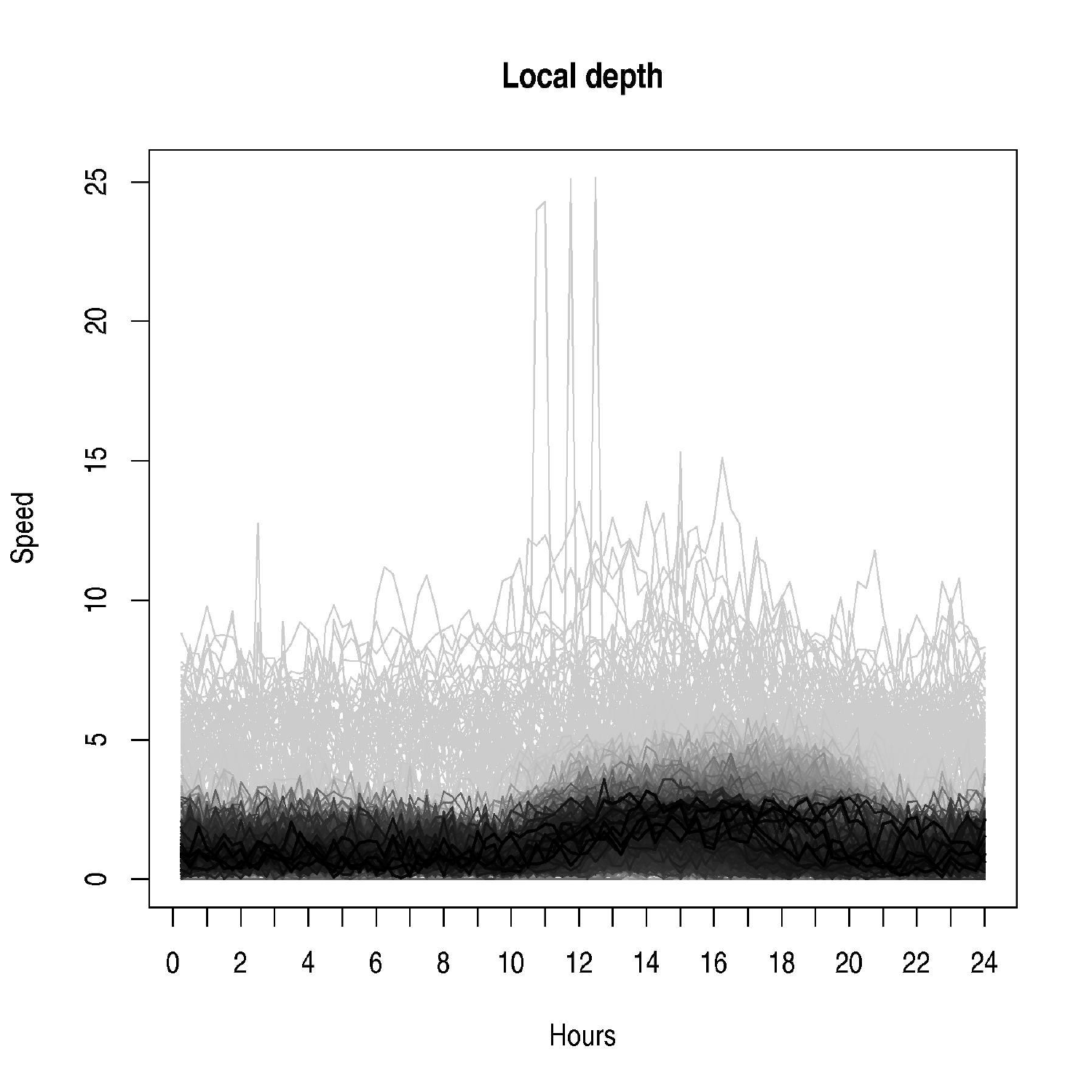}  
\end{center}
\caption{Wind Speed. Ranks provided by modified half-region depth (left panel) and local modified half-region depth (right panel). Darker and wider means higher rank.}
\label{fig_vel_ranks}
\end{figure}  

\clearpage

\begin{figure}
\begin{center}
  \includegraphics[width=0.45\textwidth]{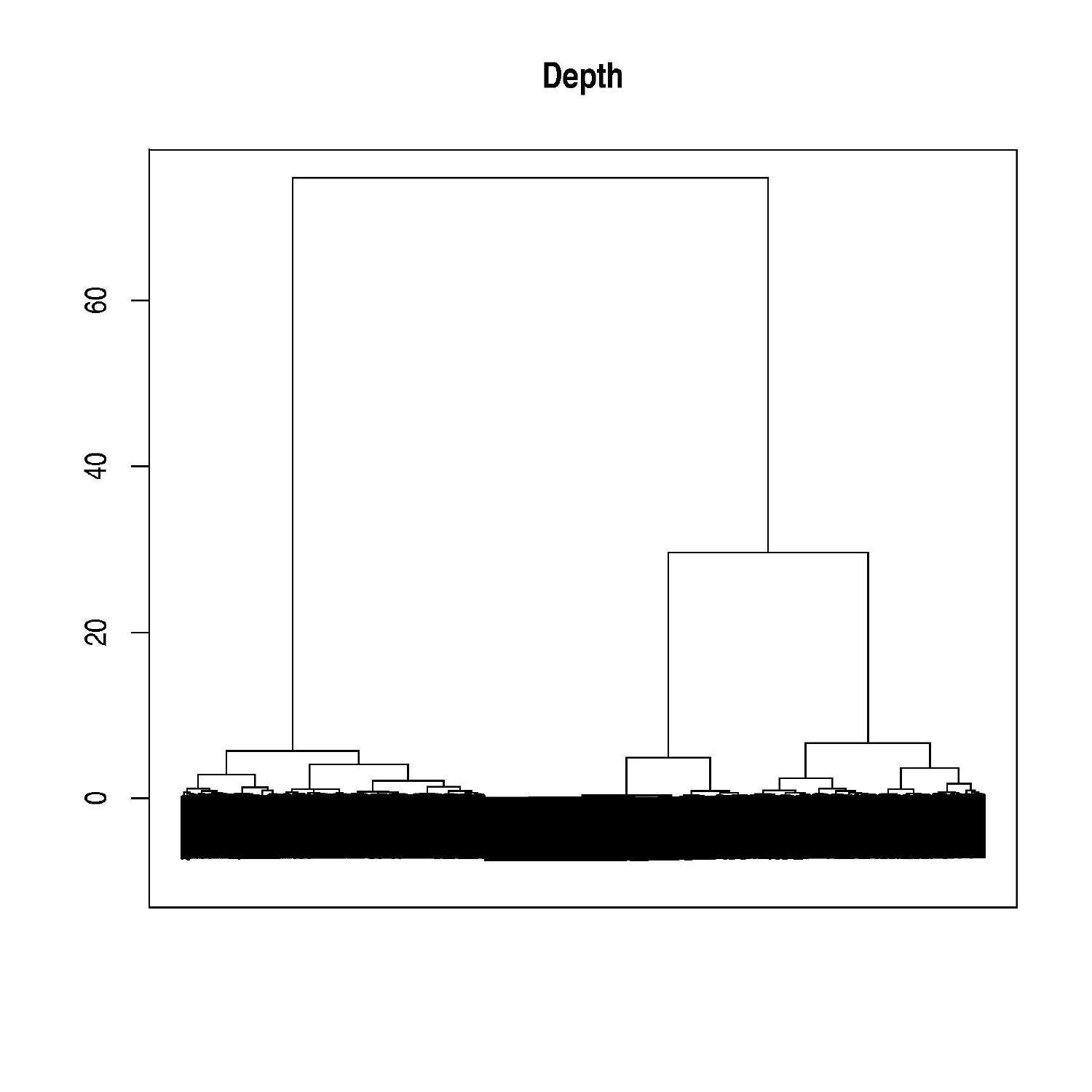}
  \includegraphics[width=0.45\textwidth]{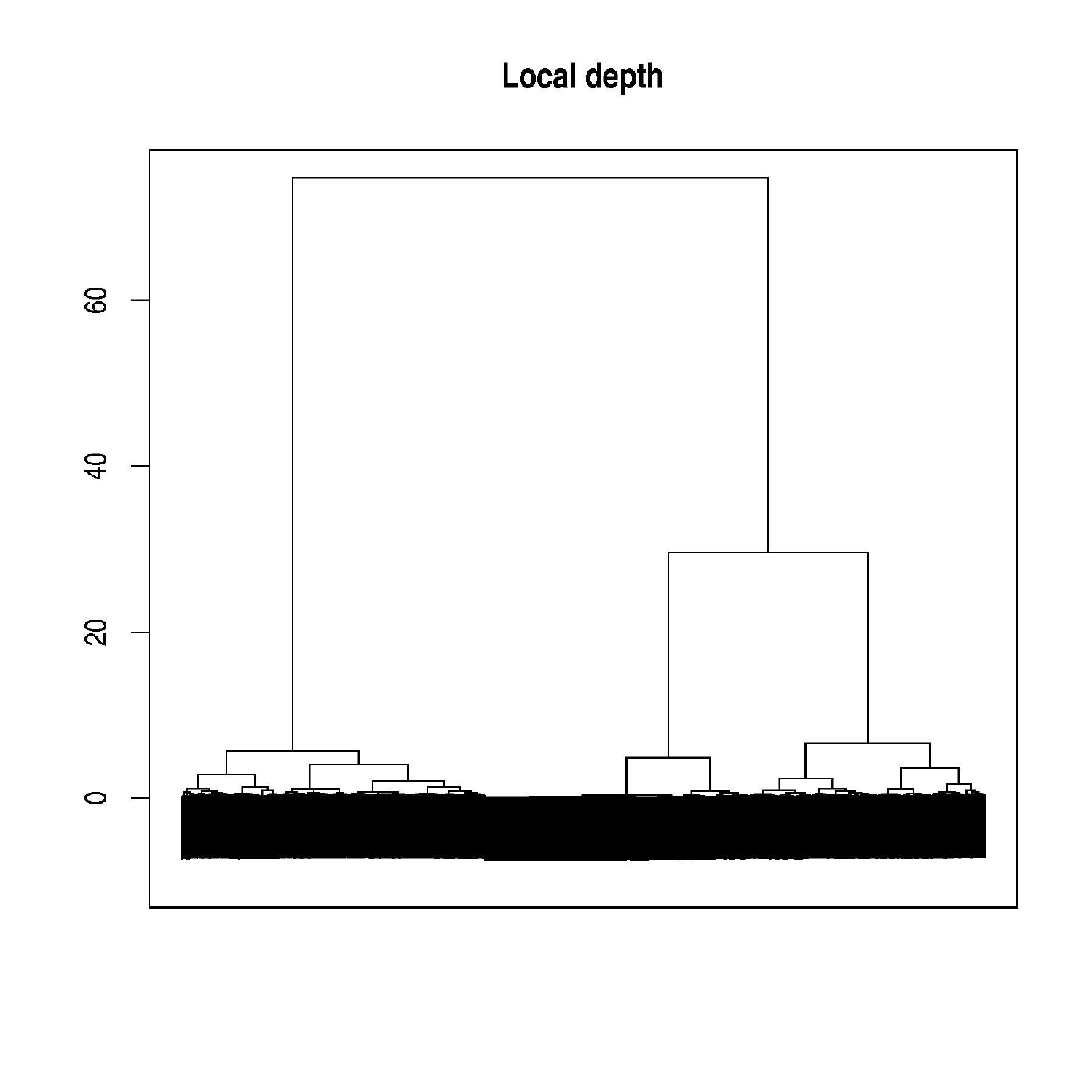}
  
  \includegraphics[width=0.45\textwidth]{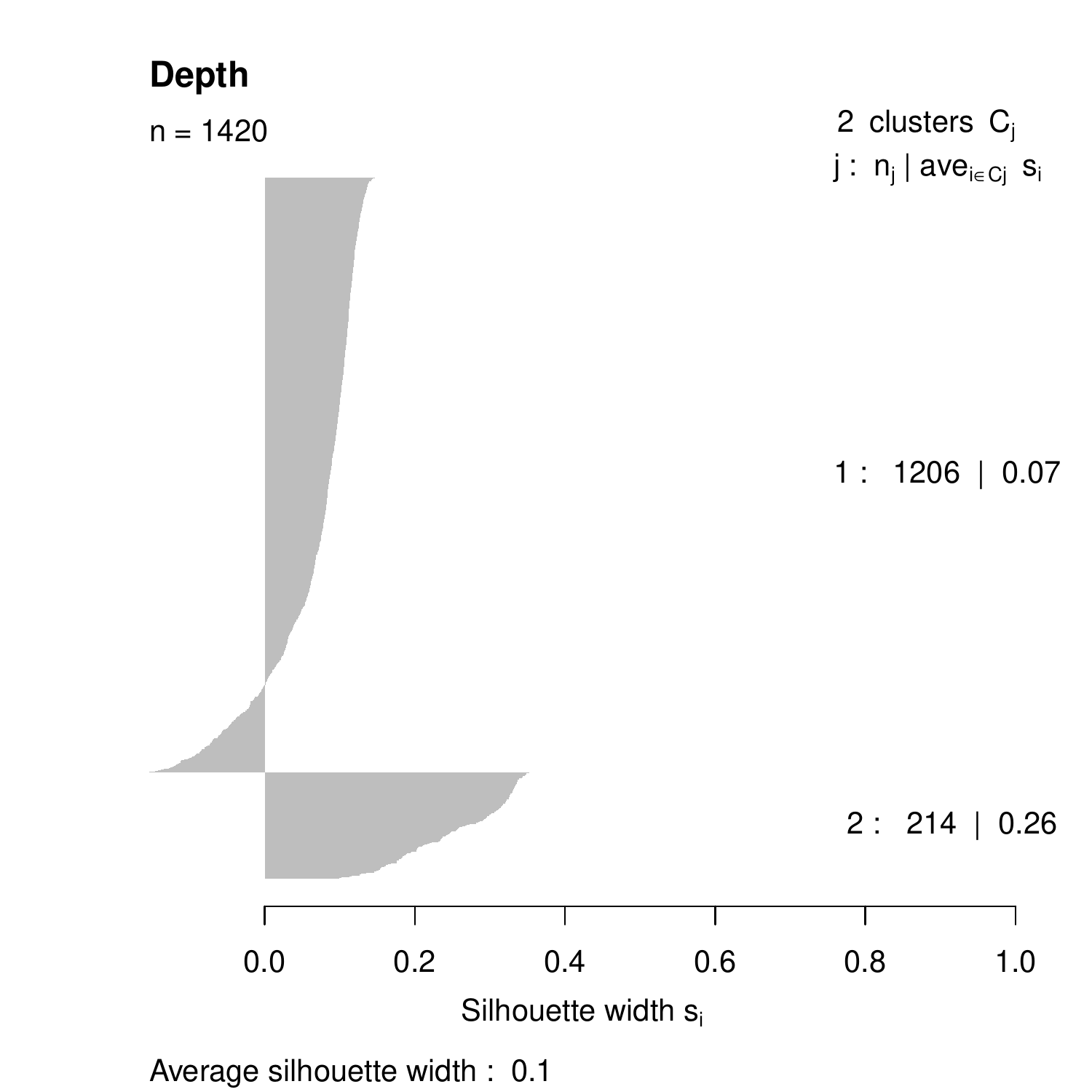}
  \includegraphics[width=0.45\textwidth]{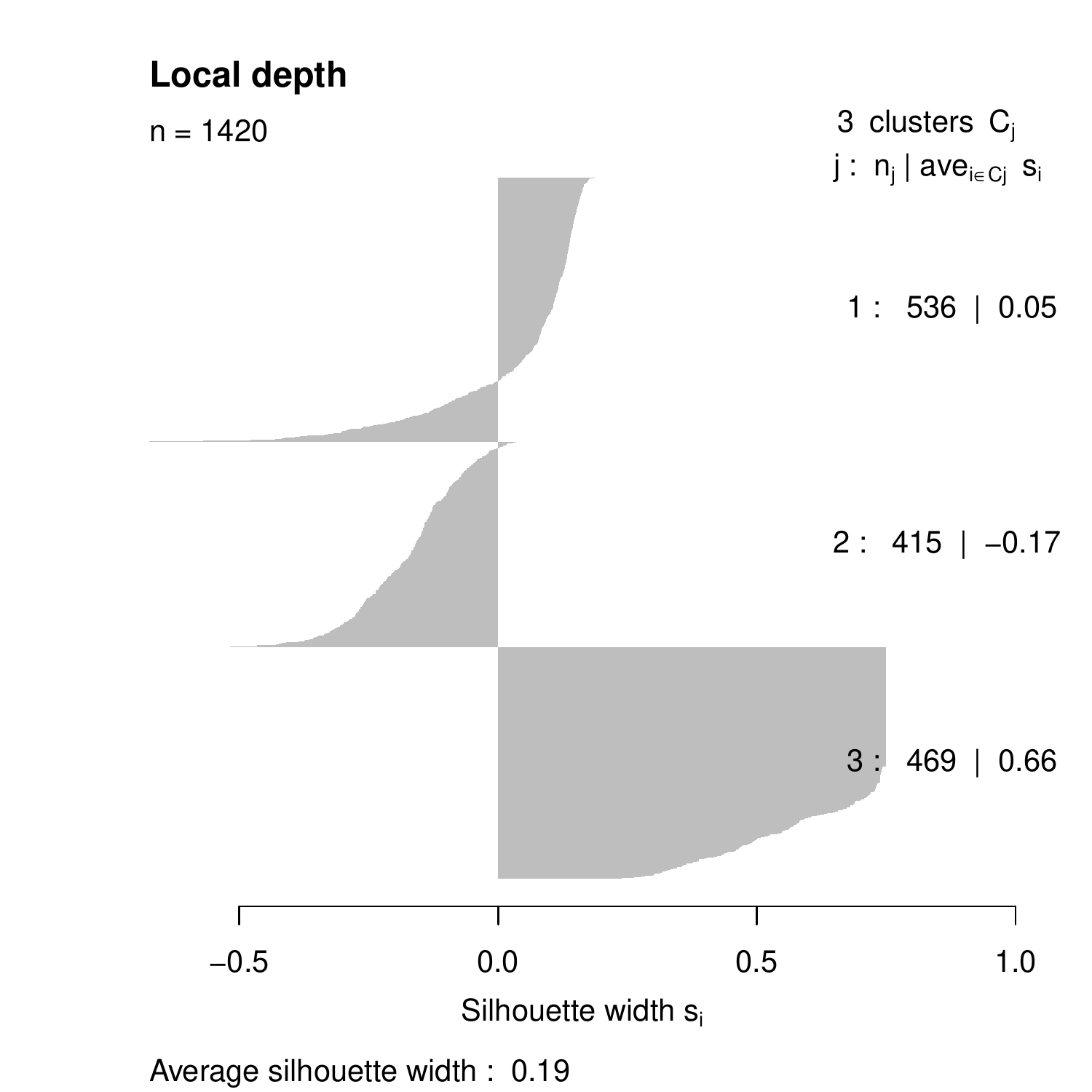}
  
\end{center}
\caption{Wind Speed. First row, dendograms, second row silhouette plot, first column modified half-region depth, second column local modified half-region depth.}
\label{fig_vel_den_sil}
\end{figure}  

\clearpage

\begin{figure}
\begin{center}  
  \includegraphics[width=0.3\textwidth]{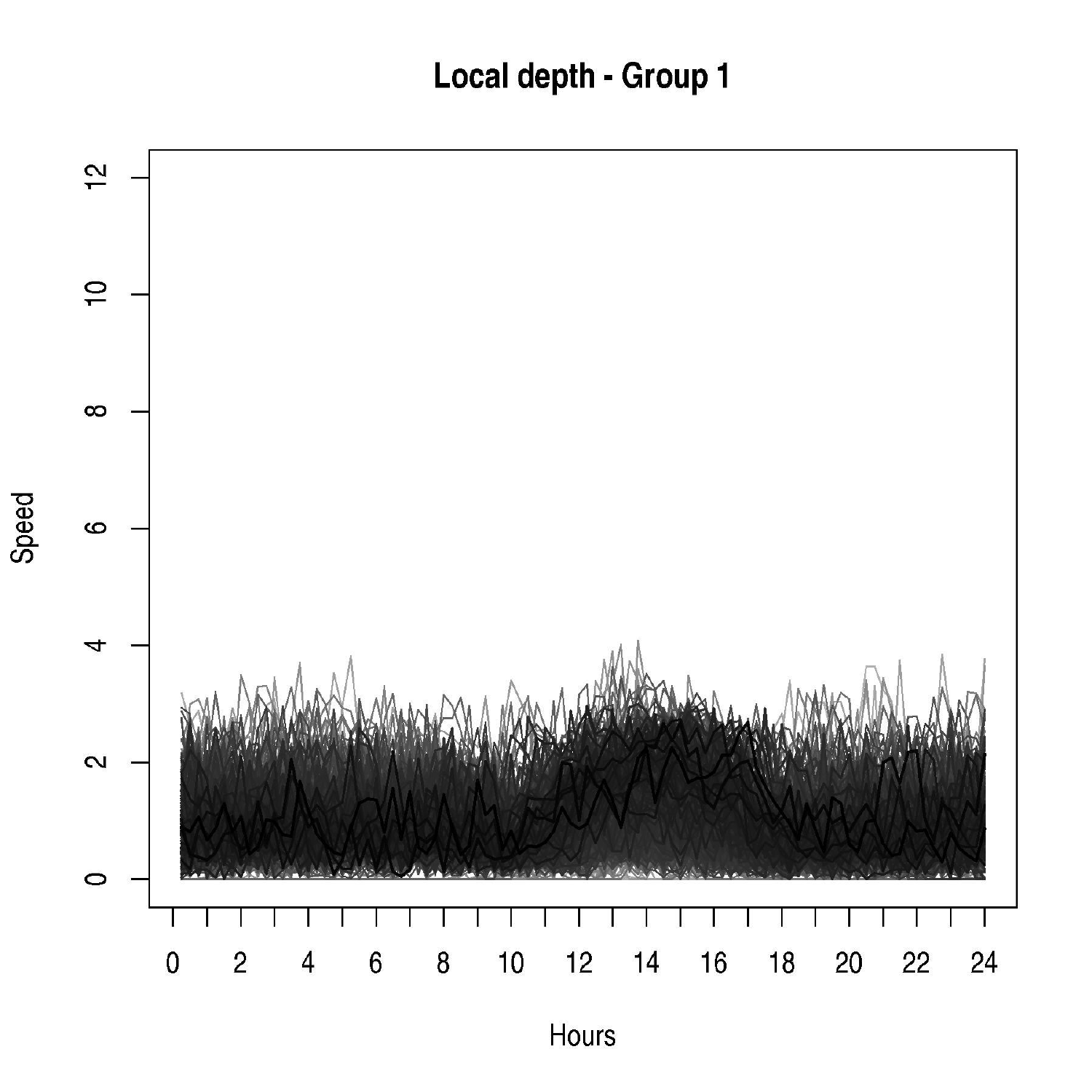}
  \includegraphics[width=0.3\textwidth]{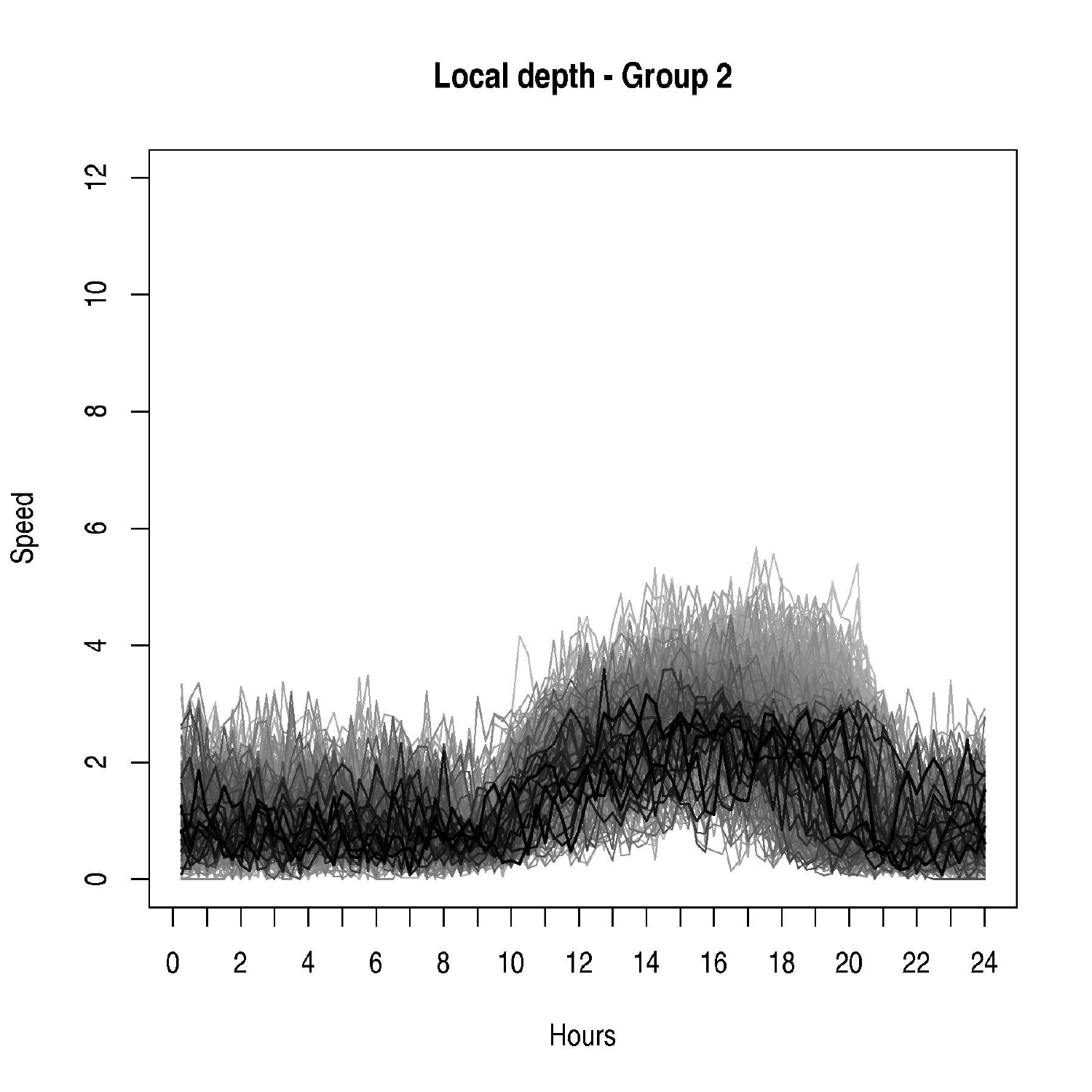}
  \includegraphics[width=0.3\textwidth]{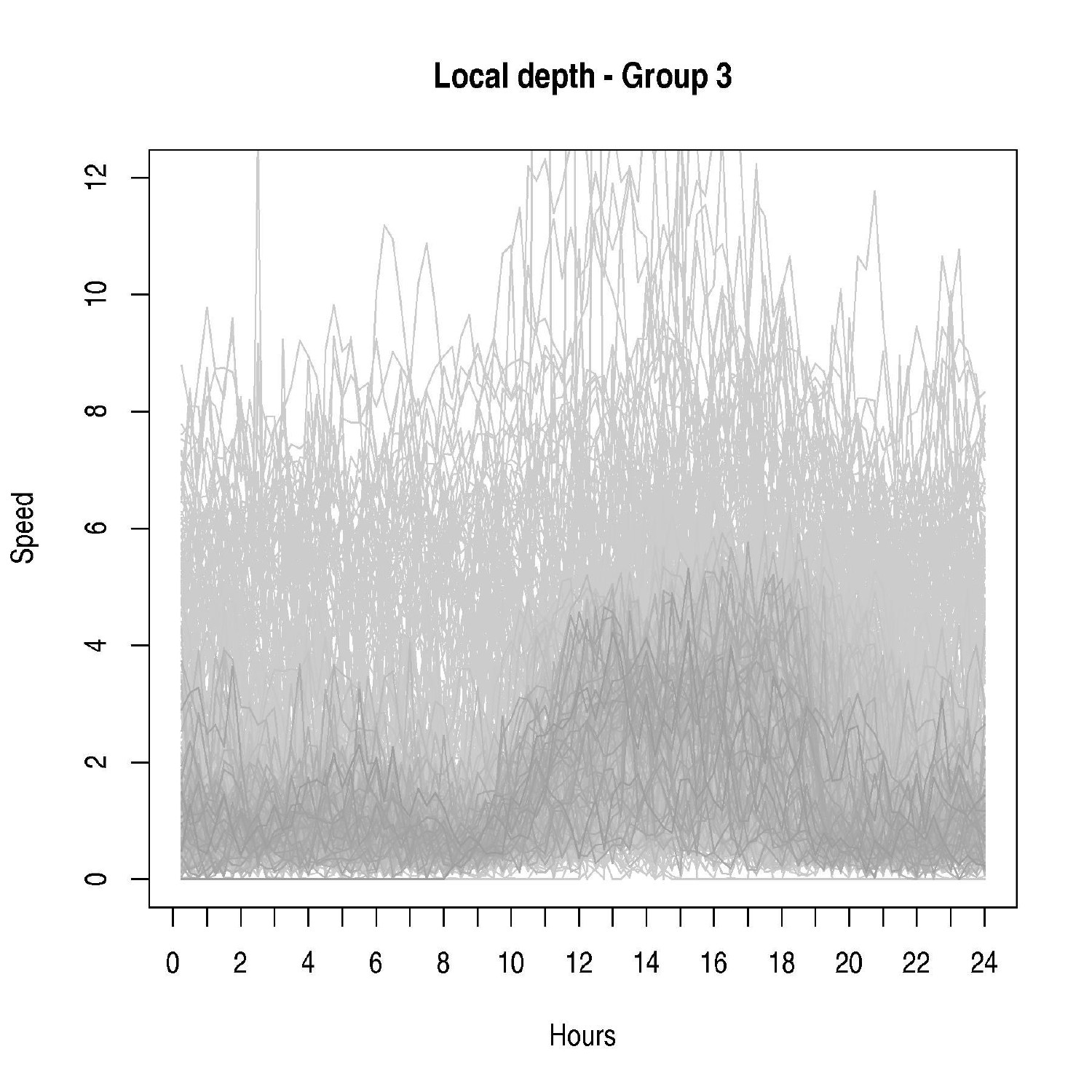}
\end{center}
\caption{Wind Speed. Groups provided by a cluster analysis based local modified half-region depth similarity when three groups are formed. The curves are plotted with color and thickness according to their depth/local depth.}
\label{fig_vel_clusters3}
\end{figure}  

\clearpage

\appendix

\section*{Appendix}

Section \ref{sec:plhrdfdc} studies the properties of the local half-region depth in the finite dimensional case; Section \ref{sec:velff} continues the example illustrated in Section \ref{SecExamples} with further comments and figures. Section \ref{sec:hpc} shows an example based on individual household electric power consumption and Section \ref{sec:out5d} analyzes the out5d data set which contains data collected from remote sensing devices. In this last example a short illustration of the \texttt{R} package \texttt{ldfun} is provided. Finally, Section \ref{sec:proofs} contains all the proofs.

\section{Properties of the local half-region depth in the finite dimensional case}
\label{sec:plhrdfdc}

The half-region depth is invariant with respect to translations and some types of dilation. Let $A$ be a positive (or negative) definite diagonal matrix and $\mathbf{b} \in \mathbb{R}^p$, then 
\begin{equation*}
d_{HR}(A \mathbf{x} + \mathbf{b}, A \mathbf{X} + \mathbf{b}) = d_{HR}(\mathbf{x}, \mathbf{X}) \ .
\end{equation*}
The local half-region depth exhibit the same type of invariance if the vector $\boldsymbol{\tau}$ is transformed accordingly, i.e., 
\begin{equation*}
ld_{HR}(A \mathbf{x} + \mathbf{b}, A \mathbf{X} + \mathbf{b}, A \boldsymbol{\tau}) = ld_{HR}(\mathbf{x}, \mathbf{X}, \boldsymbol{\tau}) \ .
\end{equation*}
On the other hand, if we restrict the values of $\boldsymbol{\tau}$ always to be equal, than the local half-region depth is invariant to translations and only dilation $A = d I$, where $d$ is a constant scalar and $I$ is the identity matrix. This is, since we use one measure for all the slabs. In this particular situation, the same type of invariance as the half-region depth could be achieved by a componentwise standardization of the data. In practice, only asymptotic invariance can be obtained using a scale estimates for each of the component. The Median Absolute Deviation, without consistency constant could be used effectively.

In the rest of this subsection we study the main properties of the local half-region depth in the finite dimensional case. Most of the propositions are extensions of the results presented in \citet{lopez-pintado_half-region_2011}. The proofs are reported in the Appendix. We start by describing the behaviour of the local half-region depth as function of the $\tau$ parameter.
\begin{proposition}
\label{PropLocalHalfRegionTauFinite}
Let $ld_{HR}(\cdot, \mathbf{X}, \boldsymbol{\tau})$ be the local half-region depth. Then, for any fixed $\mathbf{x} \in \mathbb{R}^{p}$
\begin{enumerate}
\item[i.] if $\tau_1(j) < \tau_2(j)$ ($j=1, \cdots, p$), $ld_{HR}(\mathbf{x}; \mathbf{X}, \boldsymbol{\tau}_1) \leq ld_{HR}(\mathbf{x}; \mathbf{X}, \boldsymbol{\tau}_2)$;

\item[ii.] for any vector $\boldsymbol{\tau}$ such that $\tau(j) \ge 0$, $\Pr_{\mathbf{X}} (\mathbf{X}=\mathbf{x}) \leq ld_{HR}(\mathbf{x}; \mathbf{X}, \boldsymbol{\tau}) \leq d_{HR}(\mathbf{x}; \mathbf{X})$;

\item[iii.] $\lim_{\min(\tau(j), j=1, \cdots, p) \rightarrow \infty} ld_{HR}(\mathbf{x}; \mathbf{X}, \boldsymbol{\tau}) = d_{HR}(\mathbf{x}; \mathbf{X})$;

\item[iv.] $\lim_{\| \boldsymbol{\tau} \| \rightarrow0^{+}} ld_{HR}(\mathbf{x}; \mathbf{X}, \boldsymbol{\tau}) = \Pr_{\mathbf{X}} (\mathbf{X}=\mathbf{x})$.
\end{enumerate}
\end{proposition}
Proofs are available in Section \ref{sec:proofs}. In the analysis of a data set, the vector $\tau$ could be chooses by considering the variability of each variable. For instance, $\tau(j)$ ($j=1,\cdots,p$) could be set as the quantile of the empirical distribution of the distance between all pairs of observations in the $j$ coordinate. In practice, quantile order in the interval $5\%$-$30\%$ proved to be effective in most situations.

In the univariate case, the local half-region depth is equivalent to the local halfspace depth.
\begin{proposition}
\label{PropHalfSpaceFinite}
For $p=1$ the local half-region depth $ld_{HR}$ can be expressed as
\begin{equation*}
ld_{HR}(x; X, \tau) = \min \left( F_X(x) - F_X((x-\tau)^-), F_X(x+\tau) - F_X(x^-) \right)
\end{equation*}
and it is equivalent to the local halfspace depth.
\end{proposition}
The local half-region depth decreases to zero when the point tends to infinity.
\begin{proposition}[Vanishing at infinity]
\label{PropVanishingFinite}
Let $\mathbf{x} \in \mathbb{R}^p$ and $\boldsymbol{\tau} > \mathbf{0}$ then
\begin{align*}
\lim_{M \rightarrow \infty} \sup_{\| \mathbf{x} \| \ge M} ld_{HR}(\mathbf{x}; \mathbf{X}, \boldsymbol{\tau}) = 0 \\
\lim_{M \rightarrow \infty} \sup_{\| \mathbf{x} \| \ge M} ld_{HR}(\mathbf{x}; X_n, \boldsymbol{\tau}) = 0 \\
\end{align*}
\end{proposition}
It is a continuous function when the underlying distribution is absolutely continuous.
\begin{proposition}
\label{PropSemicontinuousFinite}
$ld_{HR}(\cdot; \mathbf{X}, \boldsymbol{\tau})$ is an upper semicontinuous function. Moreover, if $F$ is absolutely continuous then $ld_{HR}(\cdot; \mathbf{X}, \boldsymbol{\tau})$ is continuous.
\end{proposition}
Let $X_n$ be a sample of size $n$ from the random vector $\mathbf{X}$, and $\Pr_{X_n}$ and $\hat{F}_{X_n}$ be the corresponding empirical probability and distribution function respectively. In the next proposition we establish the consistency of the empirical local half-region depth and of its maximizer.
\begin{proposition} \label{PropMaximizerFinite}
$ld_{HR}$ is uniformly consistent:
\begin{equation*}  
\sup_{\mathbf{x} \in \mathbb{R}^p} \left| ld_{HR}(\mathbf{x}; X_n, \boldsymbol{\tau}) - ld_{HR}(\mathbf{x}; \mathbf{X}, \boldsymbol{\tau}) \right| \stackrel{a.s.}{\rightarrow} 0 \ , \qquad n \rightarrow \infty \ .  
\end{equation*}  
Furthermore, if $ld_{HR}(\cdot; \mathbf{X}, \boldsymbol{\tau})$ is uniquely maximized at $\hat{\mathbf{x}}$ and $\hat{\mathbf{x}}_n$ is a sequence of random variables with $ld_{HR}(\hat{\mathbf{x}}_n; X_n, \boldsymbol{\tau}) = \sup_{\mathbf{x} \in \mathbb{R}^p} ld_{HR}(\mathbf{x}; X_n, \boldsymbol{\tau})$, then 
\begin{equation*}
\hat{\mathbf{x}}_n \stackrel{a.s.}{\rightarrow} \hat{\mathbf{x}} \ , \qquad n \rightarrow \infty \ .
\end{equation*}
\end{proposition}

\section{Wind speed}
\label{sec:velff}

Hereafter, we report some more information regarding the example based on the wind speed data set. Figure \ref{fig_vel_dd} presents the DD-plot of the local depth versus the depth. When only one group is present in the data, almost all observations have the same ranks. In our case, most of the points are on the right bottom corner, which means they have high rank in the depth but low rank according to the local depth. Figure \ref{fig_vel_outliers} shows the most $40$ unusual patterns and the most $20$ common patterns as identified by the modified half-region depth and the local modified half-region depth. The local version correctly identifies the serie with very high picks and other series with high variability, all belonging to the $3$rd cluster, while the modified half-region depth identifies curves which belong to its $2$nd cluster, see Figure \ref{fig_vel_clusters2} left panel. Silhouette plot based on two groups for local depth are reported in Figure \ref{fig_vel_den_sil2}. The two groups split seems the best clustering provided by depth, while, local depth suggests a three groups split as shown in the main document at Figure \ref{fig_vel_den_sil}. Finally, for completeness, we show the clustering obtained by two groups split for both depth and local depth in Figure \ref{fig_vel_clusters2}.

\begin{figure}
\begin{center}
  \includegraphics[width=0.45\textwidth]{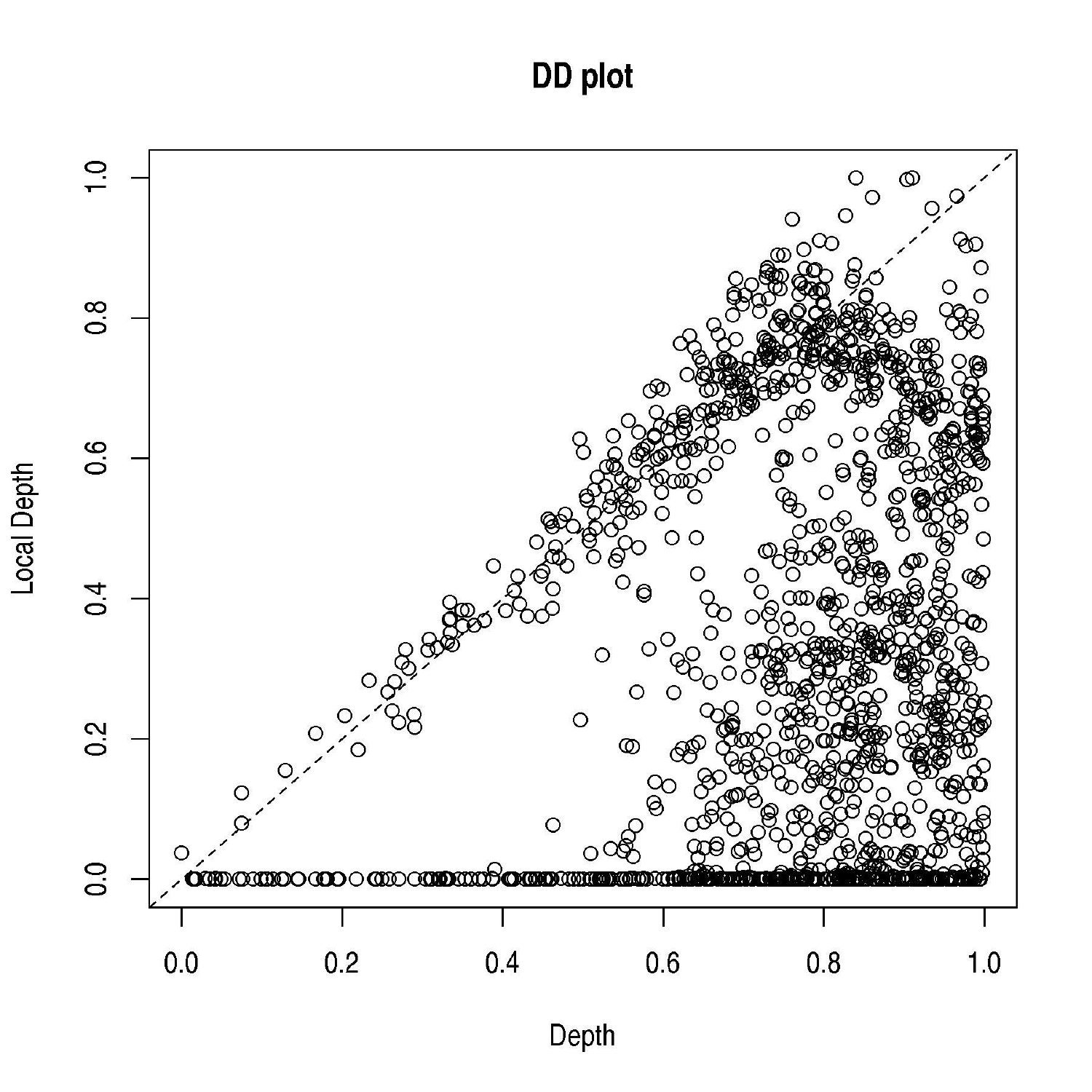}
\end{center}
\caption{Wind Speed. DD plot, local modified half-region depth versus modified half-region depth.}
\label{fig_vel_dd}
\end{figure}  

\clearpage

\begin{figure}
\begin{center}
  \includegraphics[width=0.45\textwidth]{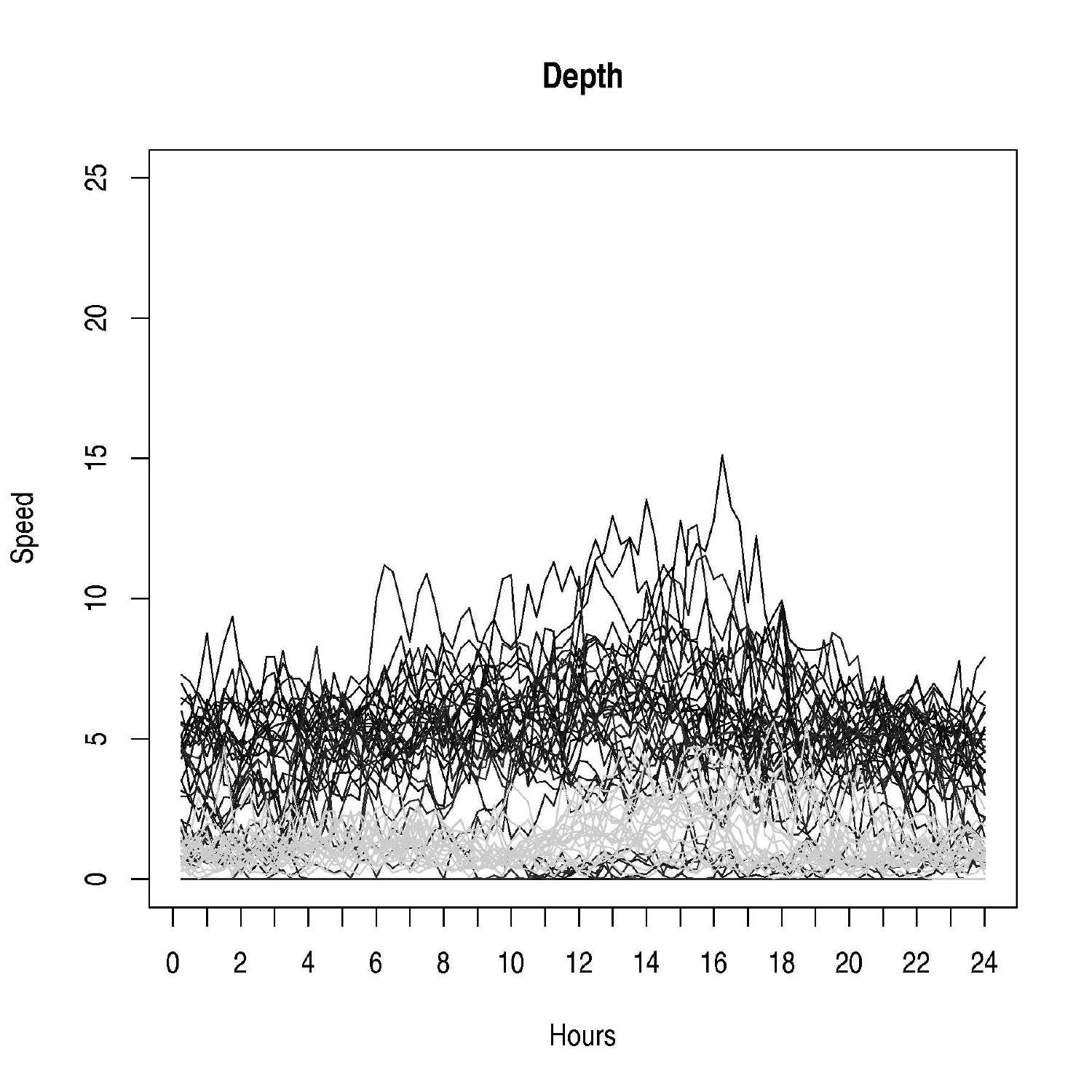}
  \includegraphics[width=0.45\textwidth]{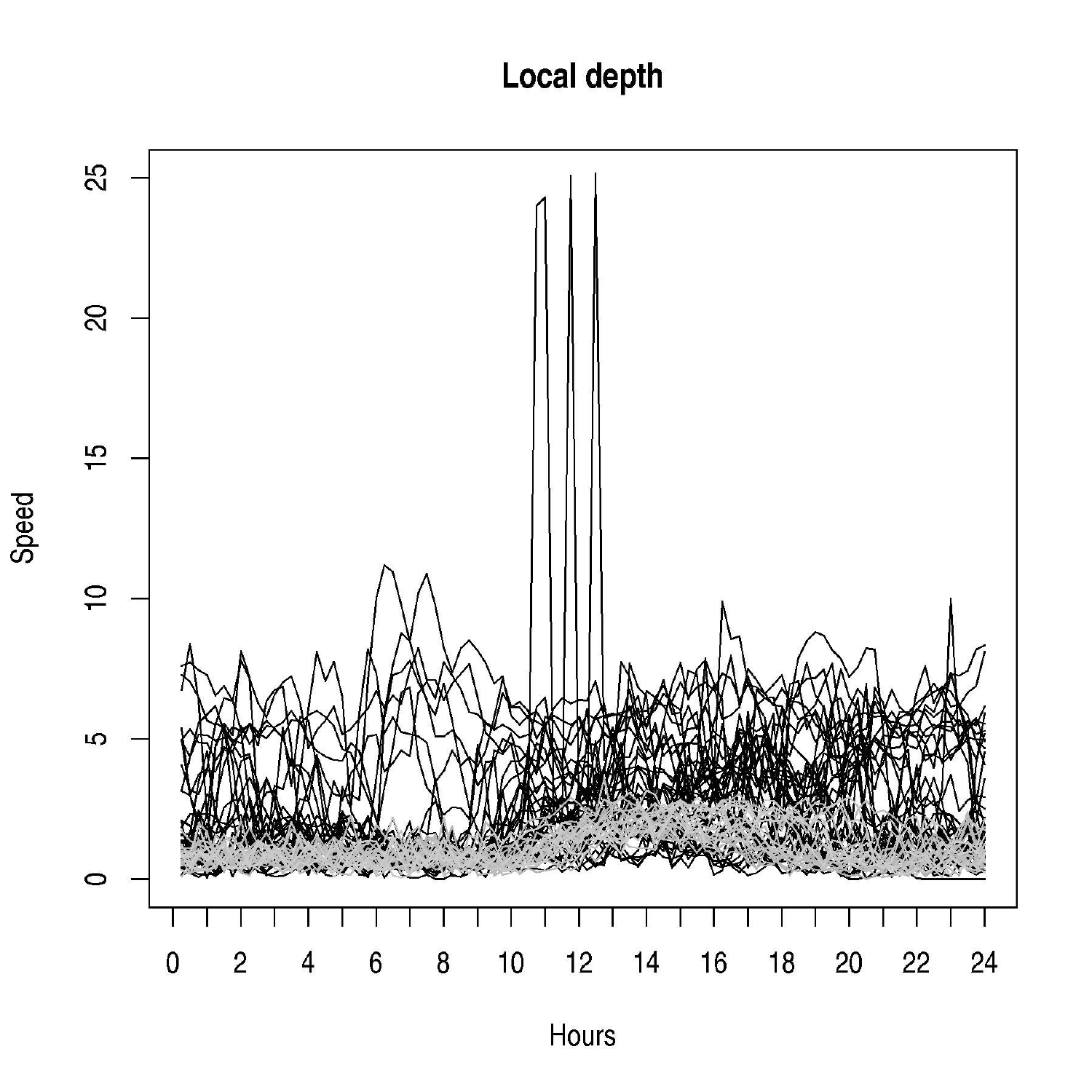}  
\end{center}
\caption{Wind Speed. Unusual patterns (black curves) and most common pattern (light gray curves) by modified half-region depth (left panel) and local modified half-region depth (right panel).}
\label{fig_vel_outliers}
\end{figure}  

\clearpage

\begin{figure}
\begin{center}
  \includegraphics[width=0.45\textwidth]{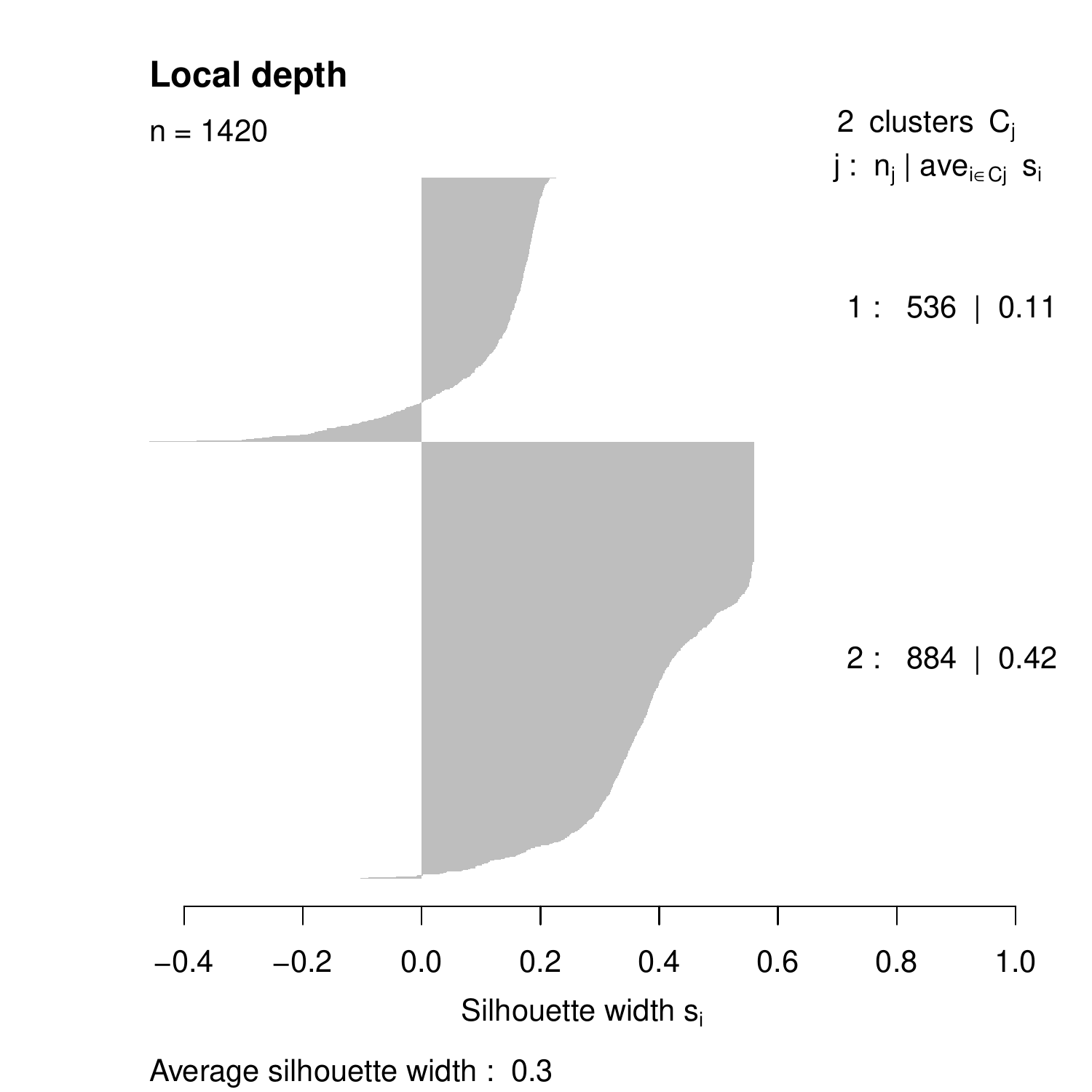}
  
\end{center}
\caption{Wind Speed. silhouette plot based on two groups and local modified half-region depth.}
\label{fig_vel_den_sil2}
\end{figure}  

\clearpage

\begin{figure}
\begin{center}
  \includegraphics[width=0.45\textwidth]{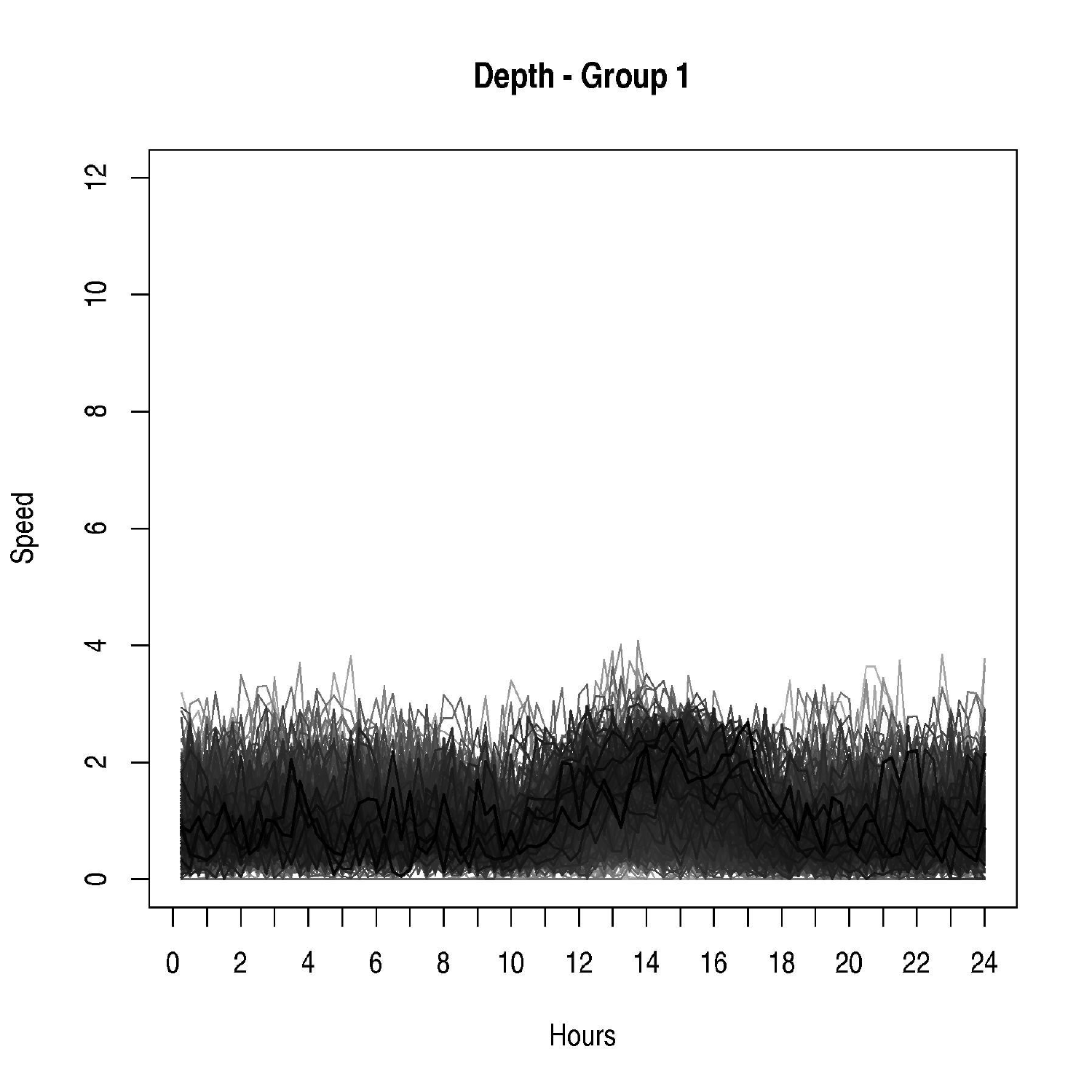}
  \includegraphics[width=0.45\textwidth]{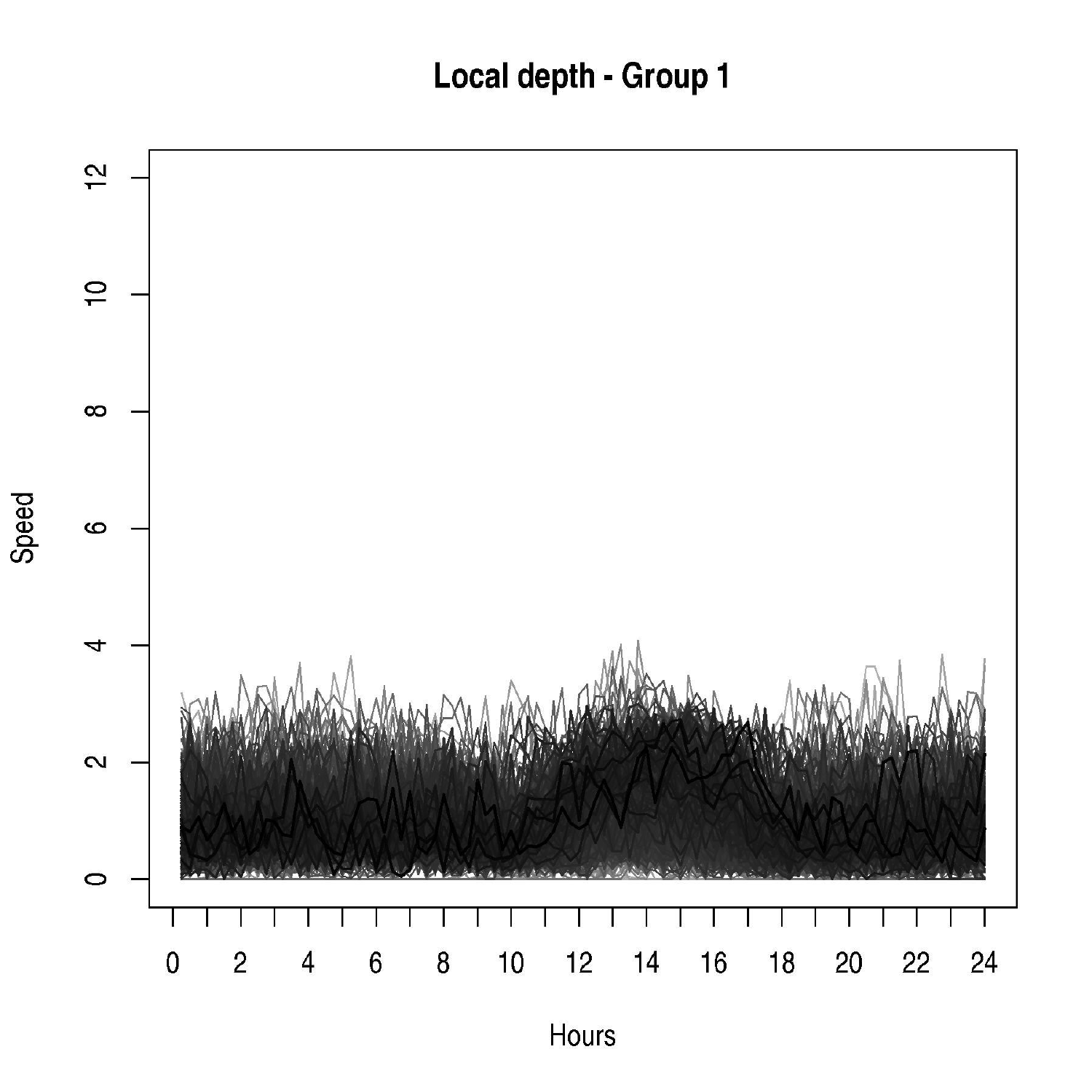}
  
  \includegraphics[width=0.45\textwidth]{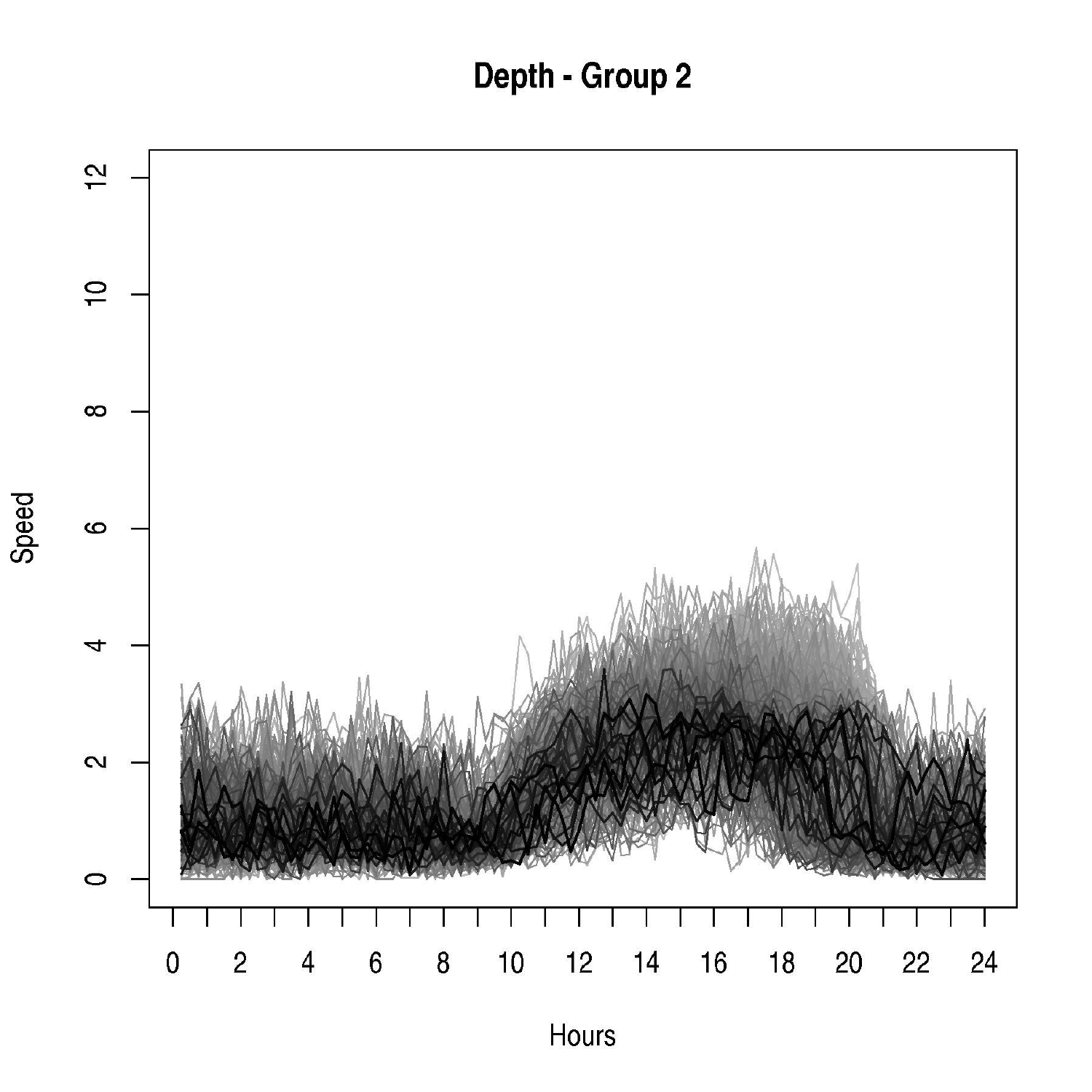}
  \includegraphics[width=0.45\textwidth]{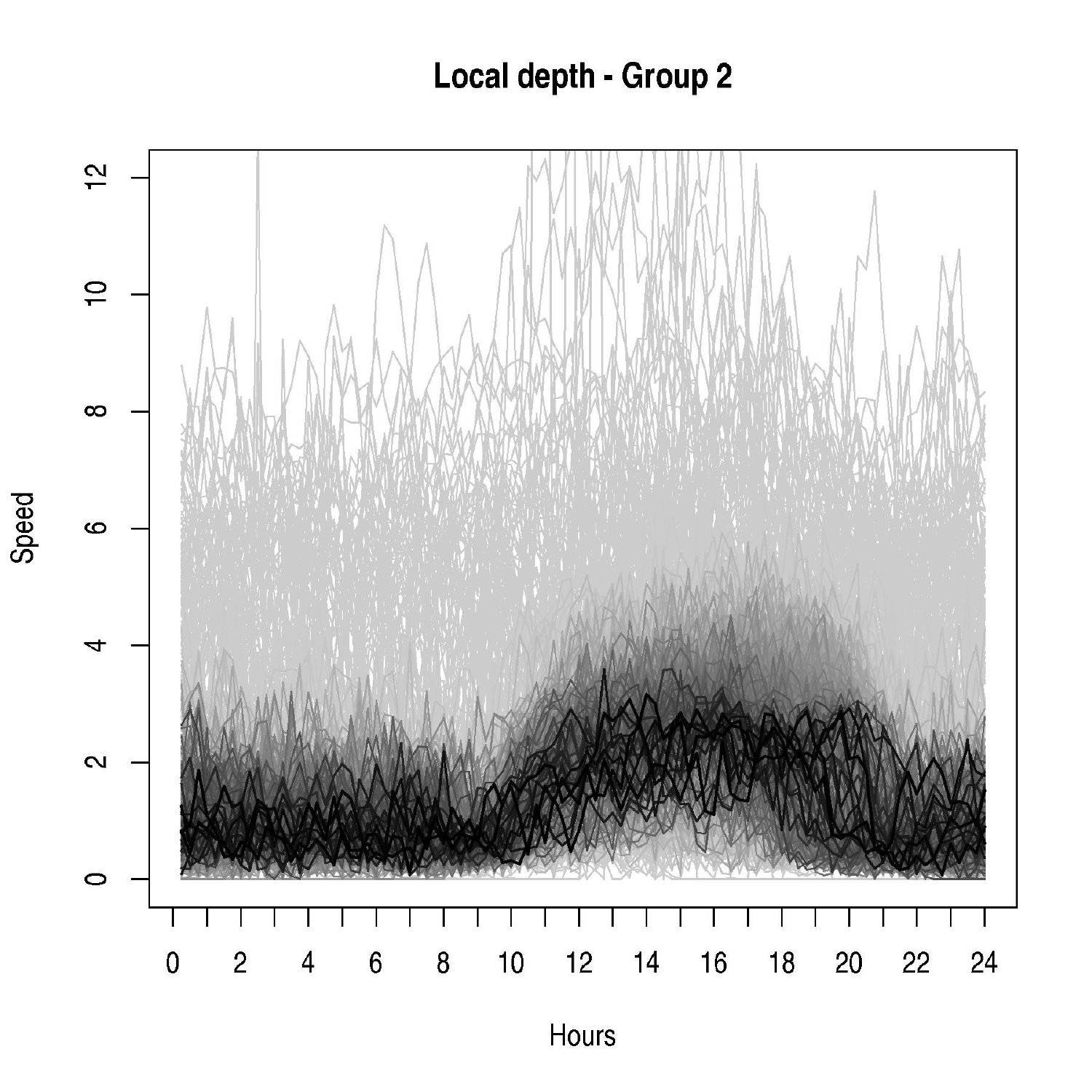}
\end{center}
\caption{Wind Speed. Groups provided by a cluster analysis based on modified half-region depth similarity (first column) and local modified half-region depth similarity (second column). The curves are plotted with color and thickness according to their depth/local depth.}
\label{fig_vel_clusters2}
\end{figure}  

\clearpage

\section{Individual household electric power consumption}
\label{sec:hpc}
%%%%http://archive.ics.uci.edu/ml/datasets/Individual+household+electric+power+consumption

This archive contains $2075259$ measurements gathered between December $2006$ and November $2010$ ($47$ months) on individual household electric power consumption \citep{bache_lichman_2013}, an it is available at their website. Here we concentrate on global active power, i.e., household global minute-averaged active power in kilowatt. After removing 2 days which contain missing values and rearranging the data we obtain $1440$ time series each based on $60\times24=1440$ observations. Figure \ref{fig_hpc_dd} reports the DD-plot. Since the possible presence of two groups, the ranks provided by the two procedures are different. In particular, the local version provides more insight on the structures, high ranks are associated with a specific pattern, highlighted in green in figures, where consumption is almost null during the night, reaches a pick around 7am-8am, return to be null or with a moderate load in the early afternoon and return to have a higher consumption between 7pm and 11pm; see Figure \ref{fig_hpc_ranks}. Figure \ref{fig_hpc_den_sil} contrasts the dendogram and the silhouette plots provided by depth (left panels) and local depth (right panels) and Figure \ref{fig_hpc_clusters} shows the curves according to group membership. The differences are not very marked, however local depth provides a better job, in one group the trajectories with common consumptions, with a particular pattern as described above and picks of consumptions below 2Kw, in the other group the consumption tends to be high during all times.

\clearpage

\begin{figure}
\begin{center}
  \includegraphics[width=0.45\textwidth]{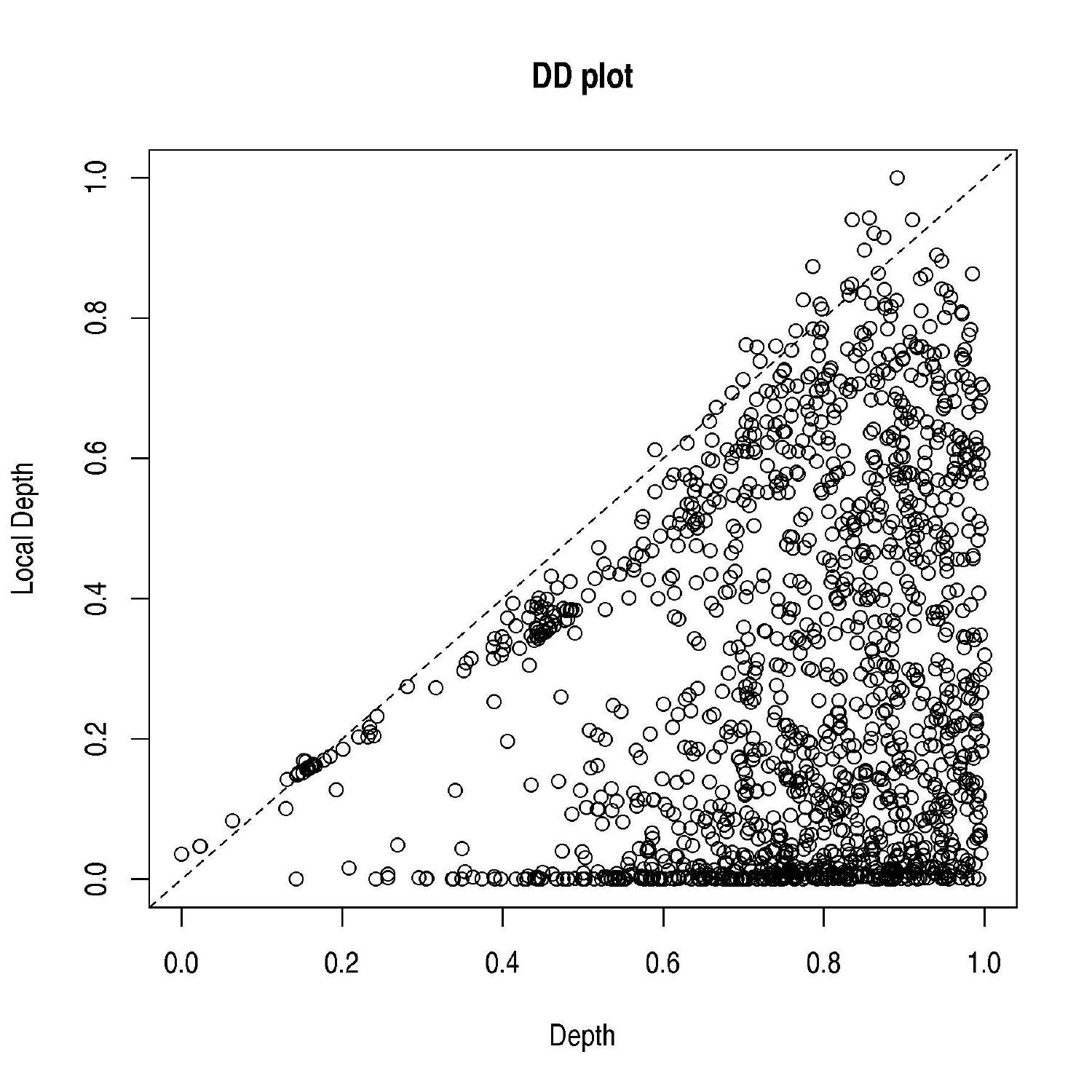}
\end{center}
\caption{Individual Household Electric Power Consumption. DD plot, local modified half-region depth versus modified half-region depth.}
\label{fig_hpc_dd}
\end{figure}  

\clearpage

\begin{figure}
\begin{center}
  \includegraphics[width=0.45\textwidth]{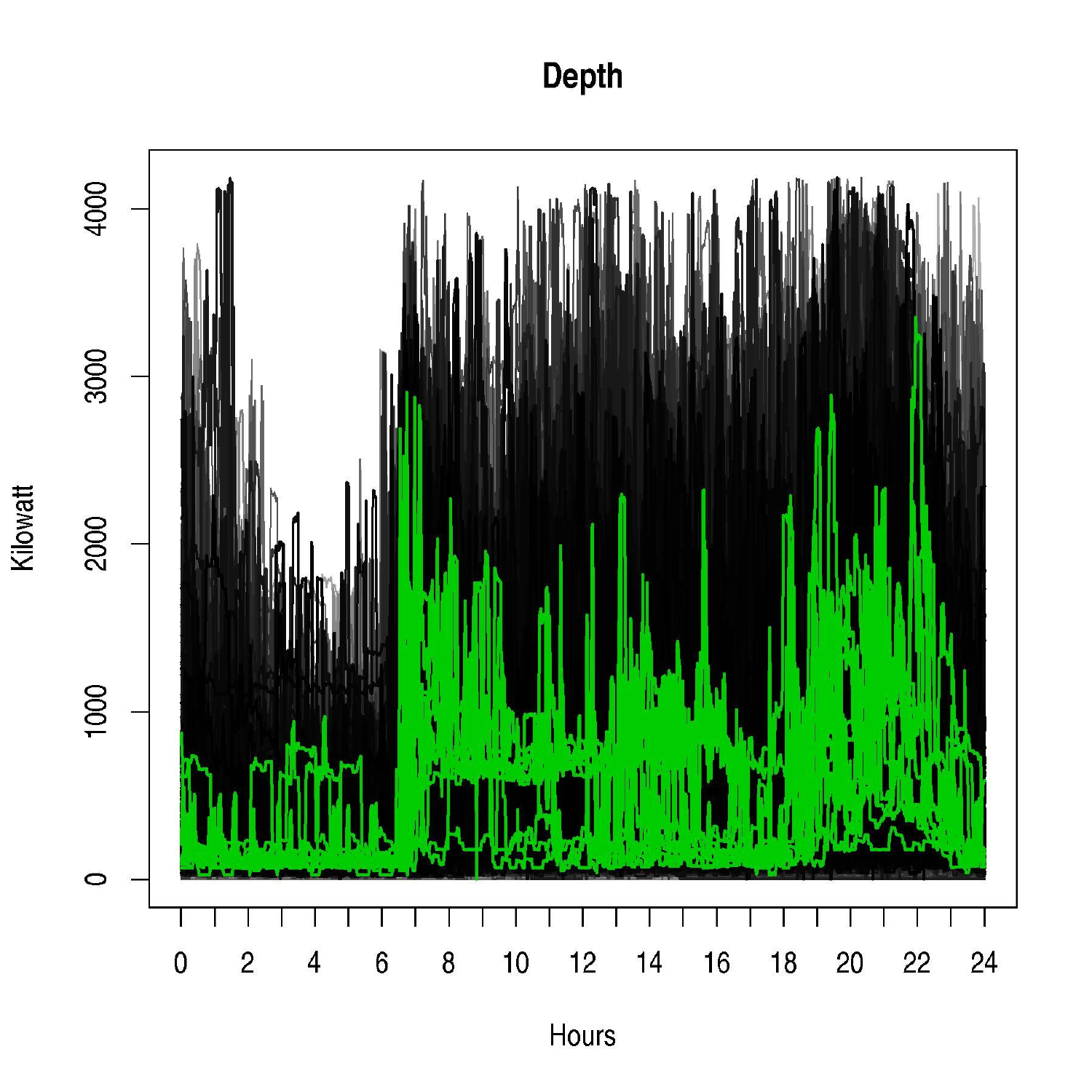}
  \includegraphics[width=0.45\textwidth]{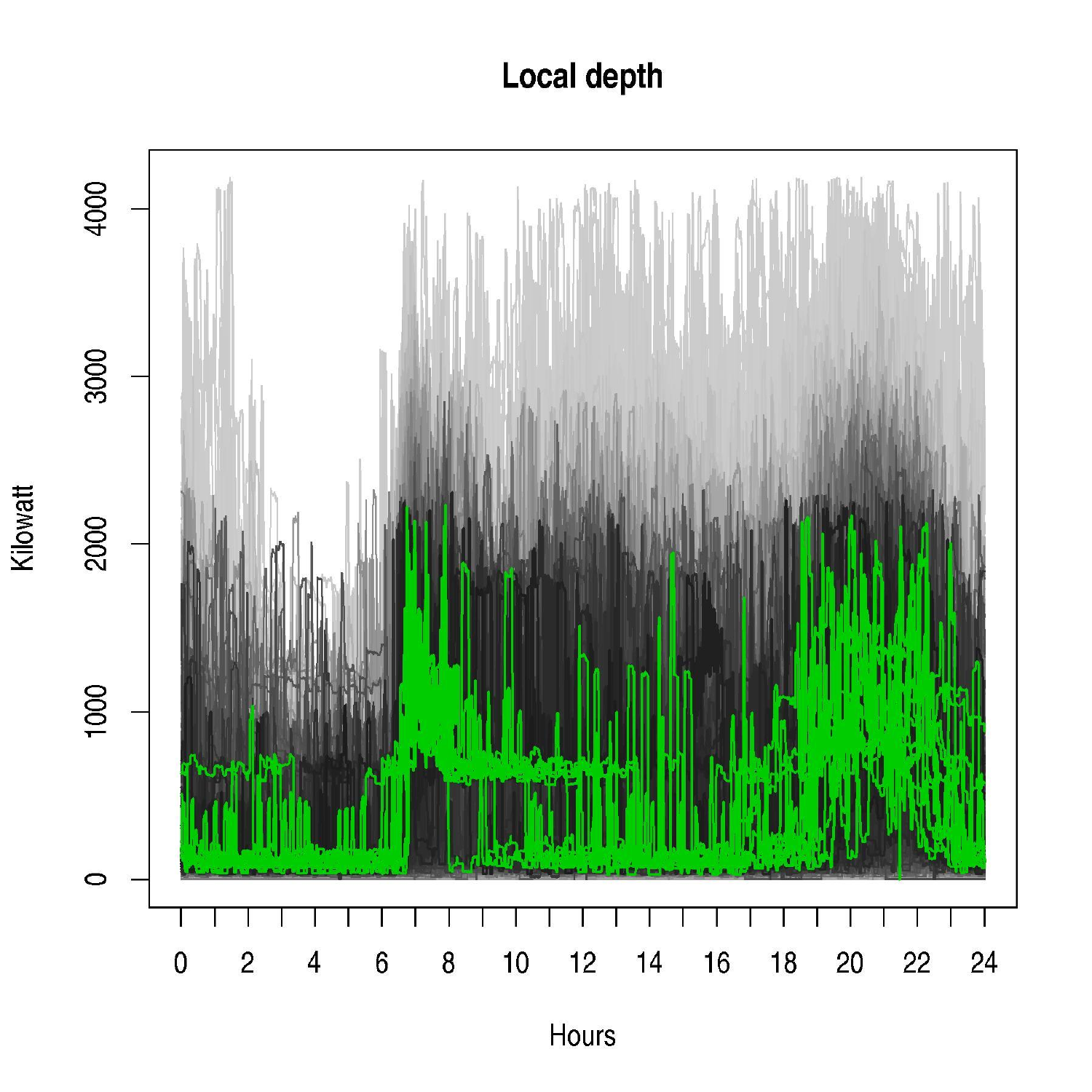}  
\end{center}
\caption{Individual Household Electric Power Consumption. Ranks provided by modified half-region depth (left panel) and local modified half-region depth (right panel). Darker and wider means higher rank.}
\label{fig_hpc_ranks}
\end{figure}  

\clearpage

\begin{figure}
\begin{center}
  \includegraphics[width=0.45\textwidth]{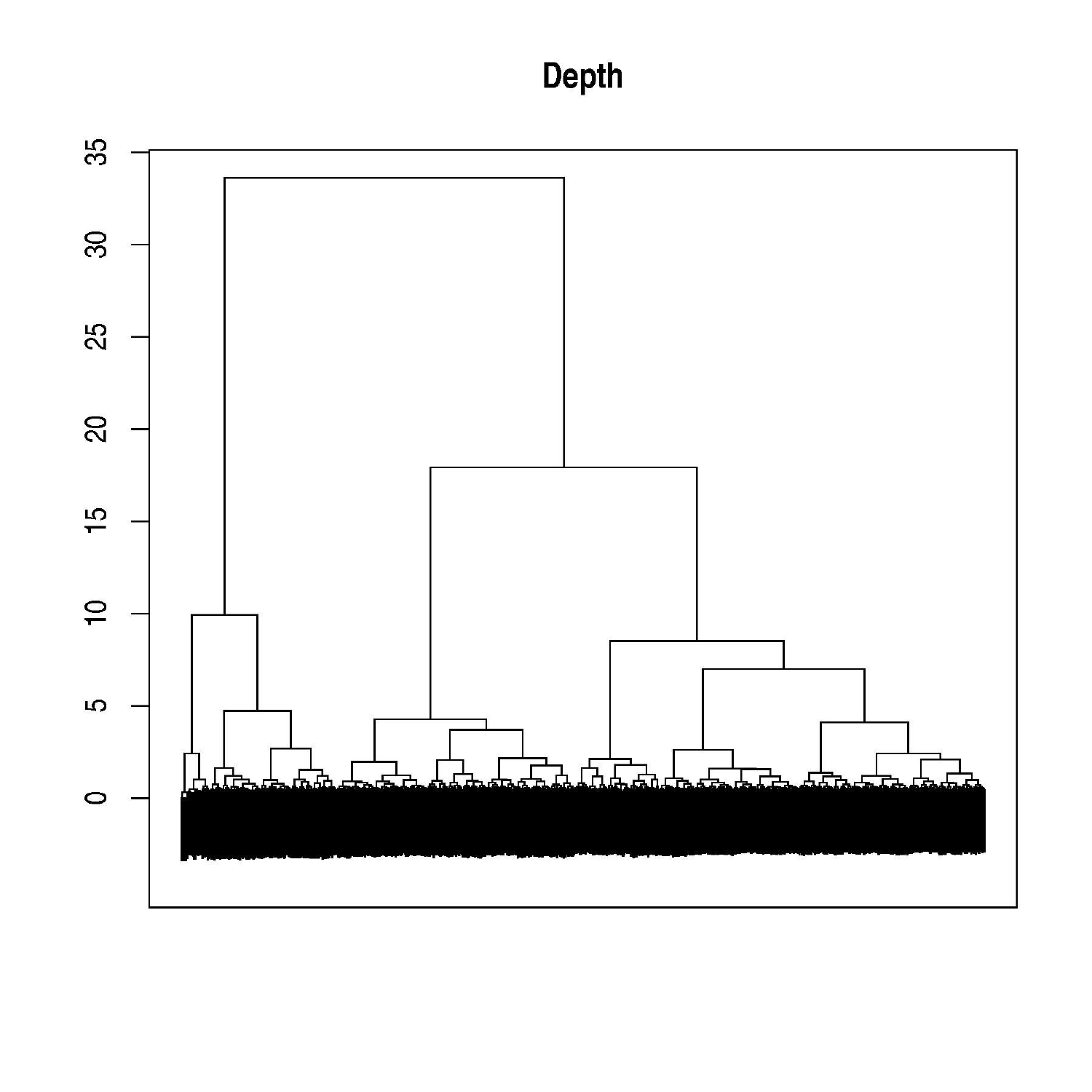}
  \includegraphics[width=0.45\textwidth]{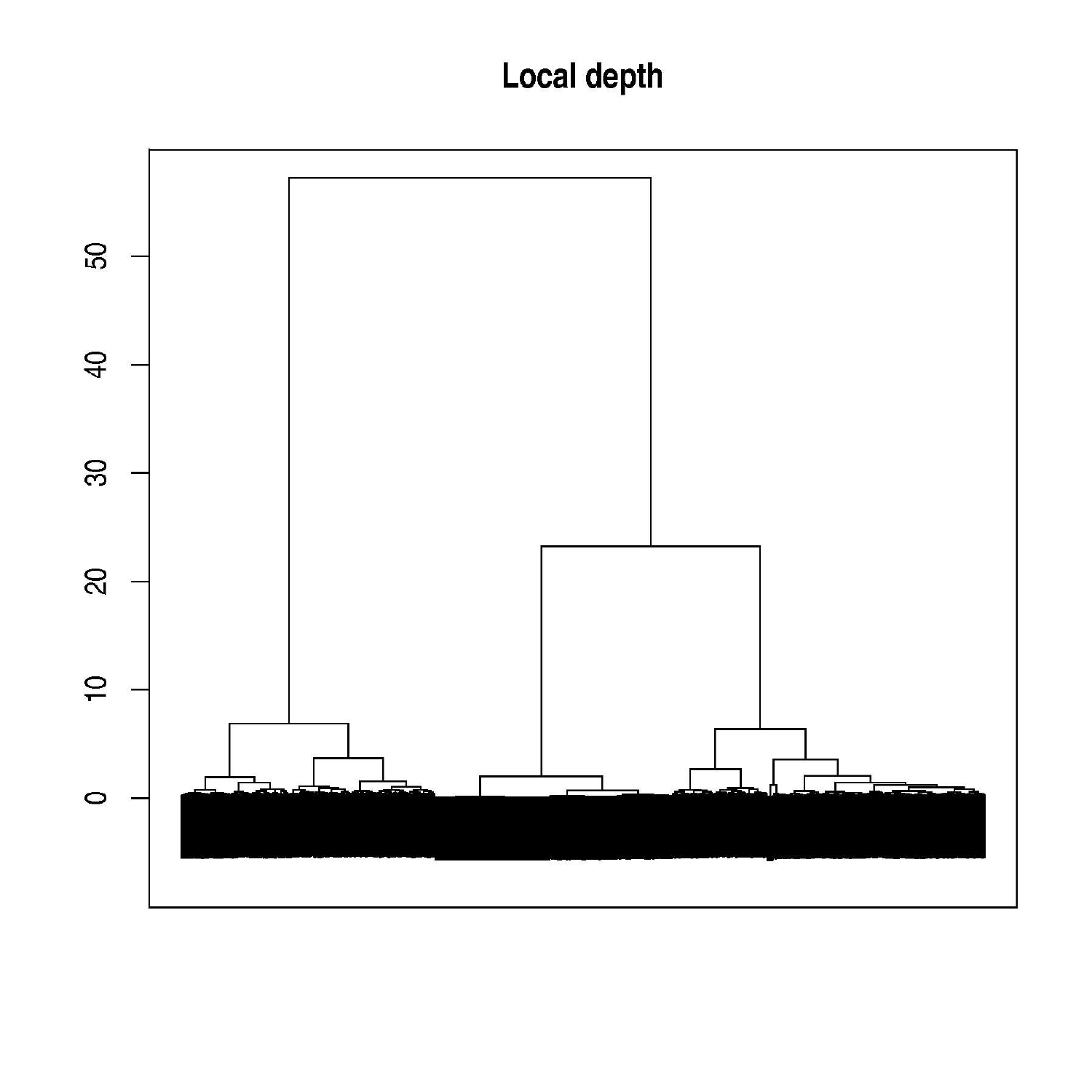}
  
  \includegraphics[width=0.45\textwidth]{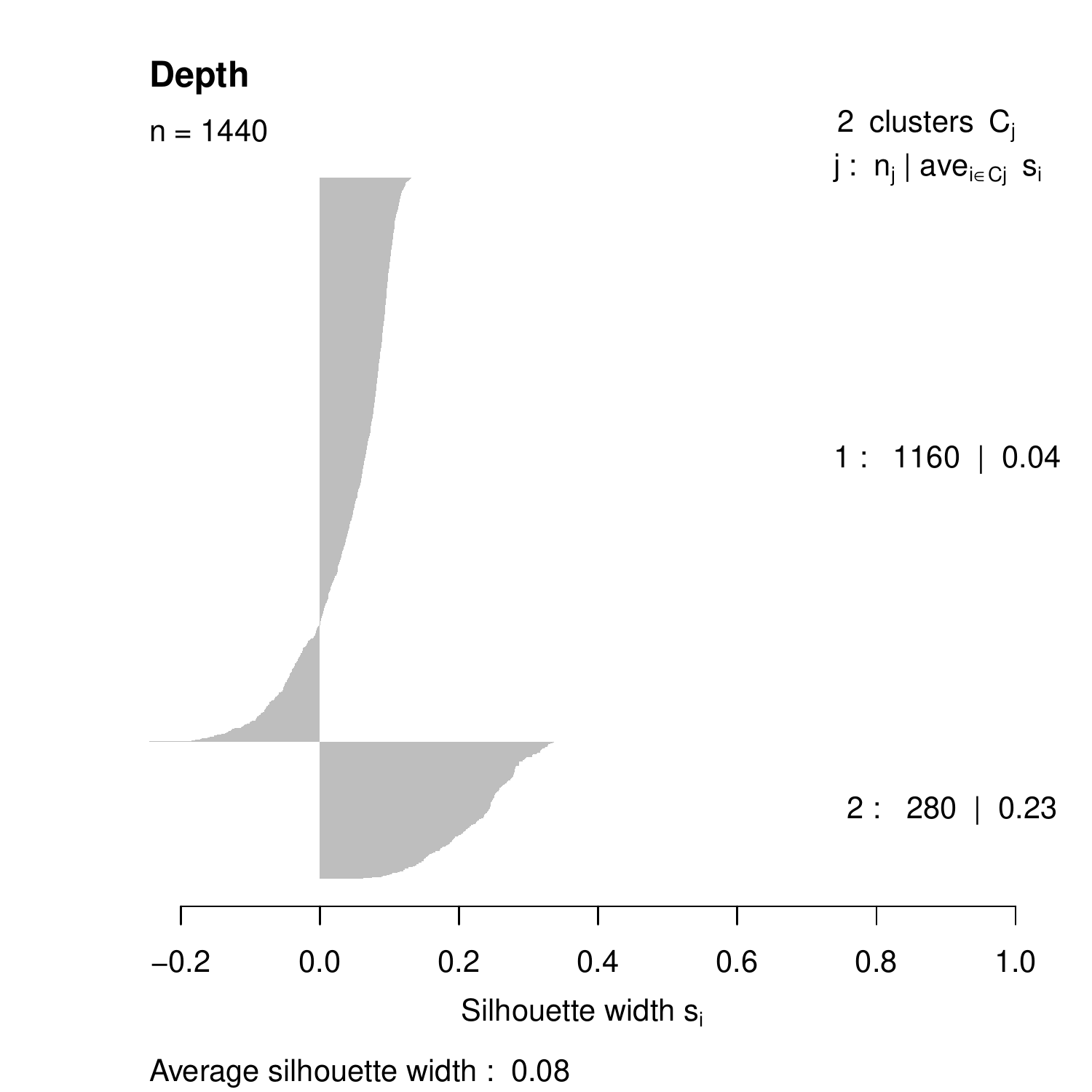}
  \includegraphics[width=0.45\textwidth]{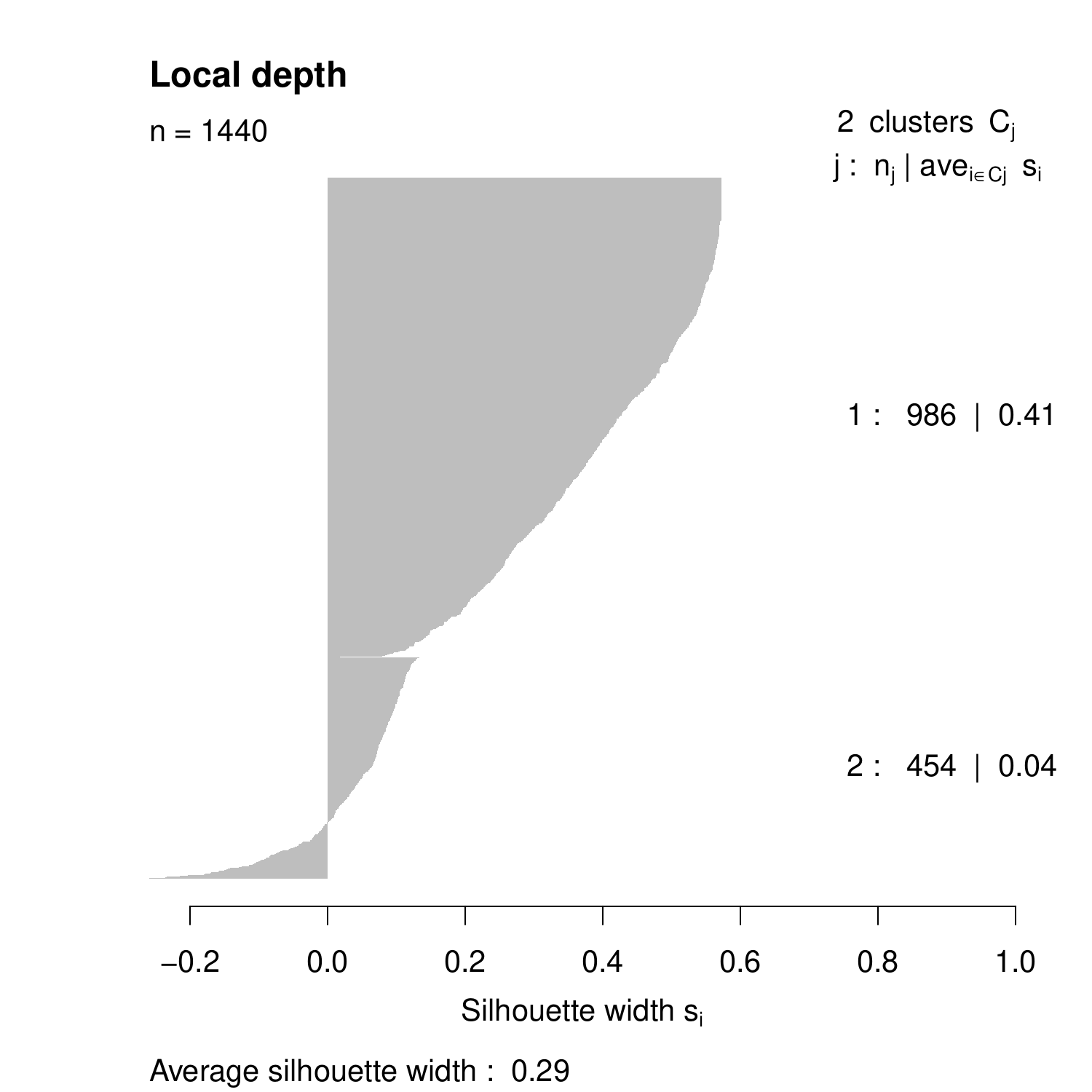}
  
\end{center}
\caption{Individual Household Electric Power Consumption. First row, dendograms, second row silhouette plot, first column modified half-region depth, second column local modified half-region depth.}
\label{fig_hpc_den_sil}
\end{figure}  

\clearpage

\begin{figure}
\begin{center}
  \includegraphics[width=0.45\textwidth]{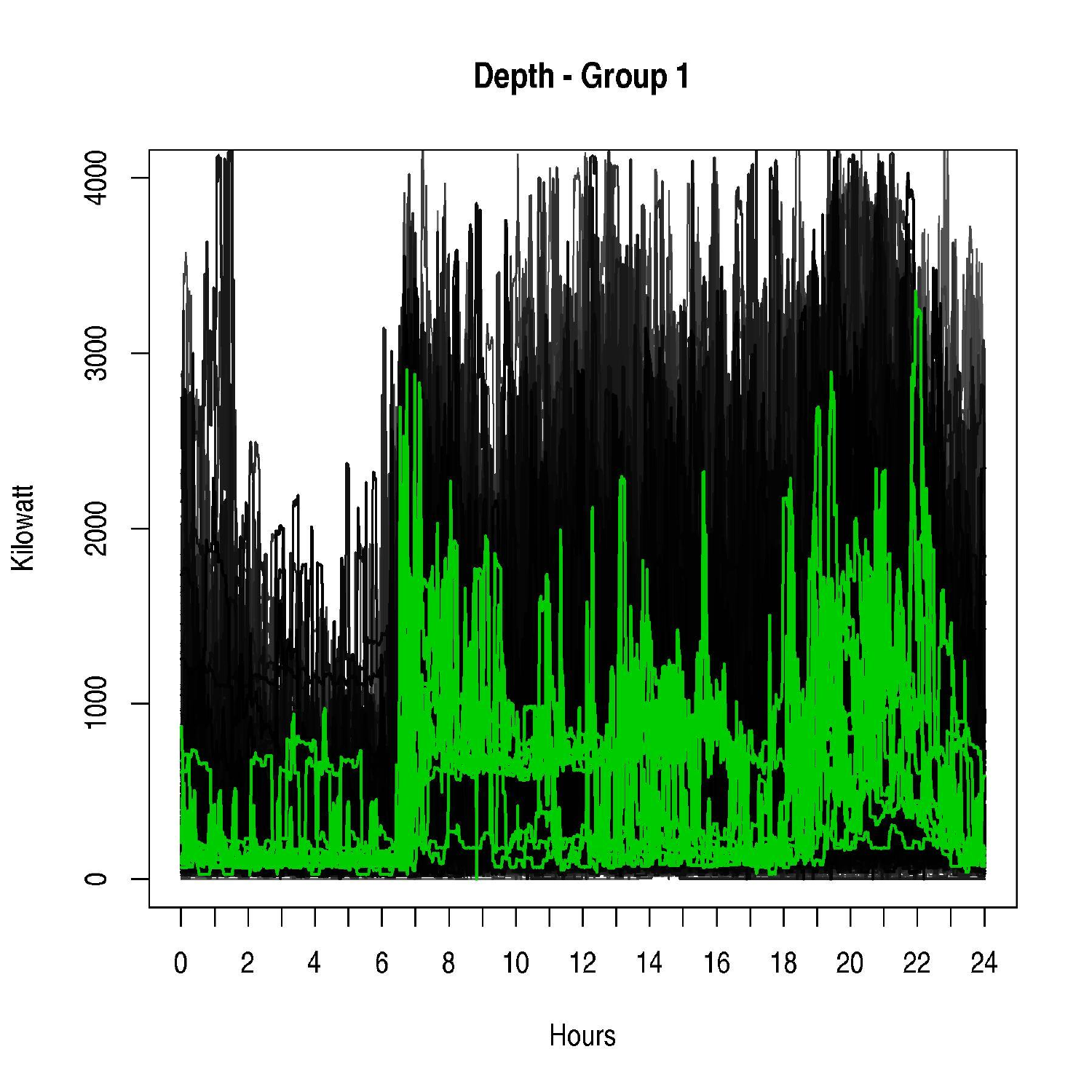}
  \includegraphics[width=0.45\textwidth]{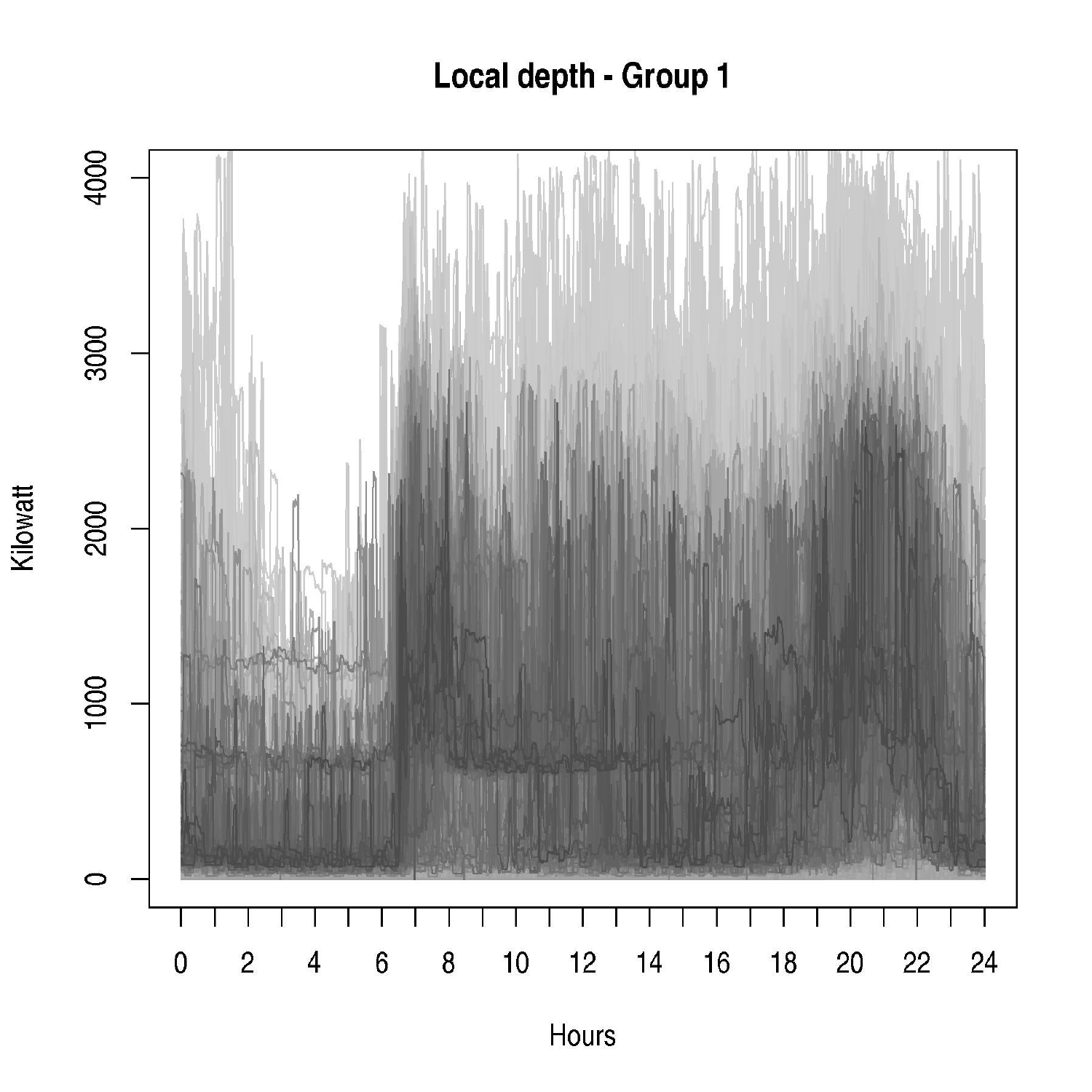}
  
  \includegraphics[width=0.45\textwidth]{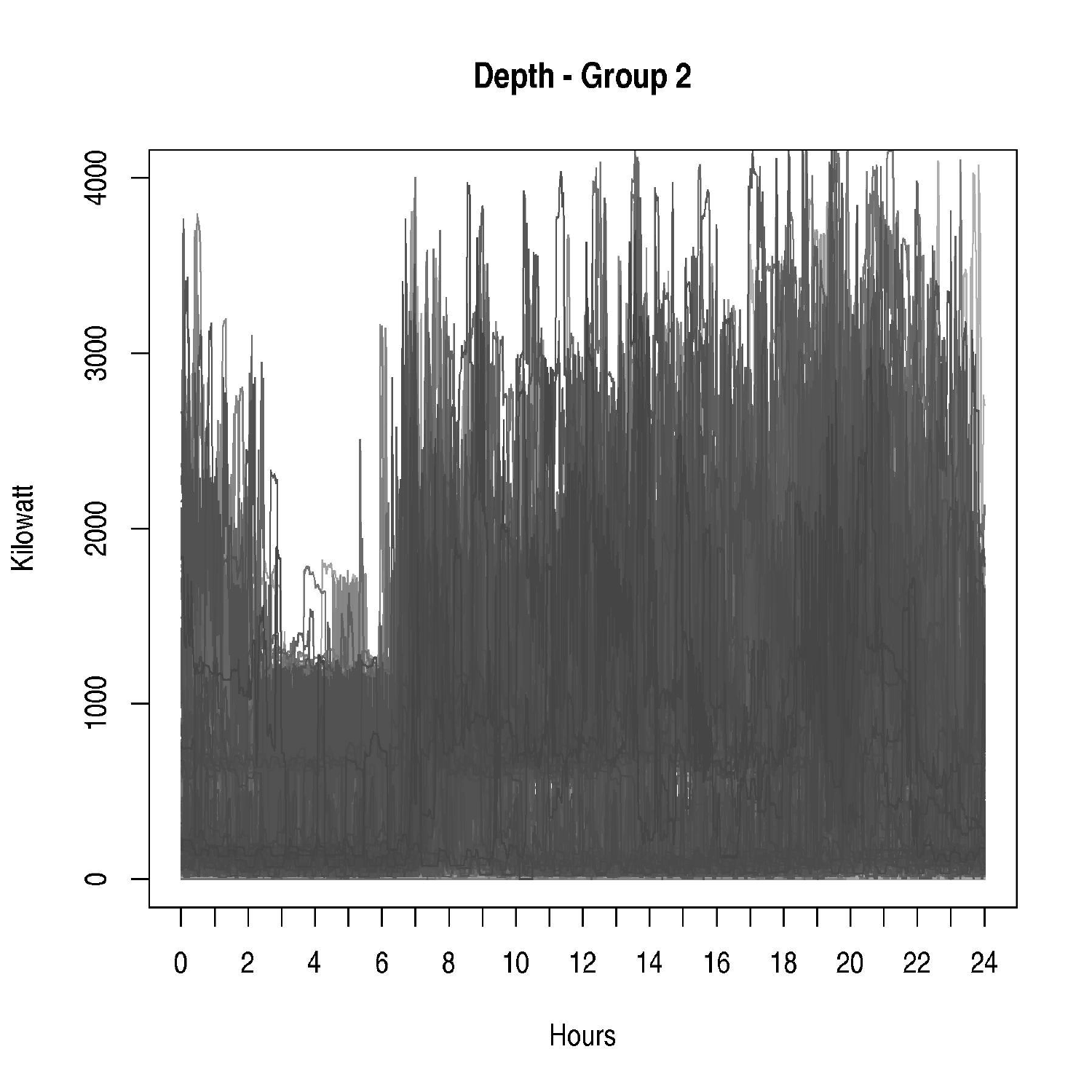}
  \includegraphics[width=0.45\textwidth]{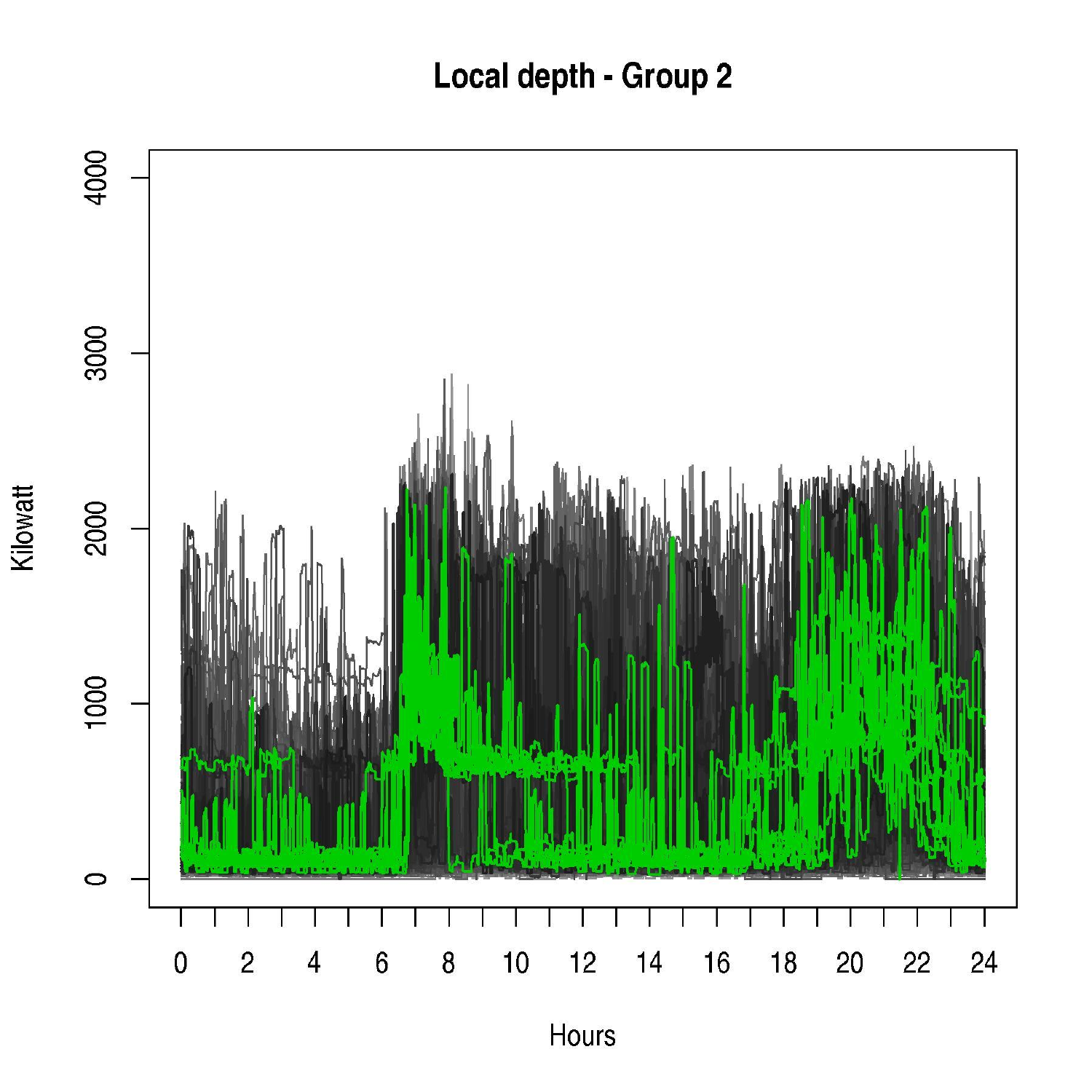}
  
\end{center}
\caption{Individual Household Electric Power Consumption. Groups provided by a cluster analysis based on modified half-region depth similarity (first column) and local modified half-region depth similarity (second column). The curves are plotted with color and thickness according to their depth/local depth. Green curves are the higher rank curves.}
\label{fig_hpc_clusters}
\end{figure}  

\clearpage

\section{out5d}
\label{sec:out5d}

%%%%Worcester Polytechnic Institute wpi
In this section we explore the out5d data set available at davis.wpi.edu/\-xmdv/\-datasets/\-out5d.html. It is a collection of 16384 observations along five
dimensions collected from remote sensing devices in a particular region in Western Australia on a $128\times128$ grid. The 5 variables are Spot, Magnetics and  3 bands of radiometrics which highlight Potassium, Thorium and Uranium. We concentrate on the local modified half-region depth method and we illustrate the use of the \texttt{R} package \texttt{ldfun} available upon request to the author. We first load the required packages
\begin{Schunk}
\begin{Sinput}
> require(ldfun)
> require(cluster)
\end{Sinput}
\end{Schunk}
and the data set from the original source
\begin{Schunk}
\begin{Sinput}
> url <- "http://davis.wpi.edu/~xmdv/datasets/out5d.tar.gz"
> download.file(url, destfile="./out5d.tar.gz")
> untar("out5d.tar.gz")
> out5d <- read.table("./out5d.okc", header=FALSE, skip=7)
> colnames(out5d)<-c("Spot", "Magnetic", "Potassium", "Thorium", "Uranium")
\end{Sinput}
\end{Schunk}
The next function will be used to normalize the (local) depth in the interval $[0,1]$ 
\begin{Schunk}
\begin{Sinput}
> normalize <- function(x) {
+   x <- (x - min(x))/(max(x) - min(x))
+   return(x)
+ }
\end{Sinput}
\end{Schunk}
The function \texttt{quantile.localdepth.functional} evaluates the distance between to curves using as default the sup norm, and return specific quantiles of its distribution. If \texttt{size=TRUE}, all the distances are also reported. The argument \texttt{byrow=TRUE} means that the observations are by rows and variables are on the columns, this is the default setting.  
\begin{Schunk}
\begin{Sinput}
> tau <- quantile.localdepth.functional(out5d, 
+          probs=c(0.1, 0.2, 0.3, 0.4, 0.05), byrow=TRUE, size=TRUE)
\end{Sinput}
\end{Schunk}
Since, we have several observations, we approximated the empirical distribution functions by plotting the percentile of the distribution, see Figure \ref{fig_out5d_sizes}.
\begin{Schunk}
\begin{Sinput}
> plot(seq(0,1,0.01), quantile(tau$stats, probs=seq(0,1,0.01)), 
+      type="l", xlab="Size", 
+      ylab="Empirical Cumulative Distribution Funtion")
\end{Sinput}
\end{Schunk}
\begin{figure}
\begin{center}
  \includegraphics[width=0.45\textwidth]{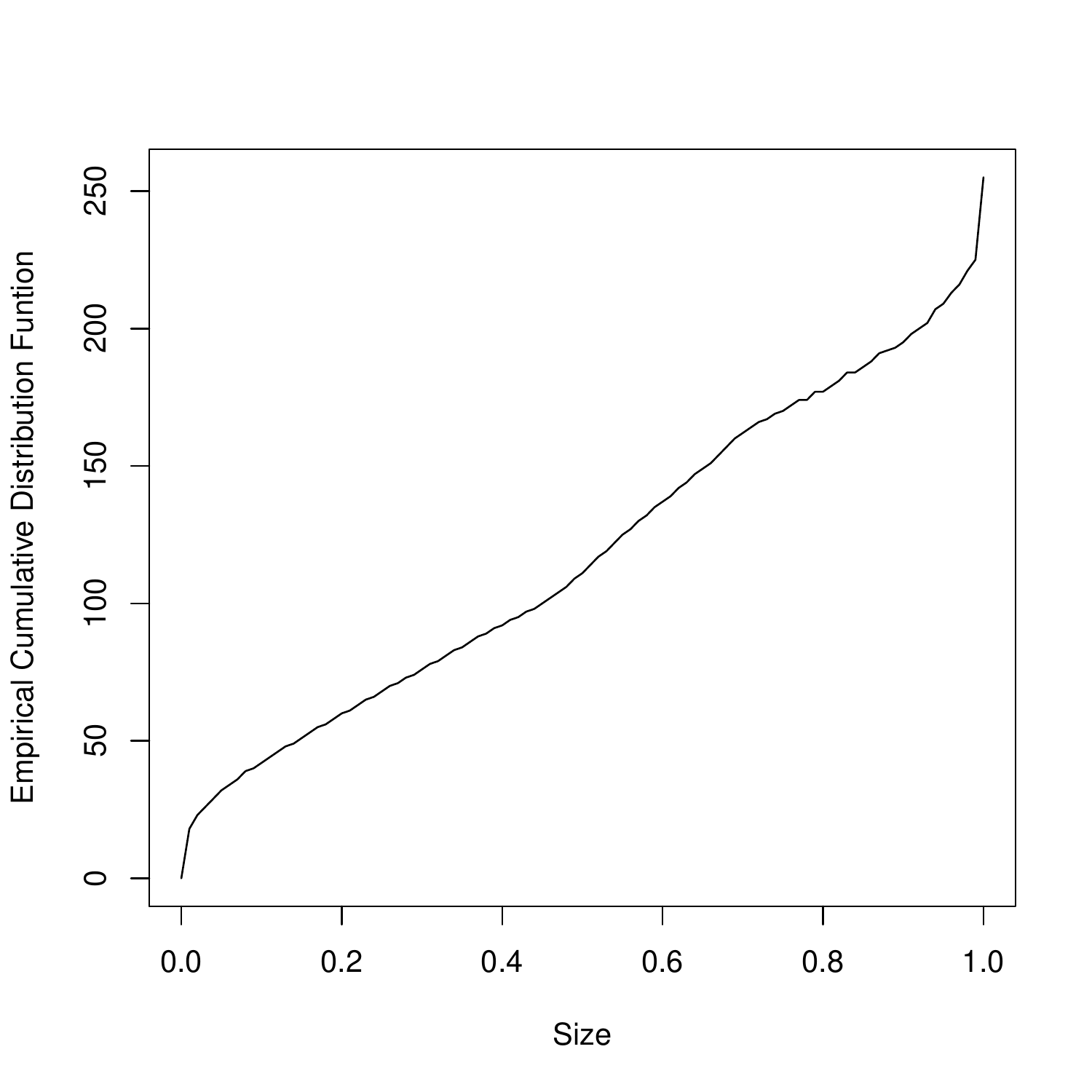}
\end{center}
\caption{out5d. Empirical cumulative distribution function of the distances between two curves.}
\label{fig_out5d_sizes}
\end{figure}
To explore the data set we set $\tau$ to be the $30\%$ quantile order of the distances distribution. The modified half-region detph and the local modified half-region depth are evaluated in the same call as
\begin{Schunk}
\begin{Sinput}
> mhr3 <- localdepth.modhalfregion(x=out5d, 
+           tau=tau$quantile[3], byrow=TRUE)
\end{Sinput}
\end{Schunk}
the resulting object is of class \texttt{localdepth} and it is well integrated with the methods available in the CRAN package \texttt{localdepth}.
Others depth measures are available in the package, see its documentation. DD-plot is easy to obtain as follows (Figure \ref{fig_out5d_dd})
\begin{Schunk}
\begin{Sinput}
> plot(mhr3)
> abline(0,1,lty=2)
\end{Sinput}
\end{Schunk}
\begin{figure}
\begin{center}
  \includegraphics[width=0.45\textwidth]{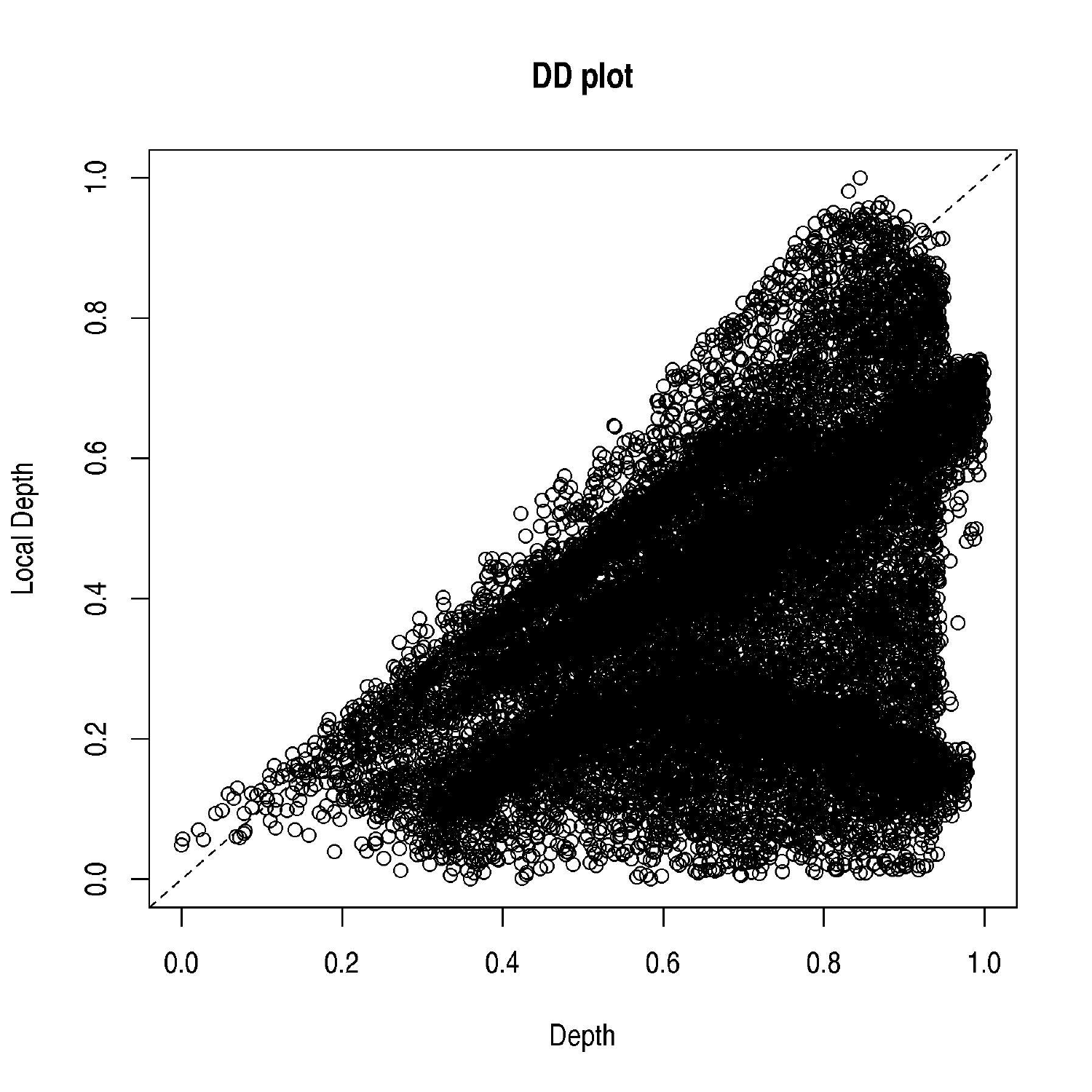}
\end{center}
\caption{out5d. DD plot, local modified half-region depth versus modified half-region depth.}
\label{fig_out5d_dd}
\end{figure}  
and a parallel plot which highlight the higher rank curves by colors (Figure \ref{fig_out5d_parallel_local3}) is obtained by
\begin{Schunk}
\begin{Sinput}
> parallelplot(mhr3, lattice=TRUE)
\end{Sinput}
\end{Schunk}
If the argument \texttt{lattice} is \texttt{TRUE} then the function \texttt{parallelplot} from the \texttt{lattice} package is used, otherwise the standard \texttt{matplot} function is called. To plot the ranks obtained by the depth, the argument \texttt{type} should be set accordingly.
\begin{figure}
\begin{center}
  \includegraphics[width=0.45\textwidth]{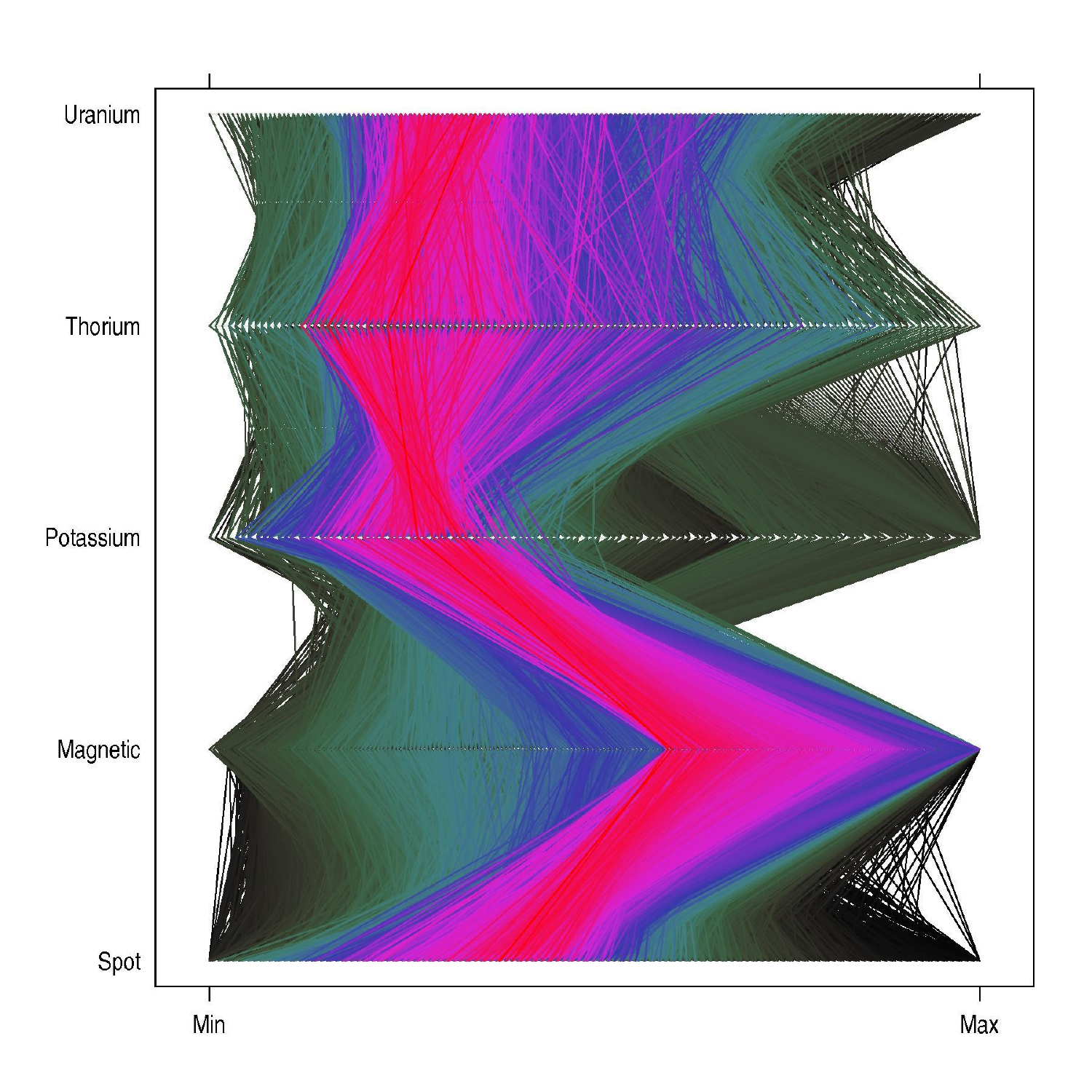}
\end{center}
\caption{out5d. Parallel plot, ranks are according to local modified half-region depth; darker means lower ranks, redder higher ranks.}
\label{fig_out5d_parallel_local3}
\end{figure}
The plot point out the presence of a dense group of curves. As we will see later on, this group is characterize by low levels of Thorium and Uranium, and intermediate values of the others. A different representation is obtained by the function \texttt{pairs} after rearraging the observations according to their ranks and preparing a suitable palette colors, see Figure \ref{fig_out5d_pairs_local3},
\begin{Schunk}
\begin{Sinput}
> omhr3 <- order(mhr3$localdepth)
> oout5d <- out5d[omhr3,]
> colmhr3 <- hsv(h = normalize(mhr3$localdepth), 
+                s = normalize(mhr3$localdepth), 
+                v = normalize(mhr3$localdepth), alpha = 1)
> ocol <- colmhr3[omhr3]
\end{Sinput}
\end{Schunk}
\begin{Schunk}
\begin{Sinput}
> pairs(oout5d, col=ocol, pch=20)
\end{Sinput}
\end{Schunk}
\begin{figure}
\begin{center}
  \includegraphics[width=0.45\textwidth]{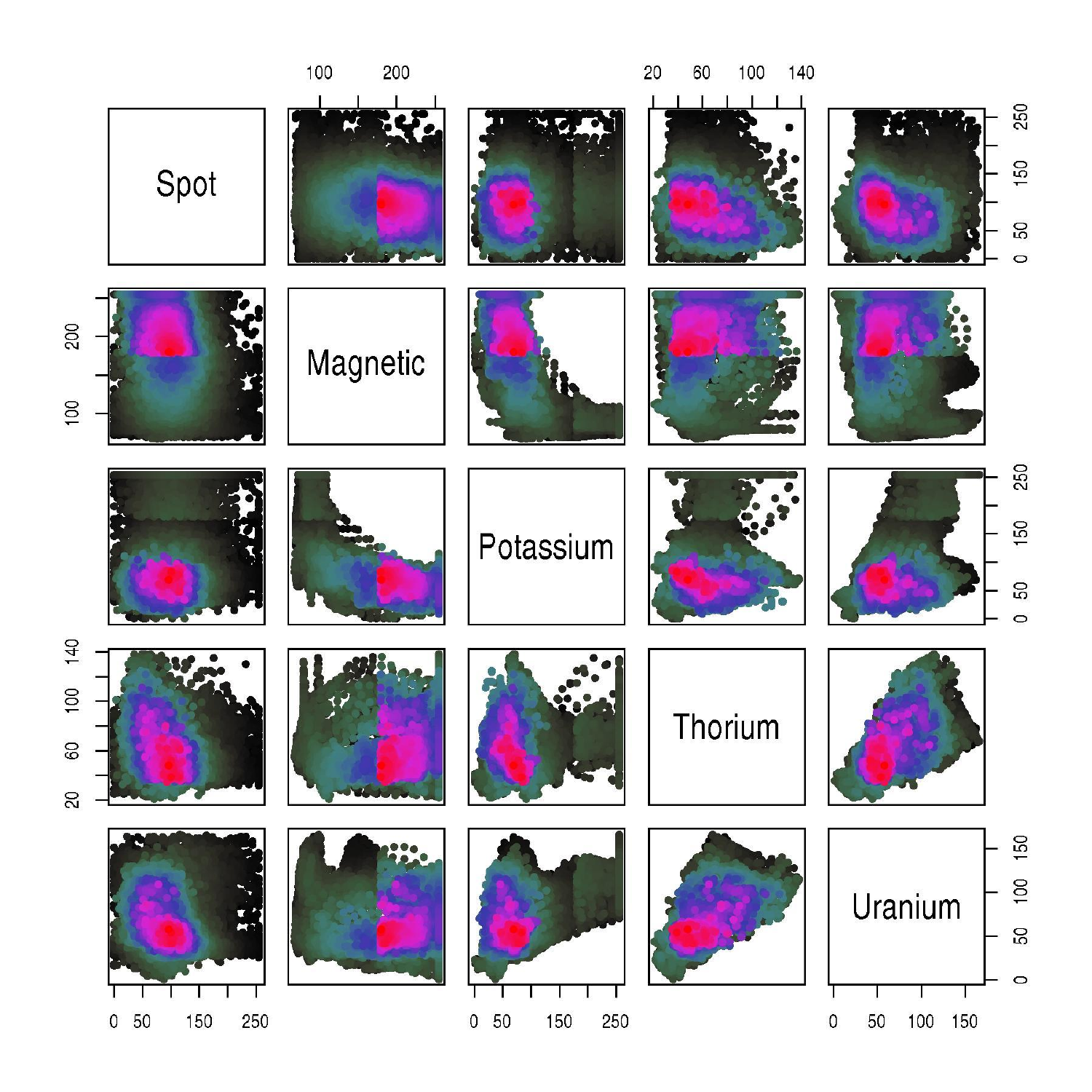}
\end{center}
\caption{out5d. Pairs plot, ranks are according to local modified half-region depth; darker means lower ranks, redder higher ranks.}
\label{fig_out5d_pairs_local3}
\end{figure}
We now turn our attention to the clustering problem. The function \texttt{localdepth.similarity.modhalfregion} provides the similarity measures based on both depth and local depth. The interface is very similar to the \texttt{localdepth.modhalfregion.functional} function, however we advise the user that this function run much slower and require a large amount of memory since it has to evaluate the (local) depth for $n \times (n+1) / 2$ points, hence for the sample sizes of the out5d data set we suggest its use in batch
\begin{Schunk}
\begin{Sinput}
> smhr3 <- localdepth.similarity.modhalfregion(x=out5d, 
+            tau=tau$quantile[3], byrow=TRUE)
\end{Sinput}
\end{Schunk}
The function \texttt{similarity2dissimilarity} will convert a similarity matrix in a dissimilarity matrix using the tranformation suggested in \citet{gower_1966}, then the result is used as input for the \texttt{hclust} function
\begin{Schunk}
\begin{Sinput}
> dissimilarity <- similarity2dissimilarity(smhr3$localdepth)
> mhrD <- hclust(as.dist(dissimilarity), method="ward")
\end{Sinput}
\end{Schunk}
In the usual way we obtain the dendrogram (Figure \ref{fig_out5d_den_sil} left panel)
\begin{Schunk}
\begin{Sinput}
> plot(mhrD, frame.plot=TRUE, labels=FALSE, ann=FALSE, main="Local depth")
> title("Local depth")
\end{Sinput}
\end{Schunk}
and the Silhouette plot (Figure \ref{fig_out5d_den_sil} right panel)
\begin{Schunk}
\begin{Sinput}
> groups <- cutree(mhrD, k=3)
> plot(silhouette(groups, dist=dissimilarity), main="Local depth")
> ##plot(s3)
\end{Sinput}
\end{Schunk}
\begin{figure}
\begin{center}
  \includegraphics[width=0.45\textwidth]{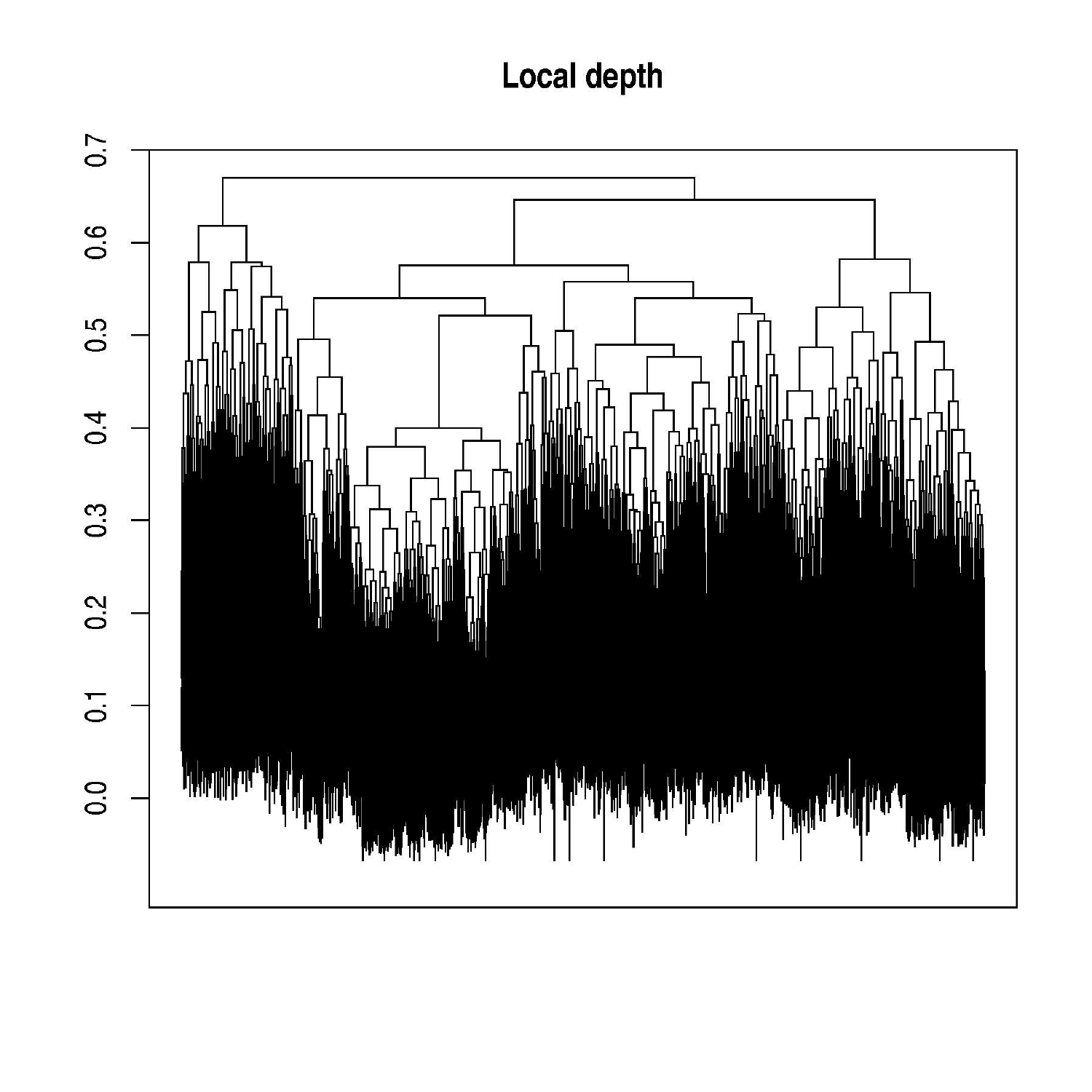}
  \includegraphics[width=0.45\textwidth]{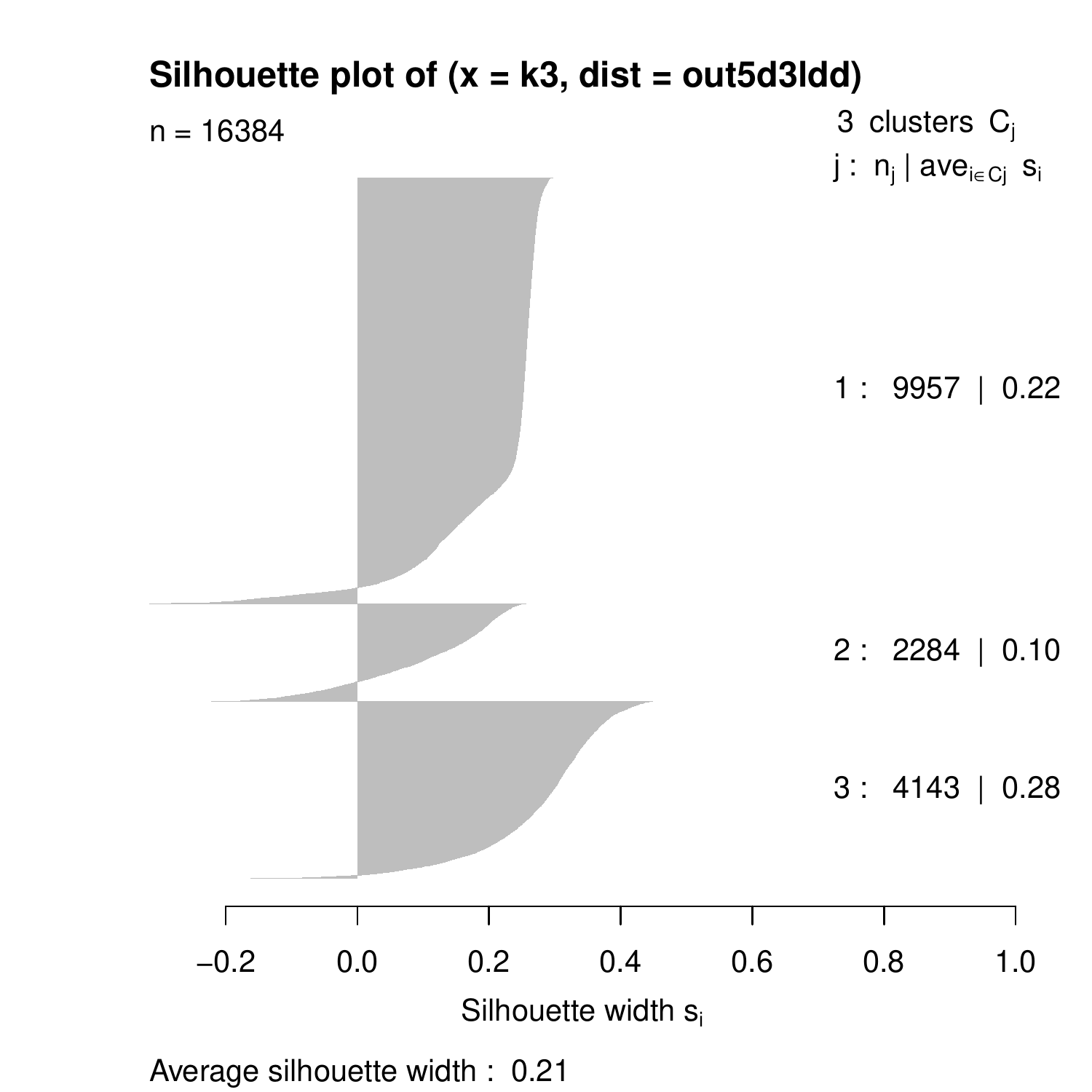}  
\end{center}
\caption{out5d. Dendogram and silhouette plot using dissimilarity obtained by local modified half-region depth.}
\label{fig_out5d_den_sil}
\end{figure}  
After inspection of the dendrogram and several silhouette plots we decide to stay with $3$ groups, to better characterize these groups we evaluate summary statistics on location and dispersion by 
\begin{Schunk}
\begin{Sinput}
> sout5d <- split(out5d, groups)
> lapply(sout5d, summary)
> lapply(sout5d, function(x) apply(x, 2, 
+                  function(y) diff(quantile(y, probs=c(0.25,0.75)))))
> lapply(sout5d, function(x) apply(x, 2, function(y) sd(y)))
\end{Sinput}
\end{Schunk}
some of these results are summarized in Tables \ref{tab:out5d3-cluster-median} and \ref{tab:out5d3-cluster-mean}.

\begin{table}
\begin{tabular}{lrrrrr}
\hline
Group ID & Spot & Magnetic & Potassium & Thorium & Uranium \\
\hline
1 & 100 (78) & 104 (74) & 90 (186) & 58 (29) & 55 (71) \\
2 &  90 (47) & 208 (28) & 64 (23) & 55 (20) & 51 (17) \\
3 &  62 (34) & 255 (0) & 60 (28) & 73 (25) & 82 (37) \\
\hline
\end{tabular}
\caption{Out5d data set. Median and Inter Quantile Range (in round parethesis) for the three groupes generated by the similarity based on local modified half-region depth}
\label{tab:out5d3-cluster-median}
\end{table}

\begin{table}
\begin{tabular}{lrrrrr}
\hline
Group ID & Spot & Magnetic & Potassium & Thorium & Uranium \\
\hline
1 & 106 (53) & 122 (46) & 132 (85) & 61 (18) & 72 (39) \\
2 &  89 (31) & 207 (17) &  62 (17) & 59 (17) & 56 (20) \\
3 &  68 (26) & 253 (7)  &  60 (19) & 75 (17) & 82 (23) \\
\hline
\end{tabular}
\caption{Out5d data set. Mean and Standard Deviation (in round parethesis) for the three groupes generated by the similarity based on local modified half-region depth}
\label{tab:out5d3-cluster-mean}
\end{table}

The first group is characterized by high levels of Spot and Potassium and low levels of Magnetic. The third group has extreme values of Magnetic, high levels of Thorium and Uranium, while low levels of Spot and Potassium. The second group presents somehow intermediate levels, with lower levels of Thorium and Uranium. Figure \ref{fig_out5d_parallel_cluster_local3} reports the parallel plots of the $3$ groups where we can appreciate some of these characteristics
\begin{Schunk}
\begin{Sinput}
> myPanel <- function(x, col, ok, ...) {
+   pp <- which.packet()
+   panel.parallel(x, col=col[ok==pp], ...)
+ }
> ogroups <- groups[omhr3]
> parallelplot(~oout5d | factor(ogroups), panel=myPanel, 
+              ok=ogroups, col=ocol, scales=list(x="same"))
\end{Sinput}
\end{Schunk}
\begin{figure}
\begin{center}
  \includegraphics[width=0.45\textwidth]{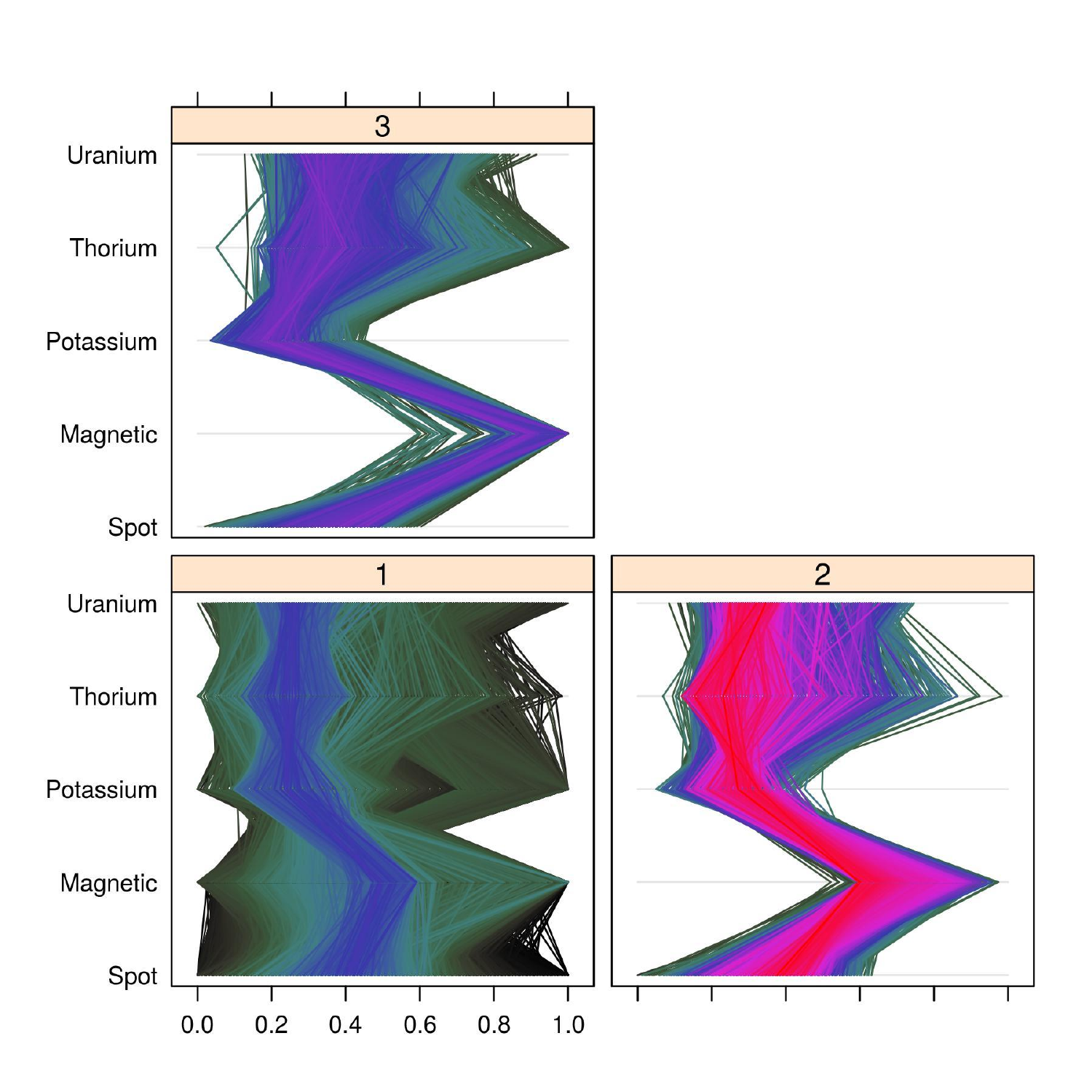}
\end{center}
\caption{out5d. Parallel plot for the $3$ groups by local modified half-region depth.}
\label{fig_out5d_parallel_cluster_local3}
\end{figure}  
To better highlight differences among the groups, we used the following code repeated for each variable
\begin{Schunk}
\begin{Sinput}
> plot(ecdf(sout5d[[1]]$Spot), pch="", verticals=TRUE, main="Spot")
> lines(ecdf(sout5d[[2]]$Spot), pch="", col=2, verticals=TRUE)
> lines(ecdf(sout5d[[3]]$Spot), pch="", col=3, verticals=TRUE)
\end{Sinput}
\end{Schunk}
which produces Figure \ref{fig_out5d_ecdf_groups_local3} on the Empirical cumulative distribution function of each variable conditional to group membership. These plots confirm our findings, the first group is characterized by high values of Spot and Potassium and low values of Magnetic, while the third group has very high levels of Magnetic, high levels of Thorium and low levels of Spot. 
\begin{figure}
\begin{center}
  \includegraphics[width=0.45\textwidth]{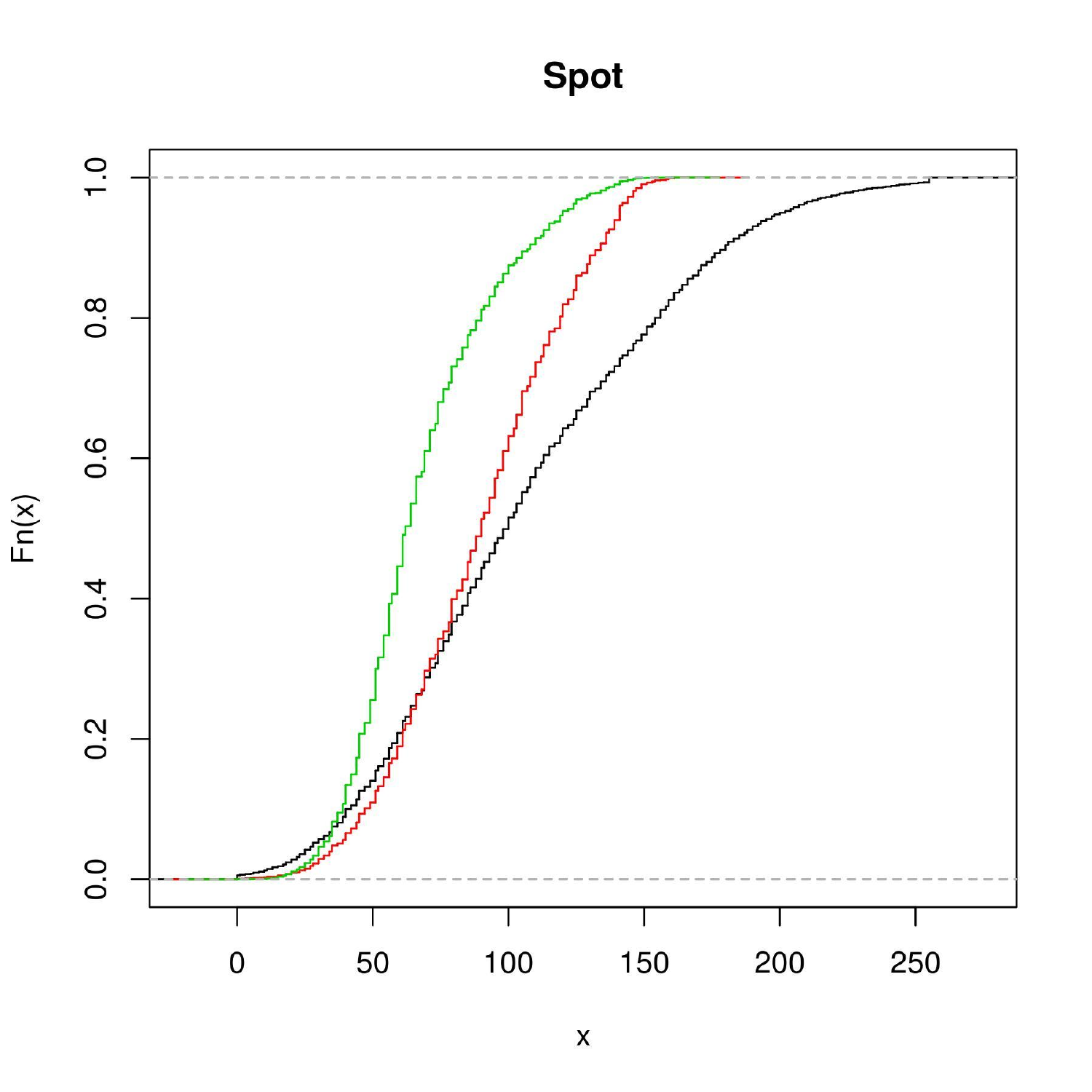}
  \includegraphics[width=0.45\textwidth]{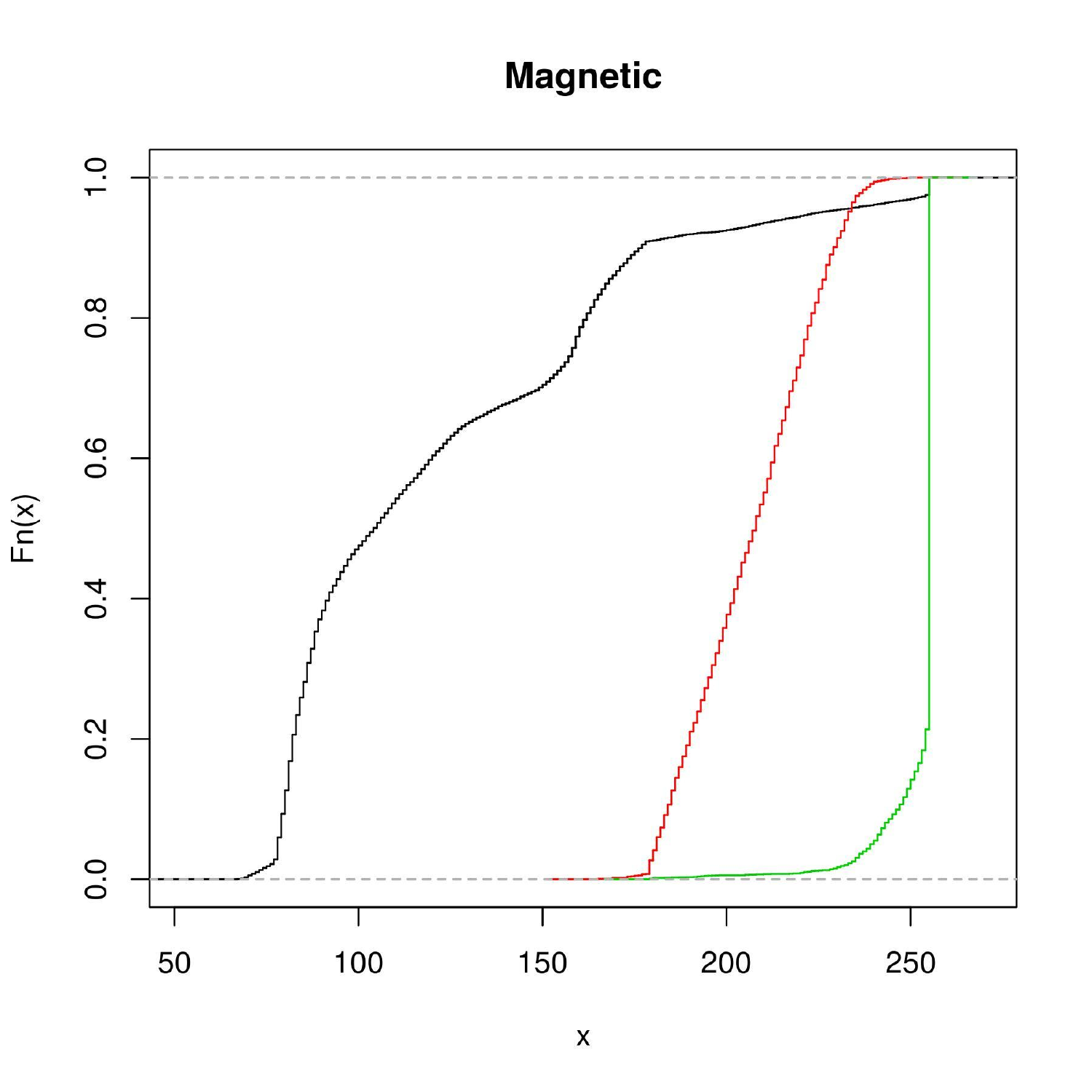}
  
  \includegraphics[width=0.45\textwidth]{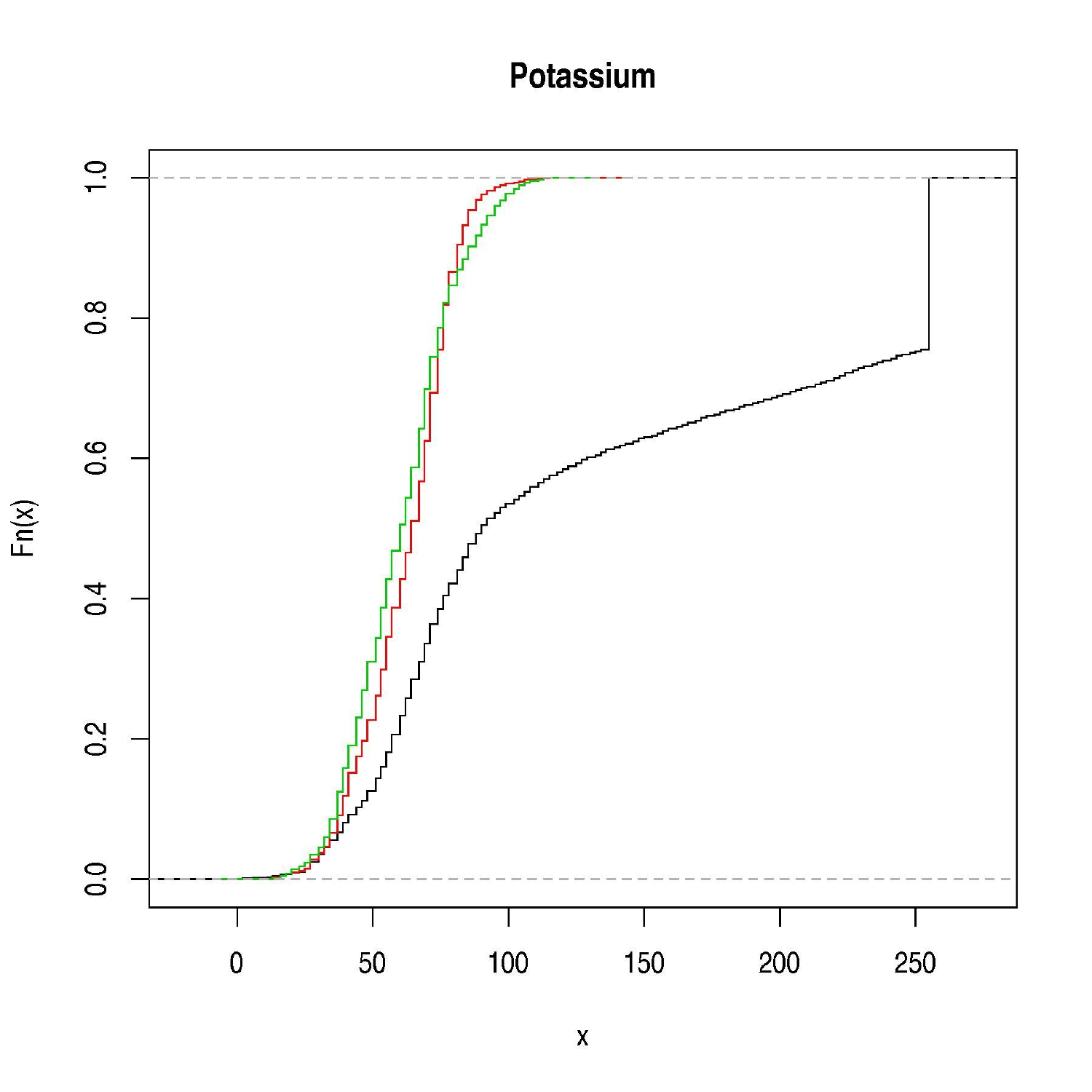}
  \includegraphics[width=0.45\textwidth]{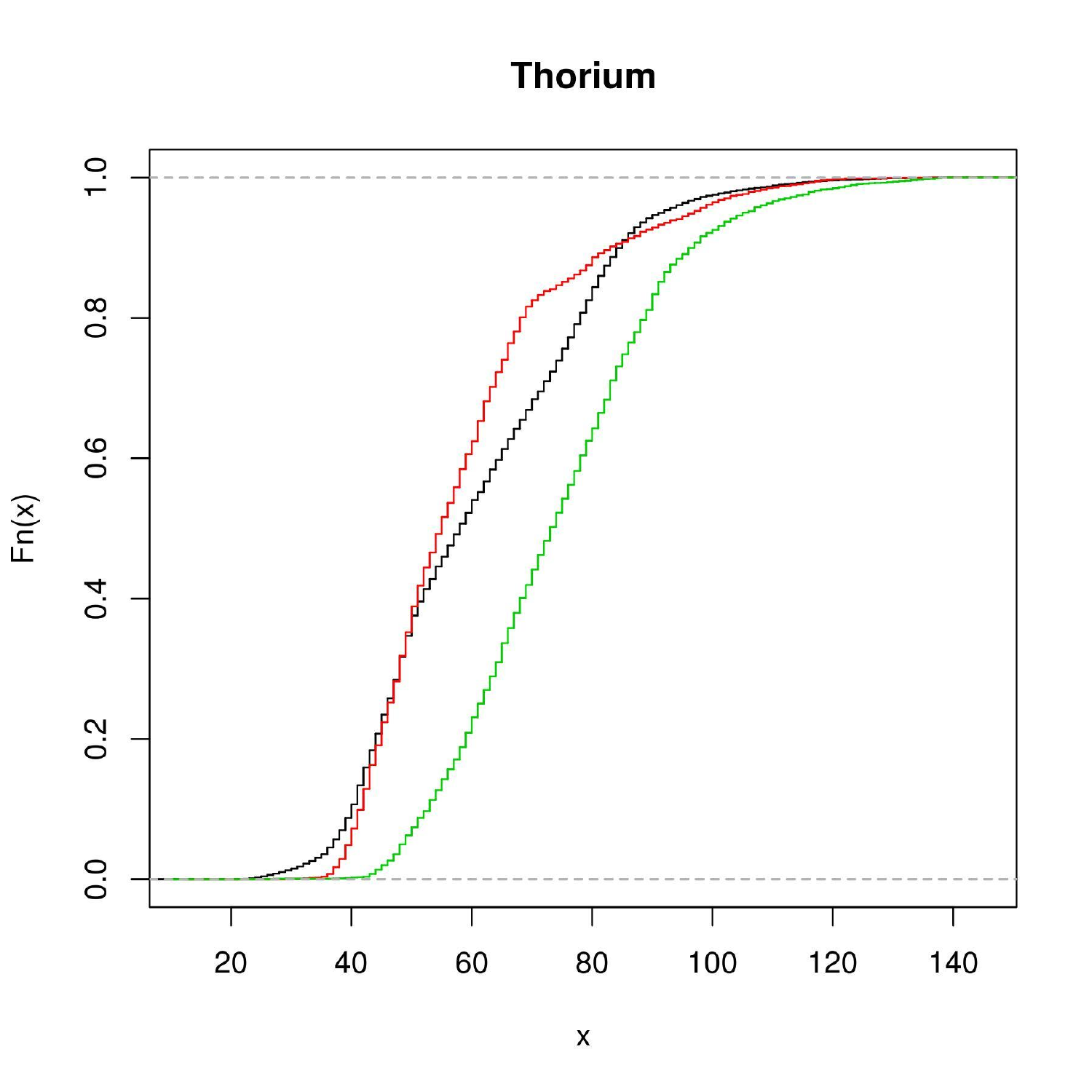}
  
  \includegraphics[width=0.45\textwidth]{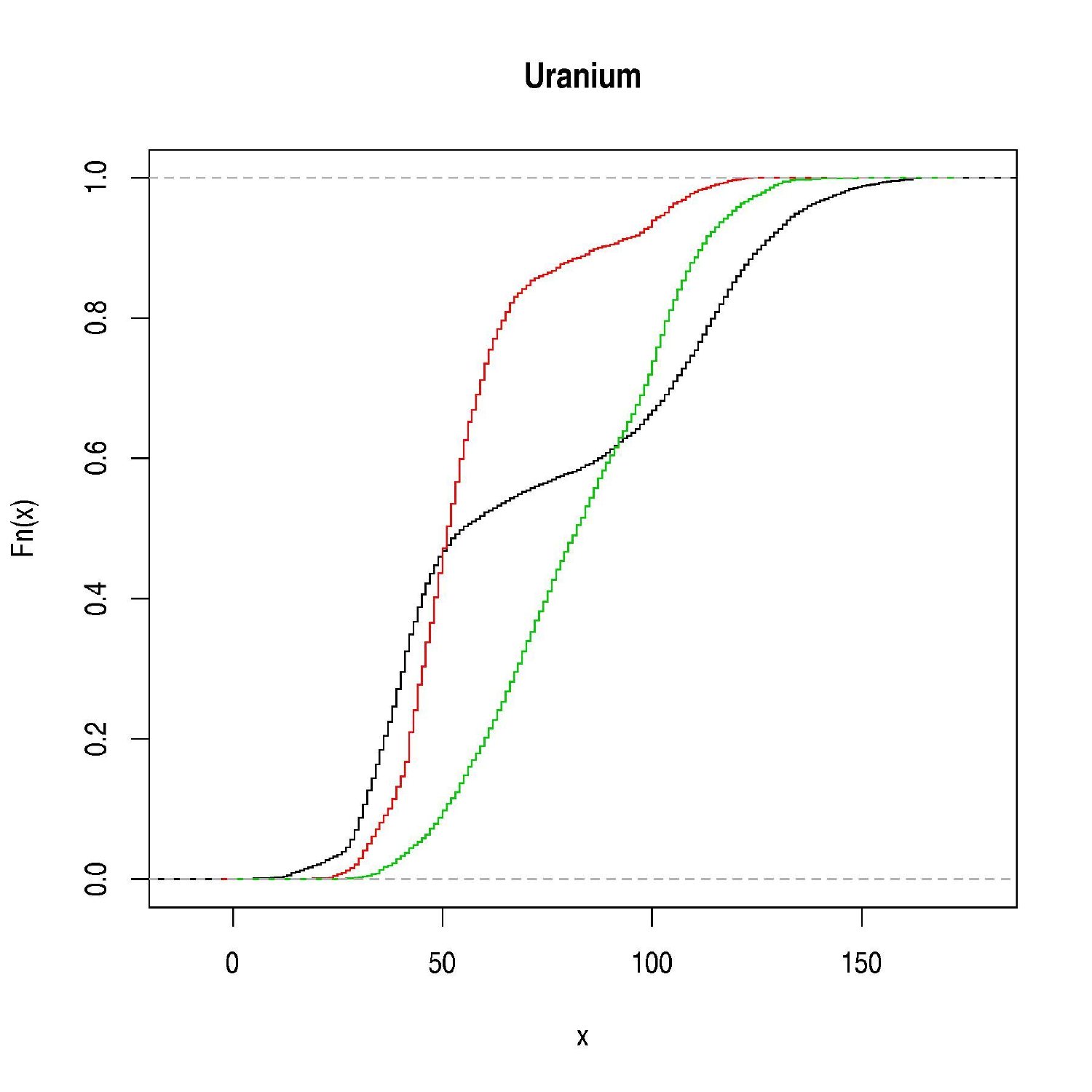}  
\end{center}
\caption{out5d. Empirical cumulative distribution function of each variable conditional to the group. Black: group $1$, red: group $2$ and green: group $3$.}
\label{fig_out5d_ecdf_groups_local3}
\end{figure}  
Finally, Figure \ref{fig_out5d_boxplot_cluster_local3} obtained with the following code
\begin{Schunk}
\begin{Sinput}
> out5dlattice <- data.frame(Measure=unlist(out5d), 
+   Cluster=rep(groups, times=ncol(out5d)), 
+   Variable=rep(c("Spot", "Magnetic", "Potassium", "Thorium", "Uranium"), 
+   each=nrow(out5d)))
> bwplot(Variable~Measure | factor(Cluster), data=out5dlattice)
\end{Sinput}
\end{Schunk}
presents box-plot of the variables conditional on group membership. 

\begin{figure}
\begin{center}
  \includegraphics[width=0.8\textwidth]{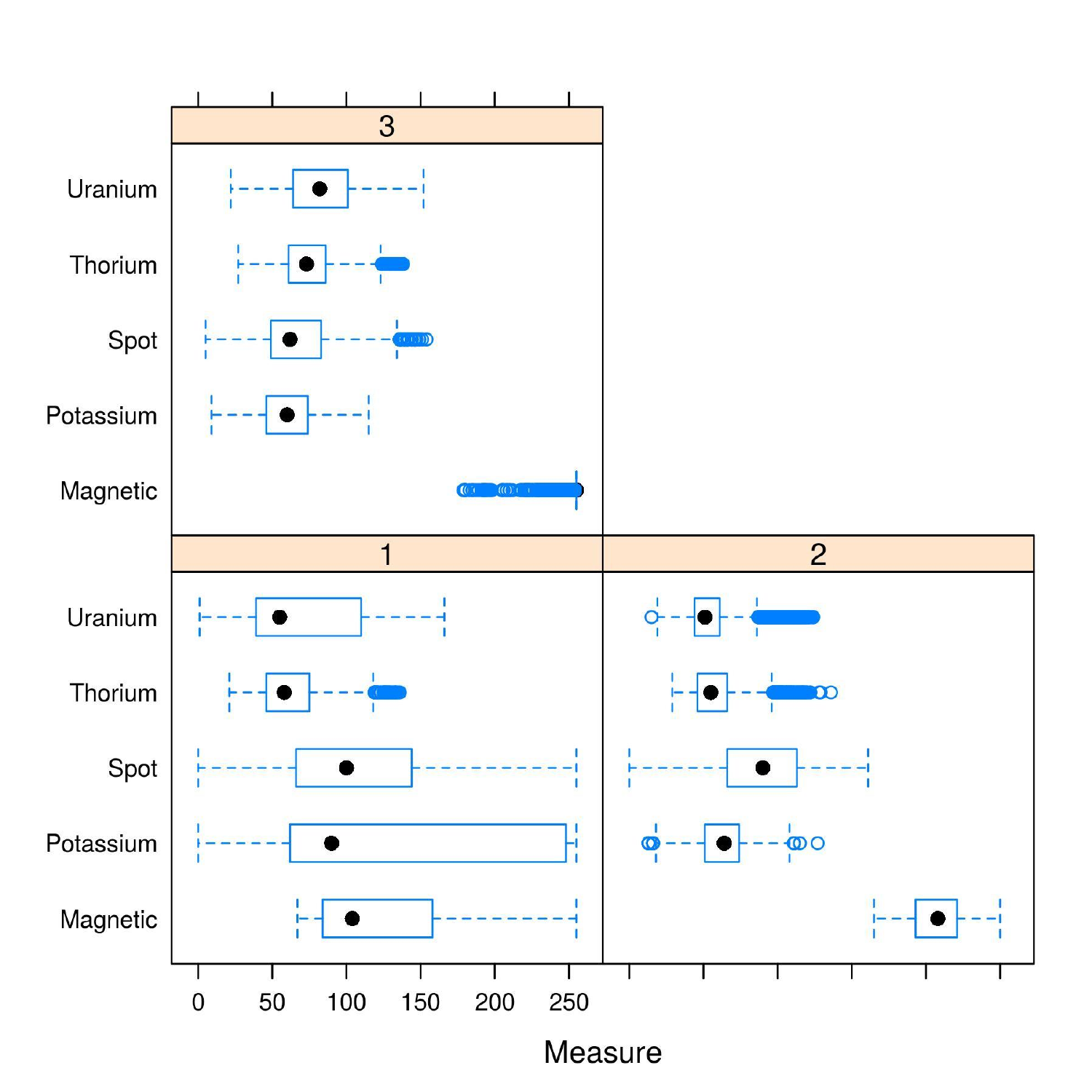}
\end{center}
\caption{out5d. Boxplot of the variables conditional on group membership.}
\label{fig_out5d_boxplot_cluster_local3}
\end{figure}  

\clearpage

\section{Proofs}
\label{sec:proofs}

\begin{proof}[Proof of Proposition \ref{PropLocalHalfRegionTauFinite}]
\begin{enumerate}
\item[i.] The result follows since $\SL_{e_j}(x(j),x(j)+\tau_1(j)) \subseteq \SL_{e_j}(x(j),x(j)+\tau_2(j))$ for all $j=1, \cdots, p$ which implies $\Pr_{\mathbf{X}} \left( \bigcap_{j=1}^p \SL_{e_j}(x(j),x(j)+\tau_1(j)) \right) \le \Pr_{\mathbf{X}} \left( \bigcap_{j=1}^p \SL_{e_j}(x(j),x(j)+\tau_2(j)) \right)$. The same holds for the other component.
\item[ii.] The first inequality holds since $\mathbf{x} \in \bigcap_{j=1}^p \SL_{e_j}(x(j),x(j)+\tau(j))$ and $\mathbf{x} \in \bigcap_{j=1}^p \SL_{-e_j}(x(j),x(j)-\tau(j))$ for all vector $\boldsymbol{\tau}$ while the second inequality holds since $\SL_{e_j}(x(j),x(j)+\tau(j)) \subset HS_{e_j}(x(j))$ for all $j=1,\cdots, p$.
\item[iii.] The result holds since $\lim_{\min(\tau(j), j=1, \cdots, p) \rightarrow \infty} \SL_{e_k}(x(k),x(k)+\tau(k)) = HS_{e_k}(x(k))$ by definition for $k=1, \cdots, p$.
\item[iv.] we have 
\begin{equation*}  
\lim_{\| \tau \| \rightarrow 0} \SL_{e_j}(x(j),x(j)+\tau(j)) = \lim_{ \tau(j) \rightarrow 0} \SL_{e_j}(x(j),x(j)+\tau(j)) = l_j
\end{equation*}
where $l_j$ is the line parallel to the $e_j$ axes and passing trough the point $\mathbf{x}$. Similarly for the other components of the half-region depth. Then $\bigcap_{j=1}^p l_j = \mathbf{x}$ and hence the result.  
\end{enumerate}  
\end{proof}

\begin{proof}[Proof of Proposition \ref{PropHalfSpaceFinite}]
The result follows immediately from the specialization of equations (\ref{EquLHRDistributionFunction}) to the case $p=1$ and the definition of the local half-region depth. 
\end{proof}

\begin{proof}[Proof of Proposition \ref{PropVanishingFinite}]
The first result follows since Proposition \ref{PropLocalHalfRegionTauFinite} implies $ld_{HR}(\mathbf{x}; \mathbf{X}, \boldsymbol{\tau}) \le d_{HR}(\mathbf{x}; \mathbf{X})$ and Proposition 2 of \citet{lopez-pintado_half-region_2011}. Similar motivation leads to the second statement.
\end{proof}

\begin{proof}[Proof of Proposition \ref{PropSemicontinuousFinite}]
To prove that $ld_{HR}(\cdot; \mathbf{X}, \boldsymbol{\tau})$ is upper semicontinuous we show that $\lim\sup_{n\rightarrow \infty} ld_{HR}(\mathbf{x}_n; \mathbf{X}, \boldsymbol{\tau}) = ld_{HR}(\mathbf{x}; \mathbf{X}, \boldsymbol{\tau})$, when $\mathbf{x}_n \stackrel{\| \cdot \|_2}{\rightarrow} \mathbf{x}$. Let $A_j = A_j(\mathbf{x}; \boldsymbol{\tau}) = \SL_{-e_j}(x(j)-\tau(j),x(j))$ and $A_{n,j} = A_j(\mathbf{x}_n; \boldsymbol{\tau}) = \SL_{-e_j}(x_n(j)-\tau(j),x_n(j))$, and similarly $B_j = B_j(\mathbf{x}; \boldsymbol{\tau}) = \SL_{e_j}(x(j),x(j)+\tau(j))$ and $B_{n,j} = B_j(\mathbf{x}_n; \boldsymbol{\tau}) = \SL_{e_j}(x_n(j),x_n(j)+\tau(j))$.
\begin{align*}
\lim\sup_{n\rightarrow \infty} ld_{HR}(\mathbf{x}_n; \mathbf{X}, \boldsymbol{\tau}) & = \lim\sup_{n\rightarrow \infty} \min \left( \Pr_{\mathbf{X}} \left( \bigcap_{j=1}^p A_{n,j} \right), \Pr_{\mathbf{X}} \left( \bigcap_{j=1}^p B_{n,j} \right) \right) \\
& \le \min \left( \lim\sup_{n\rightarrow \infty} \Pr_{\mathbf{X}} \left( \bigcap_{j=1}^p A_{n,j} \right), \lim\sup_{n\rightarrow \infty} \Pr_{\mathbf{X}} \left( \bigcap_{j=1}^p B_{n,j} \right) \right) \\
& = \min \left( \Pr_{\mathbf{X}} \left( \bigcap_{j=1}^p A_j \right), \Pr_{\mathbf{X}} \left( \bigcap_{j=1}^p B_j \right) \right) \\
& = ld_{HR}(\mathbf{x}; \mathbf{X}, \boldsymbol{\tau})
\end{align*}
We will show that 
\begin{equation*}
\left| \Pr_{\mathbf{X}} \left( \bigcap_{j=1}^p A_{n,j} \right) - \Pr_{\mathbf{X}} \left( \bigcap_{j=1}^p A_j \right)\right| \rightarrow 0 \qquad n \rightarrow \infty
\end{equation*}
which, together with the same result for the sequence of $B_{n,j}$ is sufficient to establish the continuity of $ld_{HR}(\cdot; \mathbf{X}, \boldsymbol{\tau})$. In fact,
\begin{align*}
\left| \Pr_{\mathbf{X}} \left( \bigcap_{j=1}^p A_{n,j} \right) - \Pr_{\mathbf{X}} \left( \bigcap_{j=1}^p A_j \right) \right| & \le \Pr_{\mathbf{X}} \left( \bigcap_{j=1}^p A_{n,j}  \cap \left(\bigcap_{j=1}^p A_j\right)^c \right) + \Pr_{\mathbf{X}} \left( \left(\bigcap_{j=1}^p A_{n,j} \right)^c \cap \bigcap_{j=1}^p A_j \right) \\
& = \Pr_{\mathbf{X}} \left( \bigcap_{j=1}^p A_{n,j}  \cap \bigcup_{j=1}^p A_j^c \right) + \Pr_{\mathbf{X}} \left( \bigcup_{j=1}^p A_{n,j}^c \cap \bigcap_{j=1}^p A_j \right)
\end{align*}
Since $\lim\sup_{n \rightarrow} \bigcap_{j=1}^p A_{n,j} = \bigcap_{j=1}^p A_j$, for the first term we have,
\begin{align*}
\lim\sup_{n \rightarrow} \Pr_{\mathbf{X}} \left( \bigcap_{j=1}^p A_{n,j}  \cap \bigcup_{j=1}^p A_j^c \right) & \le \Pr_{\mathbf{X}} \left( \lim\sup_{n \rightarrow} \left( \bigcap_{j=1}^p A_{n,j}  \cap \bigcup_{j=1}^p A_j^c \right) \right) \\ 
& = \Pr_{\mathbf{X}} \left( \lim\sup_{n \rightarrow} \left( \bigcap_{j=1}^p A_{n,j} \right) \cap \bigcup_{j=1}^p A_j^c \right) \\
& = \Pr_{\mathbf{X}} \left(  \bigcap_{j=1}^p A_j \cap \bigcup_{j=1}^p A_j^c \right) \\
& = 0
\end{align*}  
since $\bigcap_{j=1}^p A_j$ and $\bigcup_{j=1}^p A_j^c$ are disjoint sets.
\end{proof}

\begin{proof}[Proof of Proposition \ref{PropMaximizerFinite}]
For the first statement we have that
\begin{equation*}
\sup_{\mathbf{x} \in \mathbb{R}^p} \left| \min \left(\Pr_{X_n}(\bigcap_{j=1}^p A_j), \Pr_{X_n}(\bigcap_{j=1}^p B_j) \right) - \min \left(\Pr_{\mathbf{X}}(\bigcap_{j=1}^p A_j), \Pr_{\mathbf{X}}(\bigcap_{j=1}^p B_j) \right) \right|
\end{equation*}
is small than
\begin{equation}
\sup_{\mathbf{x} \in \mathbb{R}^p} \left|\Pr_{X_n}(\bigcap_{j=1}^p A_j) - \Pr_{\mathbf{X}}(\bigcap_{j=1}^p A_j) \right| + \sup_{\mathbf{x} \in \mathbb{R}^p} \left|\Pr_{X_n}(\bigcap_{j=1}^p B_j) - \Pr_{\mathbf{X}}(\bigcap_{j=1}^p B_j) \right| \label{EquGC1} 
\end{equation}
and we will show that both terms will goes to zero, uniformly almost surely. In fact, using the Inclusion-Exclusion formula and rearranging the terms
\begin{align*}
\sup_{\mathbf{x} \in \mathbb{R}^p} \left|\Pr_{X_n}(\bigcap_{j=1}^p A_j) - \Pr_{\mathbf{X}}(\bigcap_{j=1}^p A_j) \right| & \le \sup_{\mathbf{x} \in \mathbb{R}^p} \left| \hat{F}_{X_n}(\mathbf{x}) - F_\mathbf{X}(\mathbf{x}) \right| \\
& + \sum_{k=1}^p \sum_{I \subseteq \{1, \cdots, p\}, |I|=k} \sup_{\mathbf{x} \in \mathbb{R}^p} \left| \hat{F}_{X_n}((\mathbf{x} - \boldsymbol{\tau}_I)^-) - F_{\mathbf{X}}((\mathbf{x} - \boldsymbol{\tau}_I)^-) \right|
\end{align*}
Notice that the number of terms is finite and do not depend on sample size; applying Glivenko-Cantelli's theorem in $\mathbb{R}^p$, each terms on the right hand side of the expression goes to zero, uniformly almost surely. Similar result holds for the second term of \ref{EquGC1}. This proves the first statement of the proposition.

We prove already that $ld_{HR}(\cdot; \mathbf{X}, \boldsymbol{\tau})$ is an upper semicontinuous function that $\lim_{\|\mathbf{x}\| \rightarrow \infty} ld_{HR}(\mathbf{x}; \mathbf{X}, \boldsymbol{\tau}) = 0$ and if $ld_{HR}(\cdot; \mathbf{X}, \boldsymbol{\tau})$ is uniquely maximized at $\hat{\mathbf{x}}$, we have that for every $\varepsilon > 0$,  $\delta = ld_{HR}(\hat{\mathbf{x}}; \mathbf{X}, \boldsymbol{\tau}) - \sup_{\| \hat{\mathbf{x}} - \mathbf{x} \| \ge \varepsilon} ld_{HR}(\mathbf{x}; \mathbf{X}, \boldsymbol{\tau}) > 0$. To prove that $\hat{\mathbf{x}}_n \stackrel{a.s.}{\rightarrow} \hat{\mathbf{x}}$, it is sufficient to establish that
\begin{equation*}
\Pr\left[ \sup_{n \ge l} \| \hat{\mathbf{x}}_n - \hat{\mathbf{x}}  \| > \varepsilon \right] \rightarrow 0 \ , \qquad l \rightarrow \infty \ .
\end{equation*}

\begin{align*}
\Pr\left[ \sup_{n \ge l} \| \hat{\mathbf{x}}_n - \hat{\mathbf{x}}  \| > \varepsilon \right] & \le \Pr\left[ \sup_{n \ge l} \left( ld_{HR}(\hat{\mathbf{x}}; \mathbf{X}, \boldsymbol{\tau}) - ld_{HR}(\hat{\mathbf{x}}_n; \mathbf{X}, \boldsymbol{\tau})\right) > \delta \right] \\
& \le \Pr \left[ \sup_{n \ge l} \left( ld_{HR}(\hat{\mathbf{x}}; \mathbf{X}, \boldsymbol{\tau}) - ld_{HR}(\hat{\mathbf{x}}; X_n, \boldsymbol{\tau})\right) > \delta/2 \right. \\
& + \left. \sup_{n \ge l} \left( ld_{HR}(\hat{\mathbf{x}}_n; X_n, \boldsymbol{\tau}) - ld_{HR}(\hat{\mathbf{x}}_n; \mathbf{X}, \boldsymbol{\tau})\right) > \delta/2 \right] \\
& \le \Pr \left[ \sup_{n \ge l} \left( ld_{HR}(\hat{\mathbf{x}}; \mathbf{X}, \boldsymbol{\tau}) - ld_{HR}(\hat{\mathbf{x}}; X_n, \boldsymbol{\tau})\right) > \delta/2 \right] \\
& + \Pr\left[ \sup_{n \ge l} \left( ld_{HR}(\hat{\mathbf{x}}_n; X_n, \boldsymbol{\tau}) - ld_{HR}(\hat{\mathbf{x}}_n; \mathbf{X}, \boldsymbol{\tau})\right) > \delta/2 \right] \\
& \le \Pr \left[ \sup_{n \ge l} \sup_{\mathbf{x} \in \mathbb{R}^p} \left| ld_{HR}(\hat{\mathbf{x}}; \mathbf{X}, \boldsymbol{\tau}) - ld_{HR}(\hat{\mathbf{x}}; X_n, \boldsymbol{\tau}) \right| > \delta/2 \right] \\
& + \Pr\left[ \sup_{n \ge l} \sup_{\mathbf{x} \in \mathbb{R}^p} \left| ld_{HR}(\hat{\mathbf{x}}_n; X_n, \boldsymbol{\tau}) - ld_{HR}(\hat{\mathbf{x}}_n; \mathbf{X}, \boldsymbol{\tau})\right| > \delta/2 \right] \\
& \le 2 \Pr \left[ \sup_{n \ge l} \sup_{\mathbf{x} \in \mathbb{R}^p} \left| ld_{HR}(\hat{\mathbf{x}}; \mathbf{X}, \boldsymbol{\tau}) - ld_{HR}(\hat{\mathbf{x}}; X_n, \boldsymbol{\tau}) \right| > \delta/2 \right] \stackrel{l \rightarrow \infty}{\rightarrow} 0  \ .
\end{align*}
which leads to the second statement of the proposition.
\end{proof}

\begin{proof}[Proof of Proposition \ref{PropLocalHalfRegionTauInfiniteDimension}]
\begin{enumerate}
\item[i.] By the definition of $\hyp(y; \tau)$ we have that $\hyp(y; \tau_1) \subset \hyp(y; \tau_2)$  and hence $P(G(Y) \subset \hyp(y; \tau_1)) \leq P(G(Y) \subset \hyp(y; \tau_2))$. Similar results holds for the epigraph $\epi(y; \tau)$ which leads to the statement.
\item[ii.] Let $\bar{0}$ be the function always equals to $0$ and $\bar{\infty}$ be the function always equals to $\infty$. Then, by definition of $\hyp(y; \tau)$ and $\epi(y; \tau)$ we have $G(y) = \hyp(y; \bar{0}) \subseteq \hyp(y; \tau) \subseteq \hyp(y; \bar{\infty}) = \hyp(y)$ and $G(y) = \epi(y; \bar{0}) \subseteq \epi(y; \tau) \subseteq \epi(y; \bar{\infty}) = \epi(y)$. Furthermore, $ld_{HR}(y; Y, \bar{0}) = P(G(Y) = g(y))$ and $ld_{HR}(y; Y, \bar{\infty}) = d_{HR}$ together with the result at point i. lead to the inequalities.
\item[iii.] We notice that, if $\min_{t \in I} \tau(t) \rightarrow \infty$ then $\tau \rightarrow \bar{\infty}$.
\item[iv.] We notice that, if $\| \tau \|_\infty \rightarrow 0^{+}$ then  $\tau \rightarrow \bar{0}$.
\end{enumerate}
\end{proof}

\begin{proof}[Proof of Proposition \ref{PropLocalHalfRegionZero}]
The first result holds since \ref{PropLocalHalfRegionTauInfiniteDimension} implies $ld_{HR}(y; Y, \tau) \le d_{HR}(y; Y)$ and Proposition 5 of \citet{lopez-pintado_half-region_2011}. Similar motivation leads to the second statement.
\end{proof}  
  
\begin{proof}[Proof of Proposition \ref{PropSemiContinuousFunctional}]
To prove that $ld_{HR}(\cdot; Y, \tau)$ is upper semicontinuous we show that $\lim\sup_{n \rightarrow \infty} ld_{HR}(y_n; Y, \tau) \leq ld_{HR}(y; Y, \tau)$, where $y_n \stackrel{\|\cdot\|_\infty}{\longrightarrow} y$.
\begin{align*}
\lim\sup_{n \rightarrow \infty} ld_{HR}(y_n; Y, \tau) & = \lim\sup_{n \rightarrow \infty} \min(P(G(Y) \subseteq \hyp(y_n; \tau)), P(G(Y) \subseteq \epi(y_n; \tau))) \\
& \leq \min( \lim\sup_{n \rightarrow \infty} P(G(Y) \subseteq \hyp(y_n; \tau)), \lim\sup_{n \rightarrow \infty} P(G(Y) \subseteq \epi(y_n; \tau))) \\
& \leq \min(P(G(Y) \subseteq \hyp(y; \tau)), P(G(Y) \subseteq \epi(y; \tau))) \\
& = ld_{HR}(y; Y, \tau) \ .
\end{align*}
To establish the continuity of the functional $ld_{HR}(\cdot; Y, \tau)$ in $C(I)$ with respect to the supremum norm it is sufficient to prove that both $P(G(Y) \in \hyp(\cdot; \tau))$ and $P(G(Y) \in \epi(\cdot; \tau))$ are continuous. We prove the continuity of the first term, the prove is similar for the second term. We have to see that if $y_n \stackrel{\|\cdot\|_\infty}{\longrightarrow} y$ then $A_n = |P(G(Y) \subseteq \hyp(y_n; \tau)) - P(G(Y) \subseteq \hyp(y; \tau))| \stackrel{n\rightarrow\infty}{\longrightarrow} 0$. Recall that,
\begin{align*}
A_n & = |P(G(Y) \in \hyp(y_n; \tau)) - P(G(Y) \subseteq \hyp(y; \tau))| \\
& \leq P(G(Y) \subseteq \hyp(y_n; \tau) \cup G(Y) \nsubseteq \hyp(y; \tau)) \\
& + P(G(Y) \nsubseteq \hyp(y_n; \tau) \cup G(Y) \subseteq \hyp(y; \tau)) \ .
\end{align*}
Considering that the marginals of the distribution $P$ are continuous it is easy to prove that
\begin{align*}
P(G(Y) \subseteq \hyp(y_n; \tau) \cup G(Y) \nsubseteq \hyp(y; \tau)) & \stackrel{n\rightarrow\infty}{\longrightarrow} 0 \\
\text{and} & \\
P(G(Y) \nsubseteq \hyp(y_n; \tau) \cup G(Y) \subseteq \hyp(y; \tau)) & \stackrel{n\rightarrow\infty}{\longrightarrow} 0 \ .
\end{align*}
Hence $P(G(Y) \in \hyp(\cdot; \tau))$ is a continuous function.
\end{proof}

\begin{proof}[Proof of Theorem \ref{TeoConsistency}]
The random variable $ld_{HR}(\cdot; \mathbf{Y}_n, \tau)$ corresponding to the local half-region depth can be expressed as
\begin{align*}
ld_{HR}(y; \mathbf{Y}_n, \tau) & = \min \left( \frac{1}{n} \sum_{i=1}^n \ind(G(Y_i) \subset \hyp(y; \tau)), \frac{1}{n} \sum_{i=1}^n \ind(G(Y_i) \subset \epi(y; \tau)) \right) \\
& \text{by the law of large numbers and the continuity of the minimum,} \\
& \stackrel{a.s.}{\longrightarrow} \min ( P(G(Y) \subset \hyp(y; \tau)), P(G(Y) \subset \epi(y; \tau)) ) \\
& = ld_{HR}(y; Y, \tau) \ ,
\end{align*}  
and then the result.
\end{proof}

\begin{proof}[Proof of Theorem \ref{TeoEquiContinuousFunctional}]
Without loss of generality we assume that $I = [0,1)$, $P_n$ is the empirical probability of $n$ random curves $\mathbf{Y}_n = (Y_1, \cdots, Y_n)$ sampled from the stochastic process $Y$ with probability $P$, and given a measurable function $\psi$ we write $P_n\psi$ for the mean process $\sum_{i=1}^n \psi(Y_i)/n$ and $P\psi$ for the expectation of $\psi$ with respect to $P$, i.e., $\int \psi(Y) \ dP$.

We define two families of measurable functions $\mathcal{F}_+(E) = \{ \psi_+(z, y) : z, y \in E \rightarrow \{ 0, 1 \} \}$, and, similarly, $\mathcal{F}_-(E)$ where
\begin{align*}
\psi_+(z, y) & = \ind( y(t) \leq z(t), t \in I) \ind( y(t) - \tau(t) \geq z(t), t \in I) \\
\psi_-(z, y) & = \ind( y(t) \geq z(t), t \in I) \ind( y(t) + \tau(t) \leq z(t), t \in I)
\end{align*}
We have,
\begin{align*}
\sup_{y \in E} |ld_{HR}(y; \mathbf{Y}_n, \tau) - ld_{HR}(y; Y, \tau)| & \leq \sup_{y \in \underline{E}} |ld_{HR}(y; \mathbf{Y}_n, \tau) - ld_{HR}(y; Y, \tau)| \\
& + \sup_{y \in \overline{E}} |ld_{HR}(y; \mathbf{Y}_n, \tau) - ld_{HR}(y; Y, \tau)|
\end{align*}
where $\underline{E} = \{ y \in E : \| y \|_\infty \leq M \}$ and $\overline{E} = \{ y \in E : \| y \|_\infty > M \}$ for some positive constant $M$. The second term converges almost surely to zero when $M$ tends to infinity by Proposition \ref{PropLocalHalfRegionZero}, hence we have to establish that the first term goes to zero almost surely for a given $M$ sufficiently large. At this aim we note that,
\begin{equation*}
\sup_{y \in \underline{E}} |ld_{HR}(y; \mathbf{Y}_n, \tau) - ld_{HR}(y; Y, \tau)| \leq \sup_{\psi \in \mathcal{F}_+(\underline{E})} | P_n\psi - P\psi| + \sup_{\psi \in \mathcal{F}_-(\underline{E})} | P_n\psi - P\psi| \ .
\end{equation*}
Hence, it is sufficient to prove that the two families $\mathcal{F}_+(\underline{E})$ and $\mathcal{F}_-(\underline{E})$ are P-Glivenko Cantelli \citet[section 2.8.8]{kosorok_2008}. Here we show the result for $\mathcal{F}_+(\underline{E})$, similar prove can be obtain for $\mathcal{F}_-(\underline{E})$. We prove that the bracketing number $N_{[]}(\varepsilon, \mathcal{F}_+(\underline{E}), L_1(P))$ is finite, then Theorem 2.2 of \citet{kosorok_2008} leads to the result. Given $\varepsilon > 0$, we need to construct a finite number of functions $z_1, \cdots, z_p$ that determine the $\varepsilon$-brackets $[\psi_+(z_i), \psi_+(z_j)]$, covering the family $\mathcal{F}_+(\underline{E})$ and such that $P(\psi_+(z_j, Y) - \psi_+(z_i, Y)) < \varepsilon$. Hence, let fix an $\varepsilon$, then there exists $\nu$ such that by (i) we have
\begin{equation*}
P(\psi_+(z_j, Y) - \psi_+(z_i, Y)) \leq P(\ind(z_j \leq Y) - \ind(z_i \leq Y)) =  \Pr(z_j \leq Y \leq z_i) < \varepsilon \ .
\end{equation*}
when $\|z_i - z_j\| < \nu$. By the equicontinuity of $\underline{E}$, given $\nu > 0$, we let $\underline{\tau} > 2 \nu$ and there exists $\delta > 0$, such that if $| s - t | < \delta$ then $|y(s) - y(t) < \nu$ for every $y \in \underline{E}$. Consider the set of functions defined as constants in the intervals $[0, \delta)$, $[\delta, 2\delta)$, $[2\delta, 3\delta)$, $\cdots$, $[((1/\delta) -1)\delta, 1)$, and taking values in the sequence: $-[M/\nu]\nu$, $\cdots$, $-\nu$, $0$, $\nu$, $2\nu$, $\cdots$, $[M/\nu]\nu$. The total number $p$ of possible functions that can be constructed like this is finite and we denote them as $z_1, \cdots, z_p$. Define the following set of indicator functions $\{ \psi_+(z_k, y), \ k \in \{1, \cdots, p \}\}$. This set of functions allows us to construct $\varepsilon$-brackets that cover the family $\mathcal{F}_+(\underline{E})$. Since $p$ is finite this implies that the bracketing number $N_{[]}(\varepsilon, \mathcal{F}_+(\underline{E}), L_1(P))$ is finite.       
\end{proof}

\begin{proof}[Proof of Theorem \ref{TeoMaximizerFunctional}]
We have to show that, for every $\varepsilon > 0$,
\begin{equation*}
\Pr(\sup_{n \geq l} \| \hat{y}_n - \hat{y} \|_\infty \geq \varepsilon) \stackrel{l \rightarrow \infty}{\longrightarrow 0} \ .
\end{equation*}
The space $(E, \| \cdot \|_\infty)$ is a metric space and $ld_{HR}(\cdot, Y, \tau)$ is upper continuous on $E$ and satisfies
\begin{equation*}
\sup_{\| y \|_\infty \ge M} ld_{HR}(y; Y, \tau) \stackrel{M \rightarrow \infty}{\longrightarrow} 0 \ .
\end{equation*}
Then the rest of the proof is analogous to the one in Proposition \ref{PropMaximizerFinite} and it is omitted.
\end{proof}

\clearpage

%%%\bibliography{lhrd-arxiv}

\end{document}